\documentclass[10pt]{article}
\usepackage[nofonts]{dgleich-common}
\pdfoutput=1
\usepackage[T1]{fontenc}
\RequirePackage[scaled=0.8]{helvet}%
\RequirePackage[scaled=0.75]{beramono}%
\def\lstbasicfont{\fontfamily{fvm}\selectfont}
\makeatletter  
\setboolean{@dgleich@sanssmallcaps}{true}
\newcommand*{\sublabel}{%
\let\subcaption@ORI@label=\@tufte@orig@label%
\subcaption@label}

\makeatother

\usepackage{dgleich-math}

\graphicspath{{./Images/}}
\usenatbib
\usehyperref

\theoremstyle{nonumberplain}

	\theoremstyle{nonumberplain}
	\theoremheaderfont{\normalfont\small\sffamily}
	\theorembodyfont{\normalfont\upshape}
 	\theoremsymbol{\ensuremath{\blacksquare}}
	\newtheorem{proof}{PROOF}
	\newtheorem{justification}{Justification}

\begin{document}

\marginnote[48\baselineskip]{%
\noindent%
\emph{All authors}\\ \noindent 
Purdue University Computer Science
\texttt{\{dshur,huan1754,dgleich\}@purdue.edu}
\bigskip 

\noindent Both Shur and Huang contributed equally. Research supported in part by
NSF CCF-1909528, IIS-2007481, NSF Center for Science of Information STC, CCF-0939370.
}

\title{A flexible PageRank-based graph embedding framework closely related to spectral eigenvector embeddings}

\author{Disha Shur $\cdot$ Yufan Huang $\cdot$ David F. Gleich}

\maketitle

\begin{abstract}
We study a simple embedding technique based on a matrix of personalized PageRank vectors seeded on a random set of nodes. We show that the embedding produced by the element-wise logarithm of this matrix (1) are related to the spectral embedding for a class of graphs where spectral embeddings are significant, and hence useful representation of the data, (2) can be done for the entire network or a smaller part of it, which enables precise local representation, and (3) uses a relatively small number of PageRank vectors compared to the size of the networks. Most importantly, the general nature of this embedding strategy opens up many emerging applications, where eigenvector and spectral techniques may not be well established, to the PageRank-based relatives. For instance, similar techniques can be used on PageRank vectors from hypergraphs to get ``spectral-like'' embeddings.
\end{abstract}

\section{Introduction}\label{sec1}

The eigenvectors of the graph Laplacian are among the most widely used algorithmic measures of a graph.
They are used to find cuts and clusters in a variety of settings~\cite{MalikShi,Chung-1992-book,Pothen-1990-partitioning}. 
They give a signal basis for a graph~\cite{hammond,gwav}. 
And one of their original uses was to draw informative pictures of graphs in a low dimensional space~\cite{hall,kohen}.  These are all related to the idea of embedding the 
graph into a low dimensional space and recent uses have closely studied this embedding 
framework. 

Likewise, PageRank is itself a widely used algorithmic measure on a graph~\cite{origpr}. 
The uses are extremely diverse~\cite{ppr_basics}. Relationships between 
PageRank and spectral clustering are also known~\cite{andersen2006,mahoney2012,Gleich2014}.
These exist because both techniques can be related to random walks, and seeded PageRank is a 
localized type of random walk, or random walk with restart~\cite{tong}. 

In this manuscript, we study a particular type of relationship 
between a matrix of seeded PageRank vectors and the eigenvectors of
the Laplacian matrix. Our log-PageRank embedding uses the singular
vectors of the elementwise log of a random collection of seeded PageRank vectors. 
An example is in Figure~\ref{fig:intro-example}. This shows that log-PageRank 
embeddings resemble spectral clustering.


\begin{fullwidthfigure}[t]
    \centering
    \begin{subfigure}{0.49\linewidth}
        \centering
        \includegraphics[width=\linewidth]{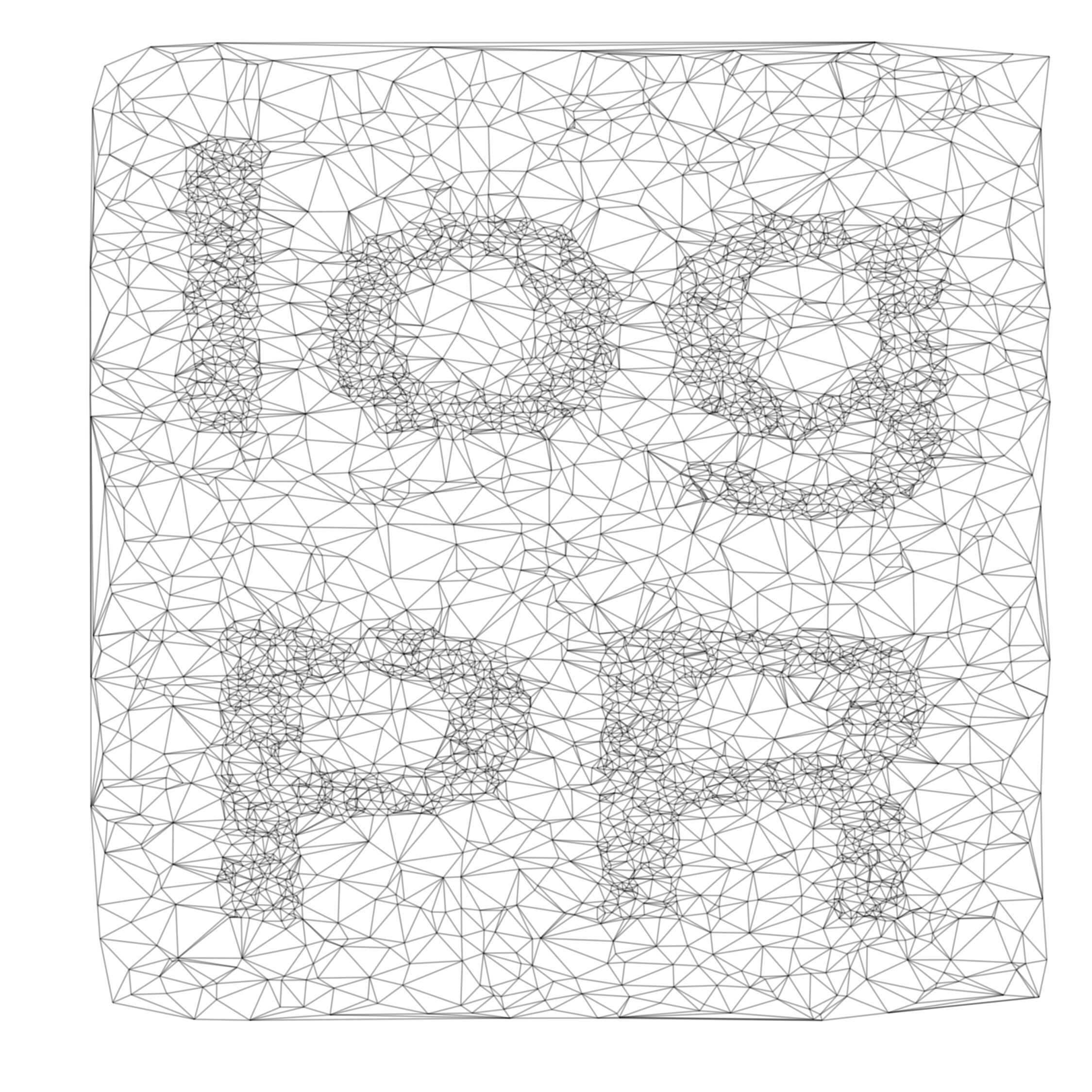}
        \caption{planar graph with 5000 nodes and 14962 edges \emph{\&} the words ``log PR''}
        \label{fig:1a}
    \end{subfigure}
    \hfill
    \begin{subfigure}{0.49\linewidth}
        \centering
        \includegraphics[width=\linewidth]{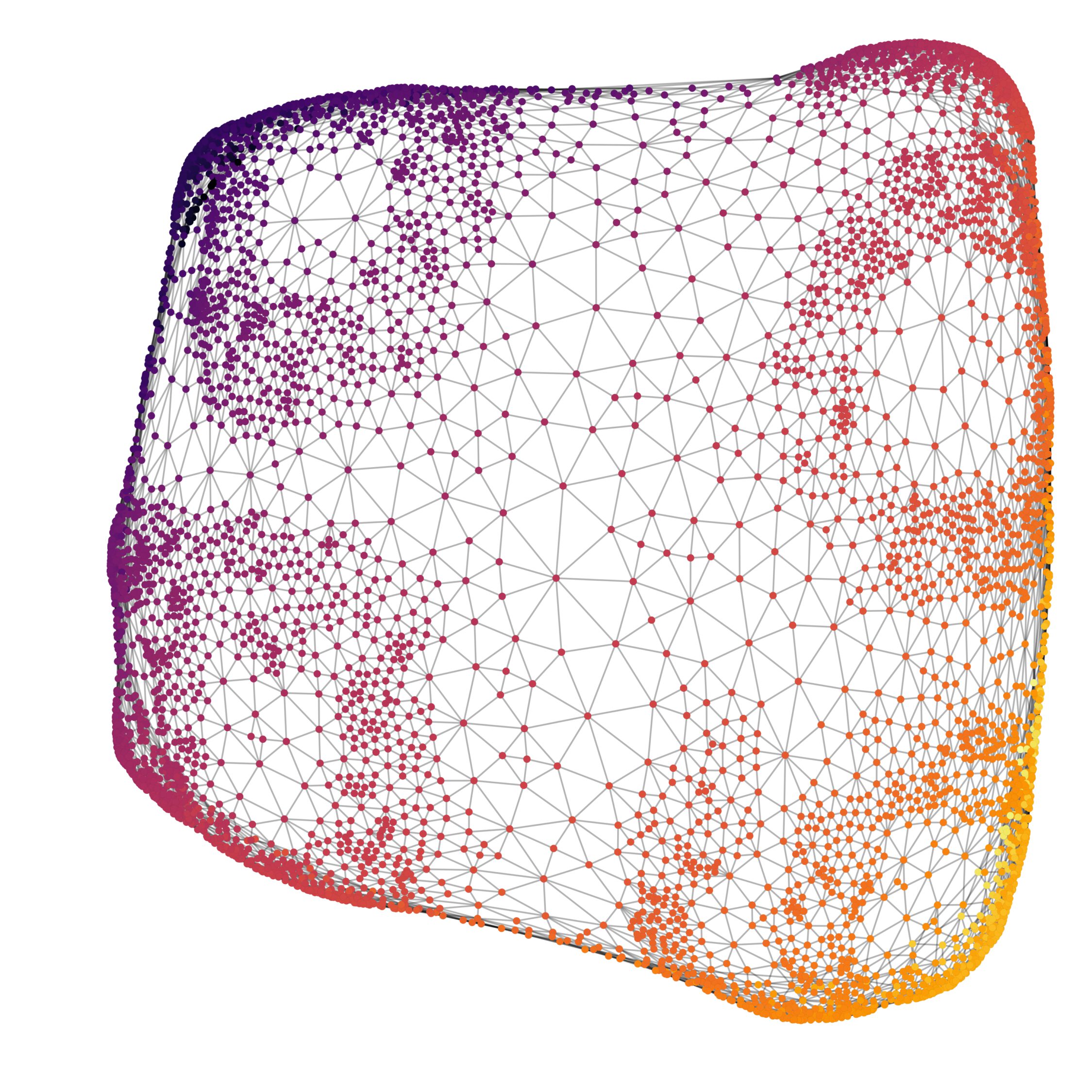}
        \caption{spectral embedding of the graph from the Laplacian eigenvectors }
        \label{fig:1b}
    \end{subfigure}

\begin{subfigure}{0.33\linewidth}
        \centering
        \includegraphics[width=\linewidth]{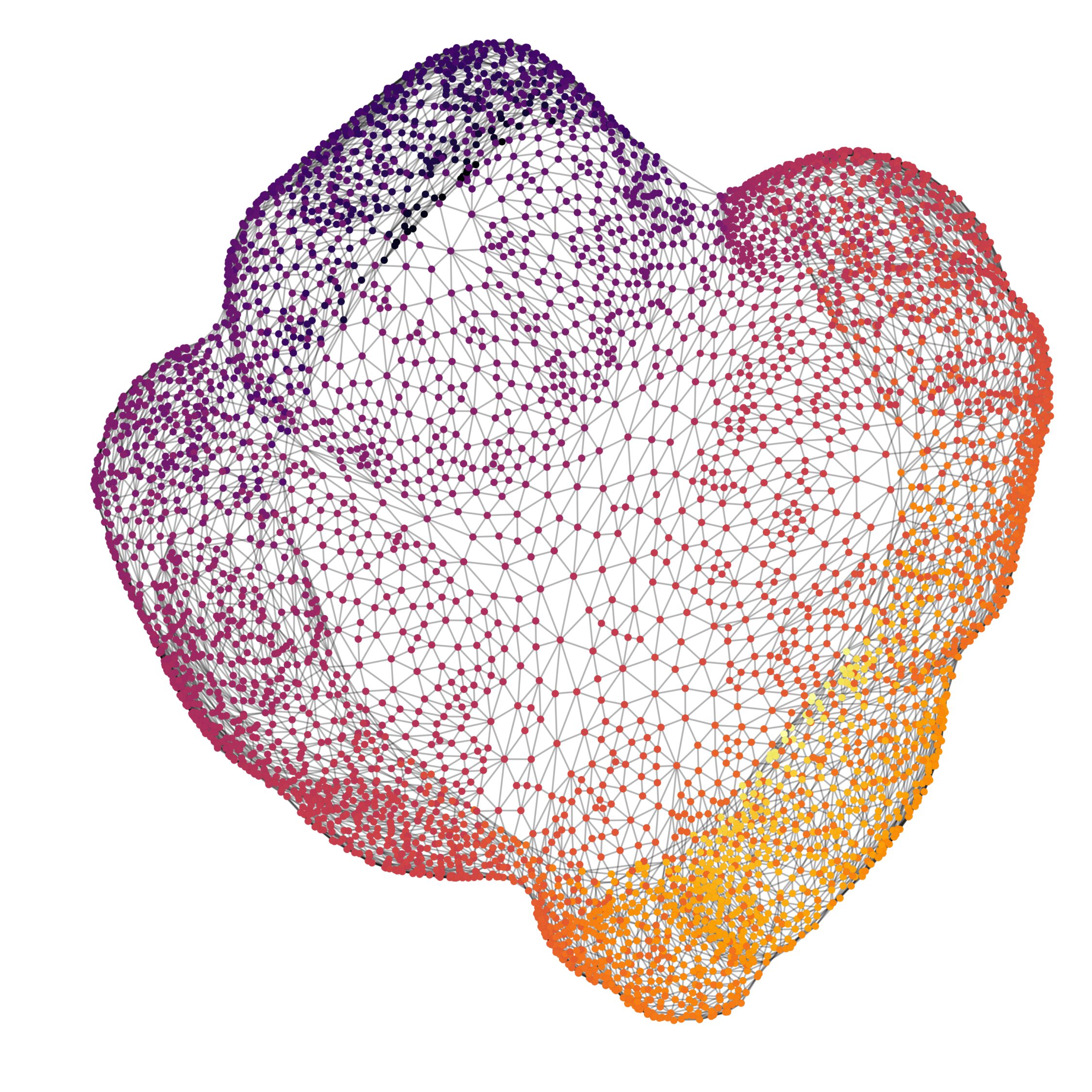}
        \parbox{0.95\linewidth}{\caption{log-PageRank embedding for $\alpha=0.85$}}
        \label{fig:1c}
    \end{subfigure}%
    \hfill%
    \begin{subfigure}{0.33\linewidth}
        \centering
        \includegraphics[width=\linewidth]{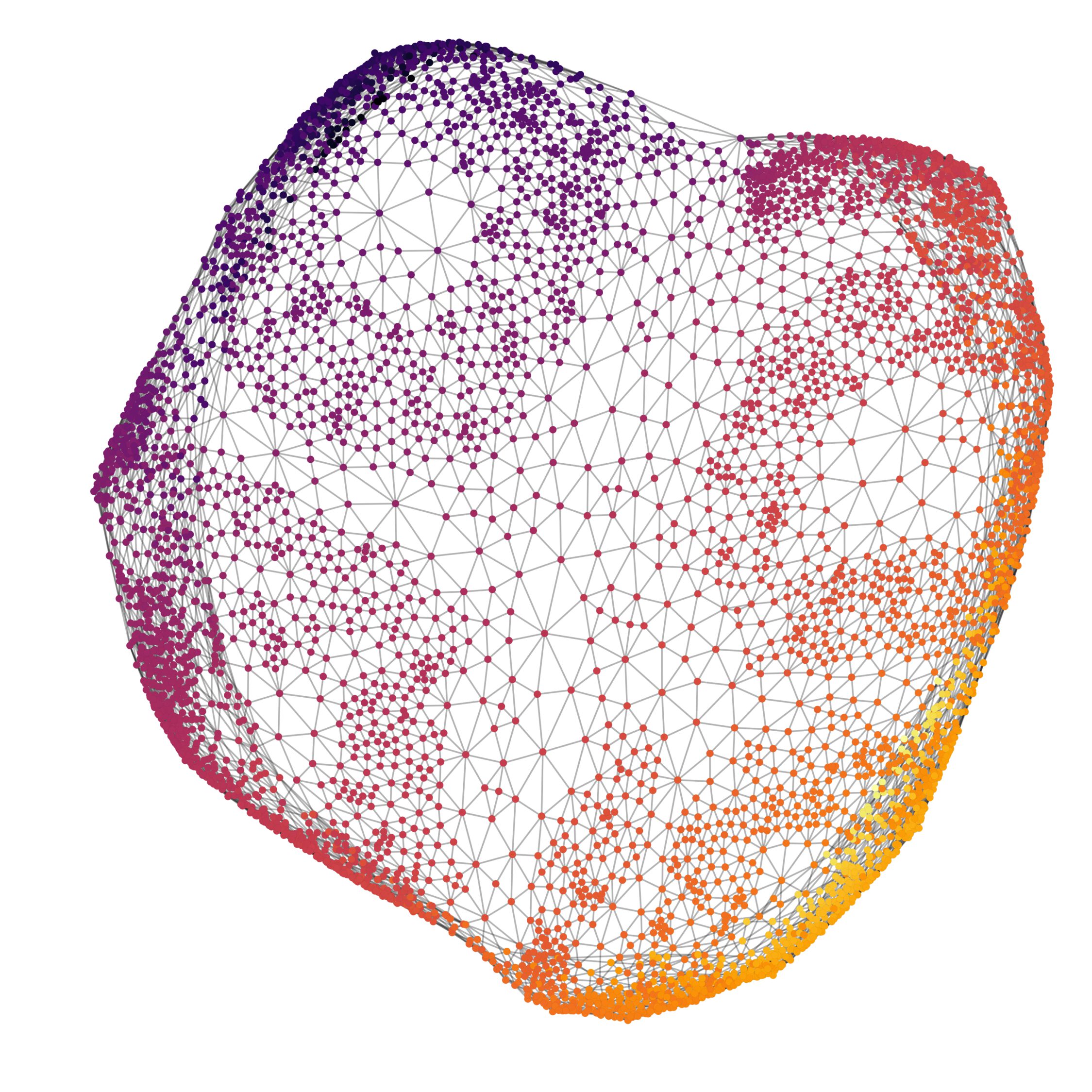}
        \parbox{0.95\linewidth}{\caption{log-PageRank embedding for $\alpha=0.99$}}
        \label{fig:1d}
    \end{subfigure}%
    \hfill%
    \begin{subfigure}{0.33\linewidth}
        \centering
        \includegraphics[width=\linewidth]{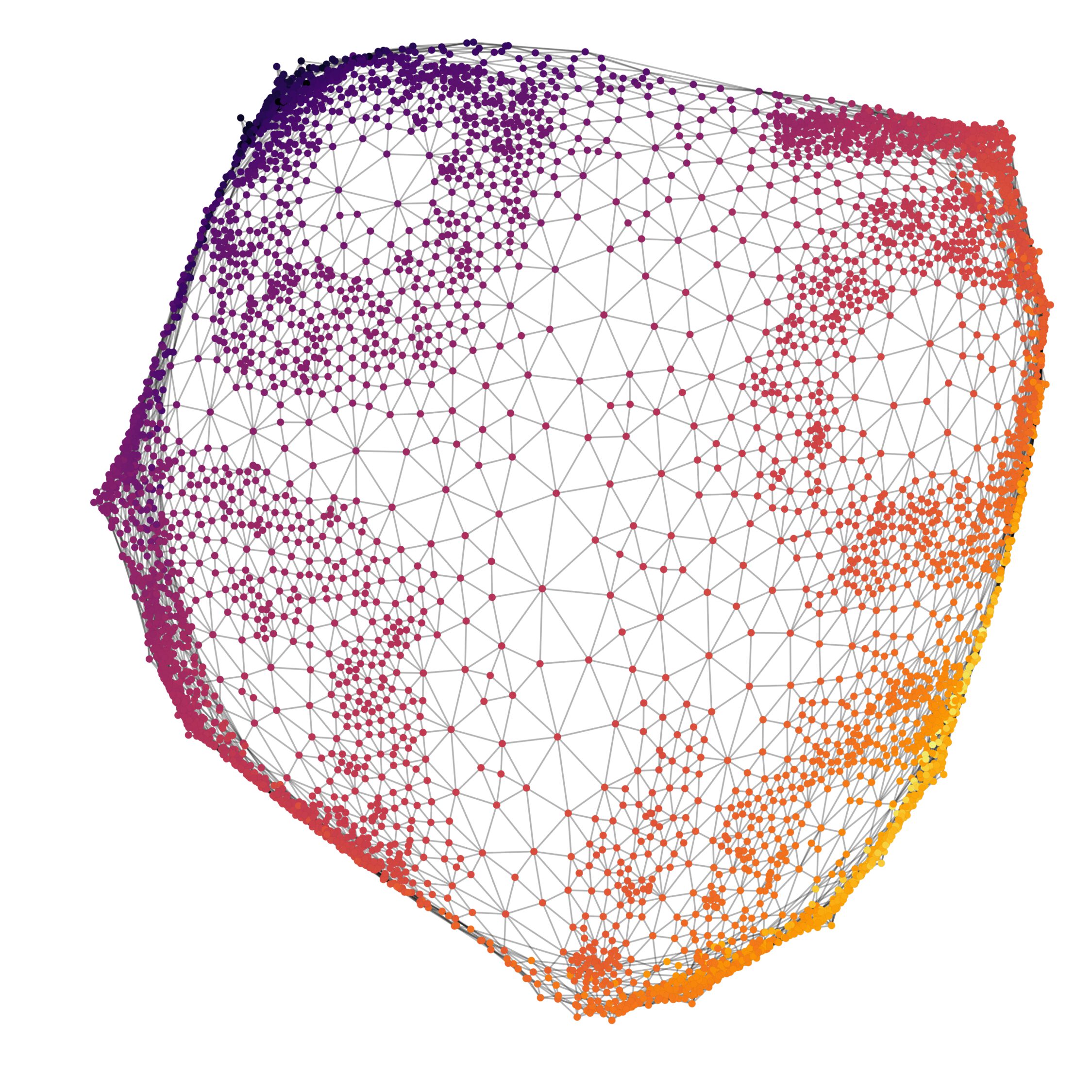}
        \parbox{0.95\linewidth}{\caption{log-PageRank embedding for $\alpha=0.999$.}\phantomcaption\label{fig:1e}}
    \end{subfigure}
\caption{The embedding pictures have all the nodes colored with the same values to show relative position. The log-PageRank embedding uses singular vectors of the element-wise logarithm of seeded PageRank vectors. 
Our paper argues that the similarity between (b), (d) and (e) is expected through an
approximation analysis. The advantage the log-PageRank embeddings is that they can be 
deployed in many emerging data scenarios where spectral embeddings and 
eigenvectors are not as well established or may be computationally expensive 
but where analogues of random walks 
or PageRank may be possible, as in hypergraphs (see Figure~\ref{fig:yelp}). }
\label{fig:intro-example}
\end{fullwidthfigure}

Our manuscript shows that this relationship is expected for degree-regular graphs through
an approximation argument (Section~\ref{sec:lpr+se}). This builds from a study on chain graphs that closely characterizes the log-PageRank values (Section~\ref{sec:chain}). 

PageRank vectors have long been viewed in terms of log-scaling. When Google
published PageRank scores for websites, they were understood to represent an approximation
of the log of Google's internal metrics~\cite{Bar-Yossef2008-reverse-PageRank}. 
When PageRank was used in spam analysis,
log scaling was used by
\Citet[Section 6.4]{becchetti2008-spam}. So a 
log-scaled analysis is not surprising. 

That said, analyzing the singular vectors of personalized PageRank vectors under log-scaling presents interesting challenges from a technical perspective. As such, we are only able to get an approximate, although compelling, understanding of the relationship illustrated in the figure. 

These log-PageRank embeddings offer a different set of computational tradeoffs compared with
eigenvectors. First, they only require random samples of a diffusion process on the graph.
Indeed, a closely related methodology to these log-PageRank embeddings was previously
used in~\citet{beginng} to compare spectral clustering with alternatives.
This shows how the ideas behind the log-PageRank embeddings give more flexible
structures to help users study their datasets. For instance, it is easy to study a variety
of localized log-PageRank embeddings that only work in a subset of a graph. Yet, these
can also be designed to pull in other nearby regions as suggested by the PageRank vectors
instead of more brittle Dirichlet eigenvector approximations~\cite{dev}. 
In this paper, we briefly explore using log-PageRank embeddings on hypergraphs to visualize their structure as well.

Moreover, the idea of customizing embeddings is highly relevant to the ongoing use
of graph embeddings for ML algorithms. 
Embedding development is the primary task in problems pertaining to network analysis, language processing, image processing or any problem that seeks to understand and use data. Literature on embedding is filled with learning models based on functions of spectral entities \cite{hypergcn, deepwalk, node2vec, nonlin_dif_hg, unifying, combining}. Like our log-scaling, these embeddings often involve nonlinearities such as sigmoid. 

PageRank or diffusion based techniques have previously been used for learning graph embedding (or clustering) \cite{gwav, gcnpr, hne, hgpr, meng, rw1} where the personalized PageRank vector based on a set of nodes, called the \textbf{seed set} is used for further computations. Interestingly, although each methodology comes with its own set of merits, all of these methods boil down to a function of random walk on the graph developed from the available data. For example, \citet{hkrnl} and \citet{hearing2011} and \citet{gwav} and \citet{hearing2018} are based on different functions of the random walk matrix. PageRank too can be expressed as an infinite geometric sum of the random walk matrix \cite{hkrnl}.

In summary, the contributions and remainder of this paper discuss:
\begin{itemize}
\item the log-PageRank embedding framework (Section~\ref{sec:logpr})
\item a study of log-PageRank values on a chain graph that shows how log-PageRank values are related to graph distance (Section~\ref{sec:chain})
\item an approximation analysis between log-PageRank embeddings and spectral clustering on d-regular graphs (Section~\ref{sec:lpr+se})
\item a computational study of similarities and differences between spectral and log-PageRank embeddings (Section~\ref{sec:examples})
\item examples of log-PageRank embeddings in hypergraphs using hypergraph PageRank~\cite{meng} (Section~\ref{sec:hypergraph}).
\end{itemize}

\section{Preliminaries}
In this manuscript we consider a connected weighted or unweighted undirected graph $G = (V, E)$ where $V$ 
is the vertex set with $n$ vertices and $E$ is the edge set with $m$ edges. 
Let $\mA$ and $\mD$ denote the adjacency matrix and degree matrix of a graph $G$
correspondingly. The Laplacian matrix $\mL$ of a graph $G$ is  
$\mD - \mA$ and the normalized Laplacian matrix $\boldsymbol{\mathcal{L}}$
is $\mI - \mD^{-1/2} \mA \mD^{-1/2}$. Let $\mW$ denote the lazy random walk 
$\frac{\mI + \mA \mD^{-1}}{2}$.
For a column-stochastic matrix $\mP$, a stationary distribution $\vpi$ is any solution to the eigensystem $\mP \vpi = \vpi$ where $\vpi$ is non-negative and sums to 1. This is an eigenvector of
$\mP$ corresponding to the eigenvalue 1. The stationary distribution of $\mP$ is unique if 
the underlying graph is connected. 

We use subscript to index entries of a matrix
or a vector: let $\mA_i$ denote the $i$th column of matrix $\mA$, $\mA_{ij}$
denote the $(i, j)$th entry of matrix $\mA$ and $\mA_{i:j}$ denote
the matrix of columns $\mA_i, \ldots , \mA_j$; let $\vx_i$
denote the $i$th entry of vector $\vx$ and $\vx_{i:j}$ denote
the vector of entries $\vx_i, \ldots, \vx_j$.
Let $\ve_1, \ldots, \ve_n$ denote the
columns of the identity matrix and the $n$ standard basis vectors of $\mathbb{R}^n$ and $\ve$ be the all-ones vector.
We use $\log.$ to denote the element-wise log operator applied to a vector. 

\subparagraph{\textbf{PageRank}}
The classical PageRank problem is defined as follows
\begin{definition}[see for example \citet{ppr_basics}]
Let $\mP$ be a column-stochastic matrix and $\vv$ be a column-stochastic vector,
then PageRank problem is to find the solution
$\vx$ to the linear system 
\begin{align}
    \label{eq:pr_lin}
    (\mI - \alpha \mP) \vx = (1 - \alpha) \vv
\end{align}
where the solution $\vx$ is called the PageRank vector, $\alpha \in (0, 1)$ is
the teleportation parameter and $\vv$ is the teleportation distribution over 
all vertices.
\end{definition}
By the definition above and the fact that all eigenvalues of a column-stochastic matrix 
have magnitude at most 1, we have $\mI - \alpha \mP$ is non-singular and the PageRank vector
can be written as $\vx = (1 - \alpha)(\mI - \alpha \mP)^{-1} \vv$. When the teleportation
distribution $\vv$ has support size 1, the PageRank problem is also called personalized 
PageRank problem and the corresponding solution $\vx$ is personalized PageRank vector or a seeded PageRank vector. 
For convenience, let $\mX(\alpha)$ denote $(1 - \alpha)(\mI - \alpha \mP)^{-1}$, $\vx(u, \alpha)$ denote
the personalized PageRank vector seeded on vertex $u$, i.e. 
$\vx(u, \alpha) = \mX(\alpha) \ve_u = (1 - \alpha)(\mI - \alpha \mP)^{-1} \ve_u$.

\section{Log-PageRank Embedding} \label{sec:logpr}
The authors of \citet{beginng} used the linearity of PageRank and the relationship with an expectation to study spectral-like embeddings of nonlinear operators. 

This inspiration led to our study of the log-PageRank embedding detailed in Algorithm~\ref{alg:logPR}. 
It takes as input the graph $G=(V,E)$ and outputs the $k$-dimensional node embeddings.
Our technique offers freedom in the algorithm being used for calculation of PageRank vector.

We randomly sample nodes of the graph, compute personalized PageRank vectors, and then compute an elementwise log of the resulting vectors. Then we compute an SVD of the overall set of vectors.  The non-dominant vectors give us our log-PageRank embedding. Note that a personalized PageRank vector has mathematically non-negative entries for a connected graph, so computing the log is always mathematically well defined. However, numerically, some of the elements may be sufficiently close to zero to cause an algorithm to return a floating point zero. For this reason, we often replace any zero entries with a value smaller than the smallest non-zero element returned before taking the log. This only occurs for small values of $\alpha$ and tends not to happen once $\alpha$ is close enough to one. 

\subparagraph{\textbf{Parameters}}
The user chosen parameters in this technique are the dimension of embedding, $k$, the teleportation parameter, $\alpha$, and the number of samples $s$.  The dimension is entirely at a user's discretion. For the number of samples, we suggest a result using a simple coupon collector-like bound that would be common in randomized matrix computations.
For the teleportation parameter, we suggest use $\alpha > 0.9$, such as $\alpha=0.99$ or $\alpha=0.999$.
Because we use many PageRank computations with large values of $\alpha$, we find it pragmatic to compute a single sparse LU decomposition of the matrix $\mI -\alpha \mP$ to repeatedly solve systems. 
Apart from the PageRank computation, the runtime depends on the SVD of the PageRank matrix, for which any type of randomized SVD computation could be used to make it more efficient.

\subparagraph{\textbf{Intuition and Analysis}}
The idea behind the algorithm is that the matrix of samples should have substantial information from other eigenspaces beyond the dominant one and the SVD will return this information. 
Our study of this algorithm revealed that the log is essential to getting qualitatively \emph{similar} pictures such as those in Figure~\ref{fig:intro-example}.  
We show in Section \ref{sec:lpr+se} that as $\alpha$ approaches $1$, the log-PageRank embedding approximates the eigenvectors of the lazy random walk matrix $\mW$.  We illustrate a simple example that motivates a relationship between log-PageRank values and a notion of distance. 

\begin{algorithm}[t]
\caption{Log-PageRank Embedding}
\label{alg:logPR}
\hspace*{\algorithmicindent} \textbf{Input:}
Graph adjacency matrix $\mA$, 
Dimension of embedding $k$, Number of samples $s \ge k+1$ (we suggest $s = (10+k)  \log n$), Teleportation parameter $\alpha$ \\
 \hspace*{\algorithmicindent} \textbf{Output:}
 Graph embedding $\mZ \in \mathbb{R}^{n\times k}$
\begin{algorithmic}[1]
\For{$i=1 \to s $}
\State $u \gets$ \texttt{random sample of 1 to $n$}
\State $\mX_i \gets$ \texttt{ pagerank on $\mA$ with seed $u$, teleportation param $
\alpha$}
\State \Comment{We use a single sparse LU on $\mI - \alpha \mP$ to compute PageRank }
\EndFor
\State $\mY \gets \log.(\mathbf{X})$ \Comment{Apply element-wise $\log$ on $\mX$}
\State $\mU, \mSigma, \mV \gets$ SVD of $\mY$
\State $\mZ \gets \mU_{2:k+1}$
\State \Return $\mZ$ \Comment{Return left singular vectors of $\mY$}
\end{algorithmic}
\end{algorithm}

\section{Log-PageRank on the chain graph}
\label{sec:chain}
We developed a closed form expression for the personalized PageRank on
chain graph and observed a linear dependence between the element-wise log of  PageRank and 
the graph distance. For a chain graph of size $n>2$, we solve the linear system defined in equation \eqref{eq:pr_lin} with $k$ as the seed node to obtain the following closed form expression in terms of $k,n,\alpha$.
\begin{align}
\label{woapp}
   \vx_i = 
   \begin{cases}
     c f(i), i \in \{2, \ldots, k-1\} \\
     \frac{c\alpha}{2} (f(k-1) + g(k+1)), i = k \\
     c g(i), i \in \{k+1, \ldots, n-1\} \\
   \end{cases}
\end{align}
where 
\begin{align*}
     f(i) = \frac{(+)^{i-1}+(-)^{i-1}}{(+)^{k-1}+(-)^{k-1}}, 
     g(i) = \frac{(+)^{n-i}+(-)^{n-i}}{(+)^{n-k}+(-)^{n-k}}, \\
     c = \sqrt{\frac{1 - \alpha}{1 + \alpha}}, 
     (+) = \frac{1+\sqrt{1-\alpha^2}}{\alpha}, 
     (-) = \frac{1-\sqrt{1-\alpha^2}}{\alpha}.
\end{align*}
Notice that when $\alpha$ is far from 1, $(-) \approx 0$ for high powers, and 
when $\alpha$ is close to 1, $(+) \approx (-)$. In both ways we get the same 
approximation that 
\begin{align*}
   \vx_i \approx
   \begin{cases}
     c (+)^{-\lvert i - k \rvert}, i \in [n] \setminus \{k\} \\
     c \frac{\alpha}{(+)}, i = k 
   \end{cases}
   .
\end{align*}
Then the logarithm of PageRank  expression for $\vx_i$ can be written as,
\begin{align*}
    \log \vx_i \approx -\lvert k-i \rvert \log((+))+\log (\sqrt{\frac{1-\alpha}{1+\alpha}})
\end{align*}

The above formulation indicates a linear relation between the log-PageRank and the distance from the seed node. This hints at log-PageRank being a good measure of the structure of the network around the seed node. 

We quickly verify that log-PageRank resembles the notion of ``distance'' in a geometric graph. The graph is created by randomly sampling points and connecting every point to its 6 nearest neighbors. The difference between PageRank and log-PageRank in this context is illustrated in Figure \ref{fig:dist}.

\begin{fullwidthfigure}[t]
     \centering
     \begin{subfigure}[t]{0.33\linewidth}
         \centering
         \includegraphics[width=\linewidth]{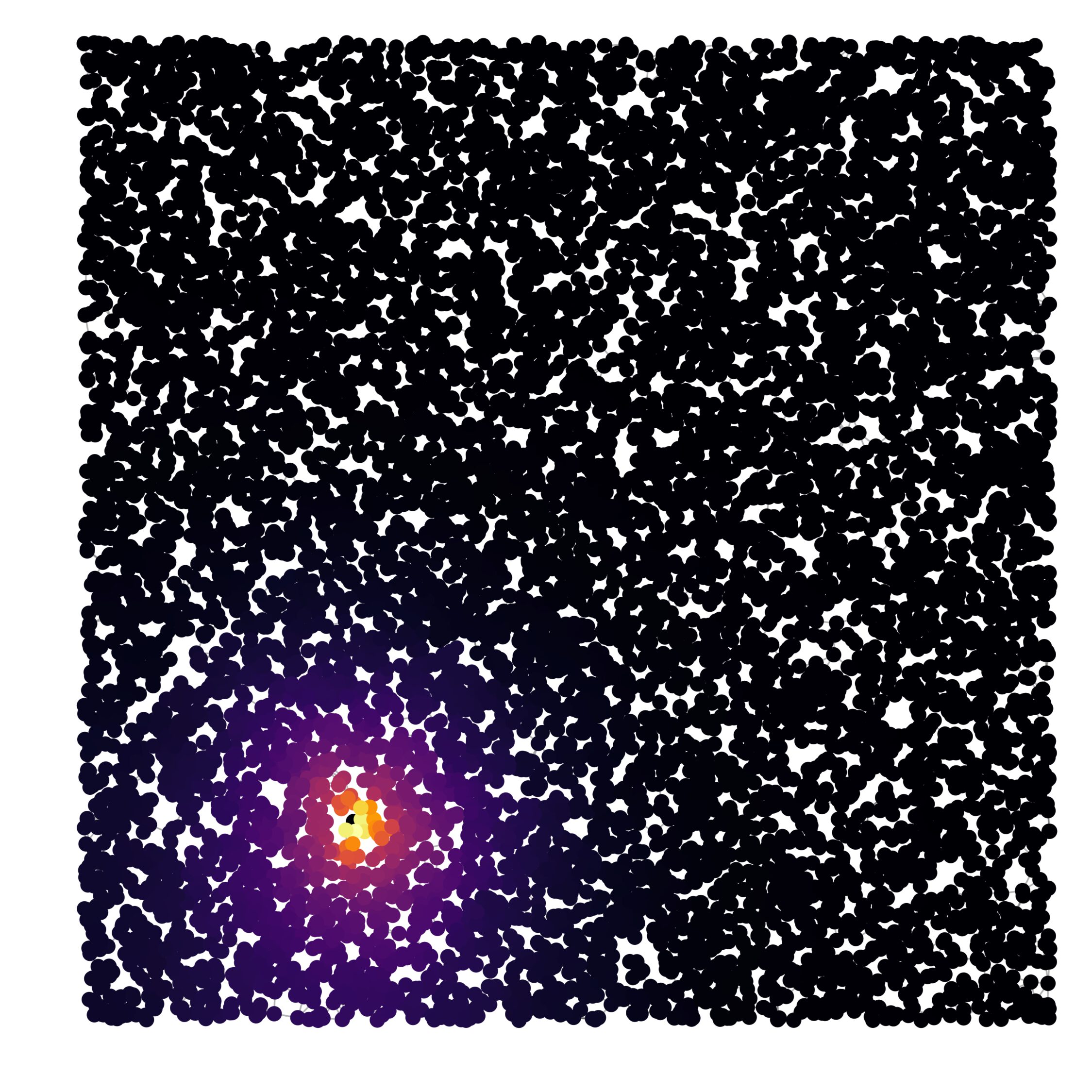}
         \parbox{0.9\linewidth}{%
         \caption{PageRank values for a 10000 node graph with 6 nearest neighbour with $\alpha = 0.999$}%
         \label{PR_distance}%
         }
     \end{subfigure}%
     \begin{subfigure}[t]{0.33\linewidth}
         \centering
         \includegraphics[width=\linewidth]{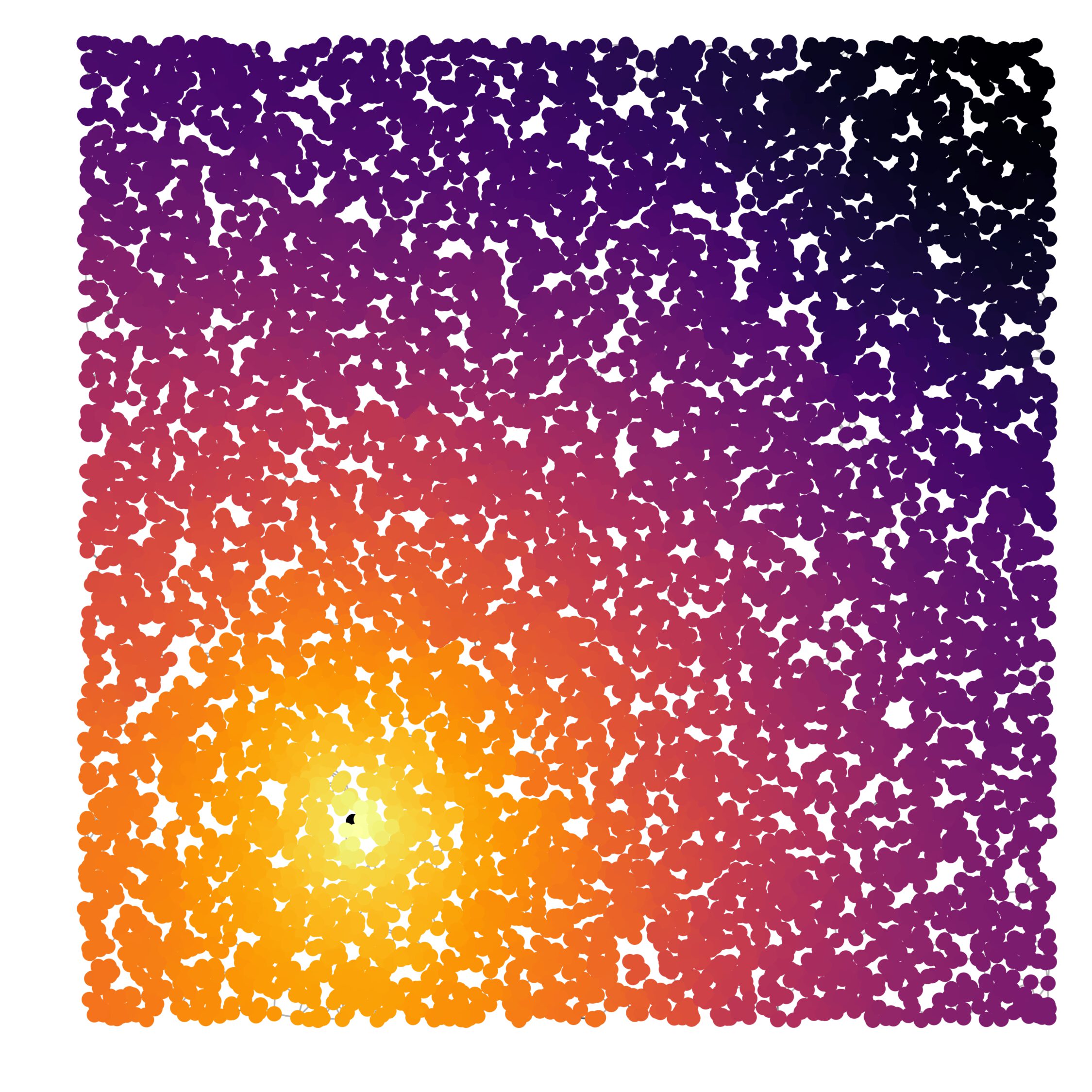}
         \parbox{0.95\linewidth}{%
         \caption{Log of PageRank values for a 10000 node graph with 6 nearest neighbour with $\alpha = 0.999$}%
         \label{logPR_distance}%
         }
     \end{subfigure}%
     \begin{subfigure}[t]{0.33\linewidth}\centering 
     \includegraphics[width=\linewidth]{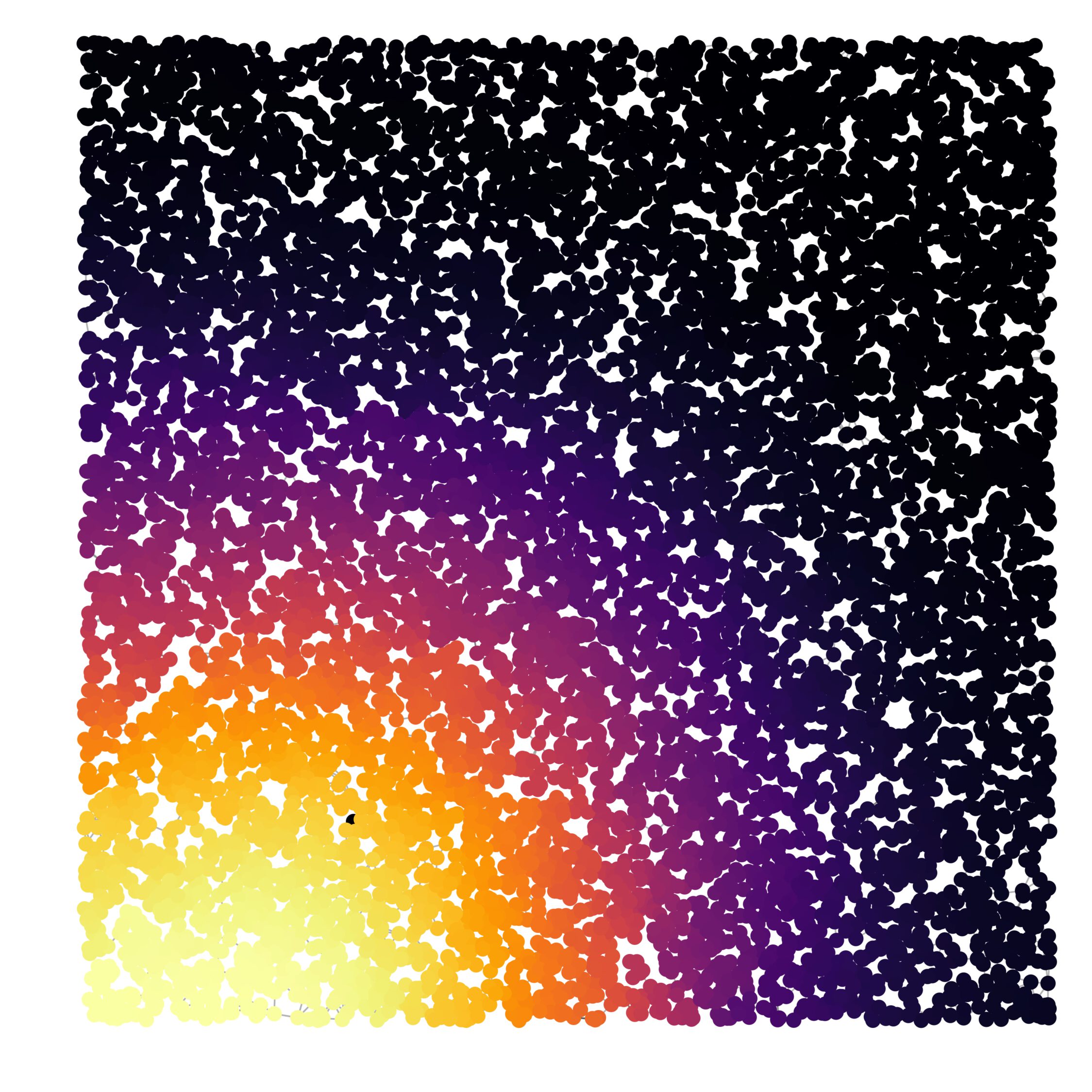}
         \parbox{0.85\linewidth}{%
         \caption{Normalized Adjacency powers for the seed used in PageRank above and $p=2000$}%
         \label{2000_matrix_power}%
         }
     \end{subfigure}
\caption{Distance effect created by log of PageRank in a geometric graph. The normalized adjacency matrix power is
 $(\mD^{-1/2}\mA \mD^{-1/2})^p\vv$ where $\vv$ is the same as the PageRank seed. Note the stronger similarity of (b) and (c). Also note the difference is at the boundary. The boundary is where we tend to see the biggest differences between log-PageRank and spectral embeddings.}
\label{fig:dist}
\end{fullwidthfigure}

\section{Relation between Log-PageRank Embedding and Spectral Embedding}
\label{sec:lpr+se}
In this section, we theoretically illustrate the relation between 
log-PageRank Embedding and Spectral Embedding on a special class of graphs,
$d$-regular graphs. 

Recall that the lazy random walk matrix is $\mW = \frac{\mI + \mA \mD^{-1}}{2}$. 
Our use of the lazy walk matrix is due to the simplicity in analyzing powers of the matrix 
because it is fundamentally aperiodic. A more intricate analysis would likely be able to remove the aperiodicity. 

Let the transition probability matrix $\mP$ of PageRank
be the lazy random walk matrix 
$\mW$. By a variety of existing analyses \cite{Serra-Capizzano-JCF,gleich-thesis}, we know that 
$\lim_{\alpha \to 1^-} \vx(u, \alpha) = 
\lim_{\alpha \to 1^-} (1-\alpha) (\mI - \alpha \mW) \ve_u  = \vpi$. 
This extends to $\log$ by continuity. Thus,  $\lim_{\alpha \to 1^-}\log.(\vx(u, \alpha)) = \log.(\vpi). $

We continue our study on $d$-regular graphs. 
Because for $d$-regular graphs, $\mA$ shares 
the same eigenvectors with $\mW$, instead of analyzing the
eigenvectors of $\mA$ as Spectral Embedding does, we
analyze the eigenvectors of $\mW$ and connect them with log-PageRank 
Embedding.

We first prove a result which characterizes the relation between 
PageRank vectors after element-wise $\log$ and columns of high power of $\mW$.
\begin{lemma}
\label{lem:log_approx}
For a connected $d$-regular graph $G$, let $\vx(u, \alpha)$ be the personalized PageRank vector
seeded on vertex $u$ with parameter $\alpha$, we have   
\begin{align*}
    \lim_{\alpha \to 1^{-}} \frac{\log.(\vx(u, \alpha))}{\|\log.(\vx(u, \alpha))\|}
    = \lim_{k \to \infty} \frac{\mW^k \ve_u}{\|\mW^k \ve_u\|}
\end{align*}
\end{lemma}

\begin{proof}
    For $d$-regular graphs, the stationary distribution $\vpi$ of $\mW$ is $\frac{\ve}{n}$. Thus
    \begin{align*}
    \lim_{\alpha \to 1^{-}} \frac{\log.(\vx(u, \alpha))}{\|\log.(\vx(u, \alpha))\|} =
    \frac{\log.(\vpi)}{\|\log.(\vpi)\|} = \frac{\ve}{\sqrt{n}}.
    \end{align*}
    Let $\mQ \mLambda \mQ^T$ be the eigenvalue decomposition of $\mW$. 
    From Perron–Frobenius theory \cite{per07, fro12}, since $G$ is connected and $\mW$ models a walk with self-loop, we have 
    that $\lambda_1 > \lvert \lambda_i \rvert $ for all $i \neq 1$, and
    \begin{align*}
        \lim_{k \to \infty} \frac{\mW^k \ve_u}{\|\mW^k \ve_u\|} = \mQ_1 = \frac{\ve}{\sqrt{n}}
    \end{align*} 
\end{proof}

Further, the approximation result below connects the eigenvectors of $\mW$ with the left singular vectors of the matrix composed
of randomly sampled columns of $\mW^k$.

\begin{approx result}
\label{lem:sample}
  For a graph $G$ and $m \leq n$, let 
  $k_1, \ldots, k_m$ be $m$ large integers and $i_1, \ldots, i_m$ be $m$
  indices uniformly sampled from $[n]$, in expectation left singular vectors of
  $\mB$ 
   equal eigenvectors of $\mW$
   where $\mB = \bmat{\mW^{k_1}\ve_{i_1} &
    \mW^{k_2} \ve_{i_2} &
    \ldots & \mW^{k_m} \ve_{i_m}}$.  
\end{approx result}
\begin{justification}
  As we know, the left singular vectors of $\mB$ are the eigenvectors of
  $\mB \mB^T$.
  Let $\mQ \mLambda \mQ^T$ be the eigenvalue decomposition of $\mW$, then we have $\mW^{k_i} = \mQ \mLambda^{k_i} \mQ^T$,
  and we have 
  \begin{align*}
    \mathbb{E}_{i_1, i_2, \ldots, i_m} \mB \mB^T 
    &=\mathbb{E}_{i_1, i_2, \ldots, i_m} \left( \sum_{j=1}^m \mQ \mLambda^{k_j} \mQ^T \ve_{i_j} \ve_{i_j}^T
    \mQ \mLambda^{k_j} \mQ^T \right)\\ 
    &=\left( \sum_{j=1}^m \mQ \mLambda^{k_j} \mQ^T \frac{\mI}{n} \frac{\mI^T}{n}\mQ \mLambda^{k_j} \mQ^T \right)\\ 
                            &= \frac{1}{n^2} \mQ \left(\sum_{j=1}^m \mLambda^{k_j} \right) \mQ^T. 
  \end{align*}
  Therefore in expectation the first $m$ eigenvectors of $\mB \mB^T$ are $\mQ_1, \ldots, \mQ_m$. Notice that when one eigenvalue of $\mW$ has multiplicity larger
  than 1, the corresponding eigenvectors of $\mB \mB^T$ may undergo one orthogonal transformation but
  the spaces they span are invariant. 
\end{justification}

The approximation result below states that the low-rank log-PageRank Embedding is expected to approximate the 
Spectral Embedding for degree-regular graphs. 
\begin{approx result}
  For a $d$-regular graph $G$ and $m \leq n$, let $i_1, \ldots, i_m$ be $m$
  indices randomly sampled from $[n]$, for $\alpha$ close to 1, left singular vectors
  of $\mC =
    \frac{1}{\sqrt{n}} \bmat{\frac{\log.(\vx(i_1, \alpha))}{\|\log.(\vx(i_1, \alpha))\|}
    & \frac{\log.(\vx(i_2, \alpha))}{\|\log.(\vx(i_2, \alpha))\|}
    & \ldots
    & \frac{\log.(\vx(i_n, \alpha))}{\|\log.(\vx(i_n, \alpha))\|}}
  $
  approximates the eigenvectors of $\mW$.
\end{approx result}
\begin{justification}
  By Lemma \ref{lem:log_approx}, we know that $\lim_{\alpha \to 1^{-}}\frac{\log.(\vx(i_j, \alpha))}
  {\|\log.(\vx(i_j, \alpha))\|}$
  $ = \lim_{k_j \to \infty} \frac{\mW^{k_j} \ve_{i_j}}{\|\mW^{k_j} \ve_{i_j}\|}$.
  Let $\mQ \mLambda \mQ^T$ be the eigenvalue decomposition for $\mW$, we have for $d$-regular graphs,
  $\mQ_1 = \frac{\ve}{\sqrt{n}}$ and $\lambda_1 = 1$. Thus, for large enough $k_j$, $\mW^{k_j} \ve_{i_j} 
  \approx \lambda_1^{k_j} \mQ_1 \mQ_1^T \ve_{i_j} = \mQ_1 \mQ_1^T \ve_{i_j}$ and 
  $\|\mW^{k_j} \ve_{i_j}\| \approx \frac{1}{\sqrt{n}}$. 
  Therefore let $k_1, \ldots, k_m$ be $m$ randomly sampled large integers,
  we have
  \begin{align*}
    \mC =
    \frac{1}{\sqrt{n}} \bmat{\frac{\log.(\vx(i_1, \alpha))}{\|\log.(\vx(i_1, \alpha))\|}
    & \frac{\log.(\vx(i_2, \alpha))}{\|\log.(\vx(i_2, \alpha))\|}
    & \ldots
    & \frac{\log.(\vx(i_n, \alpha))}{\|\log.(\vx(i_n, \alpha))\|}}
  \end{align*}    
  approximates
  \begin{align*}
  \mB = \bmat{\mW^{k_1}\ve_{i_1} &
    \mW^{k_2} \ve_{i_2} &
    \ldots & \mW^{k_m} \ve_{i_m}}  
  \end{align*}    
  Further, by Approximation Result \ref{lem:sample}, we have in expectation left singular vectors of $\mB$
  approximates the eigenvectors of $\mW$, thus left singular vectors
  of $\mC$ approximates the eigenvectors of $\mW$ as well.
\end{justification}

\section{Empirical Comparison Results}\label{sec:examples}
We study the log-PageRank embedding on synthetic and real world graphs. We focus on those where the spectral embedding gives a good picture of the graph, as spectral embeddings may fail to give useful pictures for many real-world networks~\cite{Lang-2005-weaknesses}. 
We will analyse its performance on nearest neighbour graphs and the following  graphs with a strong geometry. 

\begin{fullwidthfigure}[t]
     \centering
     \begin{subfigure}[t]{0.33\linewidth}
         \centering
         \includegraphics[width=\linewidth]{1a}
         \caption{Graph for the word ``log PR'' with \\ 5000 nodes, 14962 edges}
         \sublabel{fig:1a-rep}
     \end{subfigure}%
     \begin{subfigure}[t]{0.33\linewidth}
         \centering
         \includegraphics[width=\linewidth]{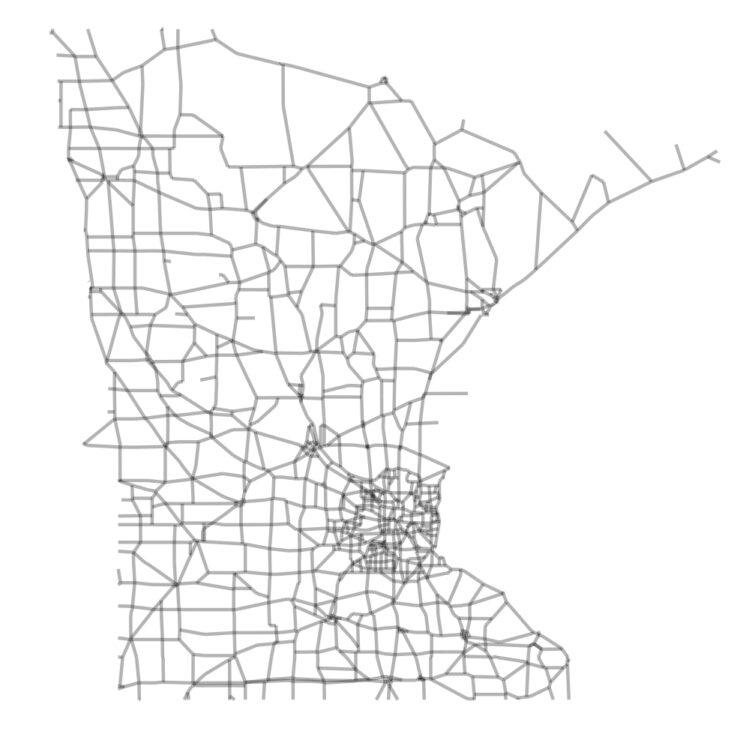}
         \caption{The Minnesota road network with \\ 2640 nodes, 3302 edges \sublabel{fig:minnesota-graph}}
     \end{subfigure}%
     \begin{subfigure}[t]{0.33\linewidth}
         \centering
         \includegraphics[width=\linewidth]{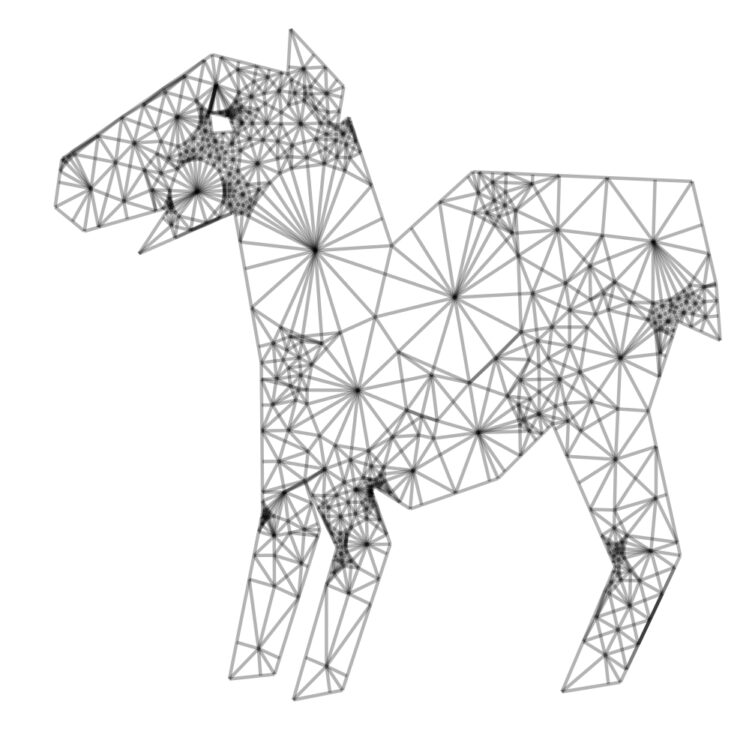}
         \caption{The Tapir \cite{tapir} graph with \\ 1024 nodes, 2846 edges
         \sublabel{fig:tapir-graph}}
     \end{subfigure}
\caption{Our geometric graphs}
\label{fig:geometric}
\end{fullwidthfigure}

\subsection{Implementation}
We implemented the log-PageRank embedding in Julia. We use the built in sparse LU solver to factorize the PageRank matrix $\mI - \alpha \mP$ to solve linear systems for large values of $\alpha$. And we use the built in dense SVD solver. There are many alternatives we could use here, but our focus was on understanding the log-PageRank embeddings rather than optimizing the speed at which we can compute them.

\subsection{Qualitative and Quantitative  Error and Approximation}
Given a spectral embedding and a log-PageRank embedding, we first simply  look at the differences in the pictures. Our results say these should look similar, not that they should be exactly the same as the graphs we study are not in the class where we expect sharp approximations. Let the second singular vectors of the log-PageRank embedding be  $\vu_2$, and the second eigenvector of the Laplacian be $\vz_2$.  Likewise for the 3rd vectors. So the spectral embedding is $\vz_2, \vz_3$ and the log-PageRank embedding is $\vu_2, \vu_3$. 

We quantitatively measure the error by evaluating the relative difference between the Rayleigh quotient with respect to the vectors used for embedding. This gives us the following measure: 
\begin{align}
\label{eq:appx_err}
    \text{approximation error} = \frac{s - p}{s}
\end{align}
where 
\begin{align*}
    s = \frac{\vz_2^T \mcL \vz_2}{\vz_2^T \vz_2}, \qquad 
    p = \frac{\vu_2^T \mcL \vu_2}{\vu_2^T \vu_2}.
\end{align*}

The second way we evaluate the embeddings is by looking at the joint plot of $\vu_2$ vs.~$\vz_2$ and $\vz_3$ vs.~$\vu_3$. If the embeddings are close, these should look like a straight line, or at least a very highly correlated relationships.

\subsection{Evaluation across graphs}
We record the error according to equation \eqref{eq:appx_err} in Table \ref{tab:emb_errors} for the following types of graphs.

\paragraph{Nearest Neighbor Graphs}
For a graph named $n-k$ nearest neighbor, there are $n$ points randomly distributed in the unit square and connected to $k$ nearest neighbors. 

\paragraph{Chain Graphs}
These are simply the chain graphs we had from the analysis in Section~\ref{sec:chain}.

\paragraph{Graphs with Strong Geometry}
These are the graphs from Figure~\ref{fig:geometric}.

\paragraph{Stochastic Block Models}
A graph named $\text{sbm}(n,k,p,q)$ has $k$ groups of $n$ vertices with inter-group probability $p$ and between group probability $q$. These show the worst approximation results and largest differences.

\begin{tuftetable}[t]
\centering
\begin{tabularx}{\linewidth}{lXXXX}
\toprule 
Graph & $\alpha=0.99$ & & \smash{\rlap{$\alpha=0.9999$}}\\
 \cmidrule(l){2-3}
 \cmidrule(l){4-5}
 & raw & log & raw & log \\
30-6 nearest neighbour & 3.27\% & 0.06\% & 2.89\% & 0.05\%  \\
 
 3000-6 nearest neighbour & 47.6\% & 0.37\% & 5.06\% & 2.88\% \\
 
 10000-6 nearest neighbour & 170.75\% & 2.13\% & 13.5\% & 1.76\% \\
 \midrule 
 30 chain & 26.88\% & 0.47\% & 28.42\%  & 6.02\% \\ 
 
 3000 chain & 2858.82\% & 1.06\% & 30.38\% & 0.75\% \\
 \midrule 
 Minnesota $n=2640$ & 16.07\% & 1.97\% & 11.15\% & 0.44\% \\
 
 Tapir $n=1024$ & 10.17\% & 1.13\% & 15.41\% & 0.66\% \\
 
 LogPR $n=5000$ & 19.95\% & 0.15\% & 4.76\% & 0.34\%  \\
 \midrule 
 sbm(50,60,0.001,0.005) & 51.77\% & 15.22\% & 51.32\% & 67.25\%  \\
 
 sbm(1000,3,0.001,0.005) & 47.35\% & 16.93\% & 45.78\% & 89.39\% \\
 
 sbm(50,60,0.25,0.005) & 17.88\% & 15.22\% & 90.13\% & 402.27\%  \\

 sbm(1000,3,0.25,0.001) & 53.7\% & 1.04\% & 16.21\% & 15.73\% \\ [1ex]
 \bottomrule 
\end{tabularx}
\caption{Error between PageRank embedding and spectral embedding for different graphs at a low teleportation probability, $\alpha = 0.99$ and at a higher one $\alpha = 0.99999$ both without log (raw) and with log.}
\label{tab:emb_errors}
\end{tuftetable}

\subsection{A few examples}

Figure~\ref{fig:minnesota} and Figure~\ref{fig:tapir} show the embeddings as $\alpha$ varies both with and without the nonlinear \emph{log} operation for the graphs in figures \ref{fig:minnesota-graph} and \ref{fig:tapir-graph}. The result on the ``log PR'' graph from Figure~\ref{fig:1a-rep} was in the introduction. 

On nearest neighbor graphs, such as Figure~\ref{fig:10knn}, these embeddings show a clear rotational ambiguity that might arise with other evaluations of this strategy. (This occured with the other graphs too.) Put plainly, the eigenvectors are almost in 2d invariant subspace. Consequently, when we randomize the method, we can only capture this near 2d subspace up to rotation. However, this will not show high error with respect to the approximation error measure as the results are all near eigenvectors. 

We show one of the examples of the stochastic block model in Figure~\ref{fig:highercon}. Although this has bad approximation with respect to the spectral embedding, the result for the log-PageRank embedding for $\alpha=0.99$ is arguably better than the spectral embedding. 


\begin{marginfigure}[-1pt]
\centering
\hspace{-15pt}\includegraphics[width=1.1\linewidth]{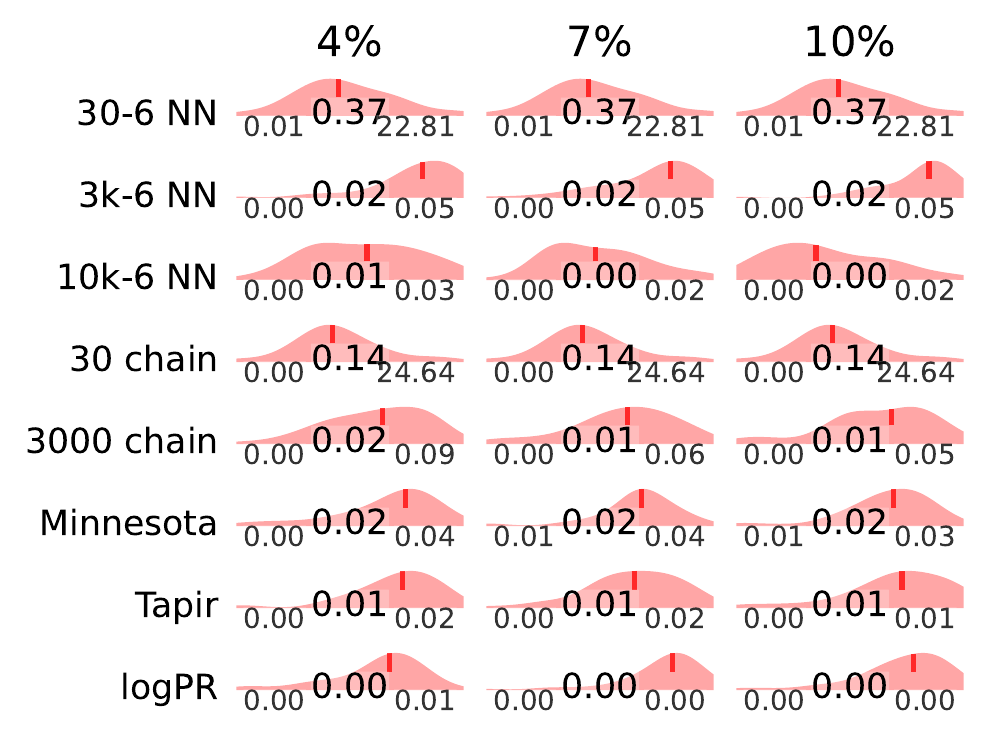}
\caption{Error variation with column for $\log$ of PageRank with $\alpha = 0.99$. The percentage indicated in the column headings are the fraction of the nodes as seeds. Each entry is the variance, the maximum and the minimum for 50 trials.}
\label{tab:err_var_log}
\end{marginfigure}

\begin{figure}[p]
\begin{fullwidth}
\centering
    \parbox{0.3\linewidth}{%
    \parbox{0.9\linewidth}{\caption{Comparison of embeddings for the Minnesota network.}
    \label{fig:minnesota}}
    \begin{subfigure}{\linewidth}
    \includegraphics[width=\linewidth]{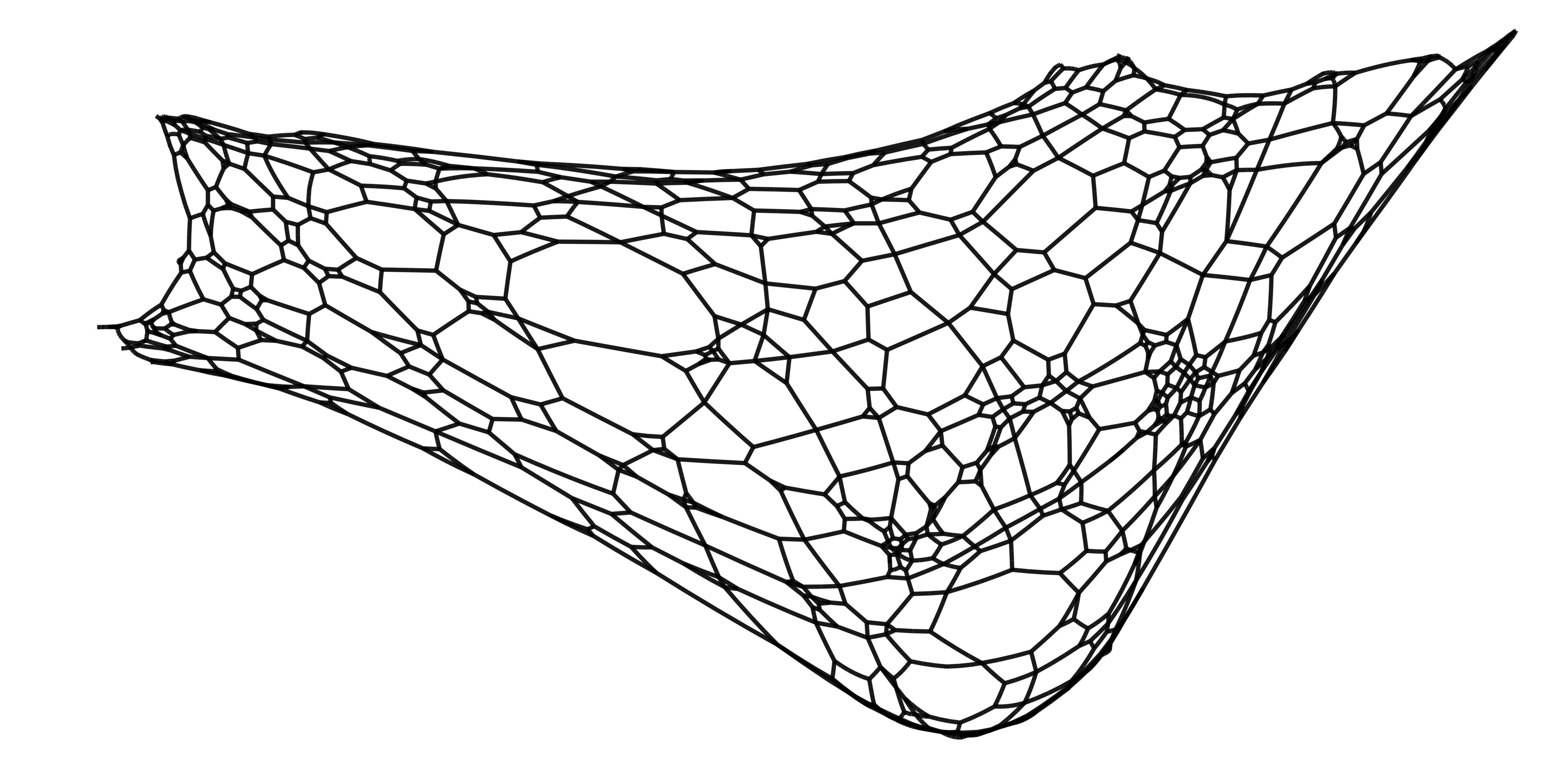}
    \caption{Spectral Embedding }
    \end{subfigure}%
    }%
    \parbox{0.7\linewidth}{%
    \smash{\raisebox{-18pt}{\llap{\rotatebox{90}{\footnotesize$\alpha=0.99$}}}}%
    \begin{subfigure}{0.33\linewidth}
    \includegraphics[width=\linewidth]{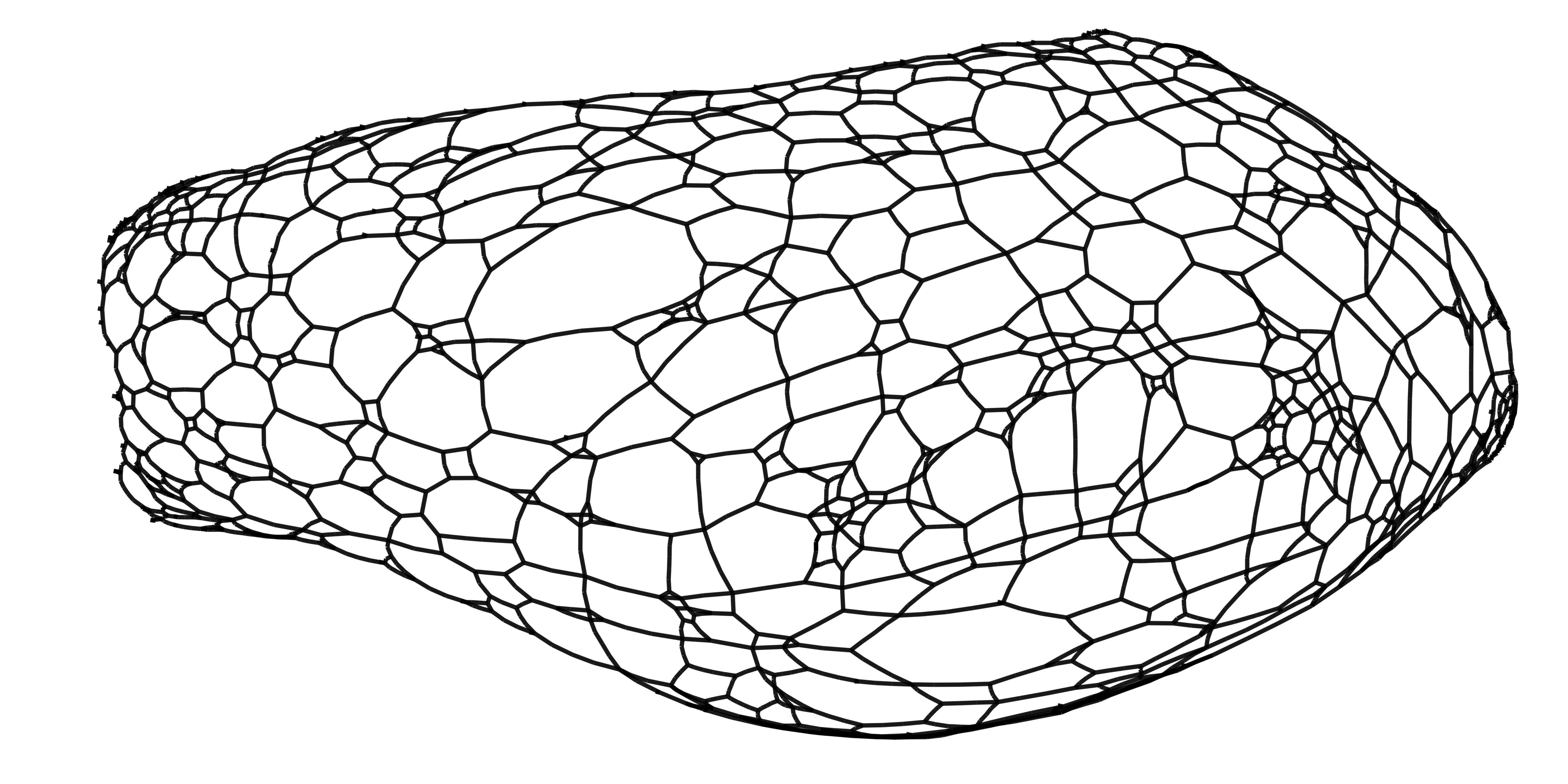}
    \caption{Log-PageRank}
    \end{subfigure}%
    \begin{subfigure}{0.33\linewidth}
    \includegraphics[width=\linewidth]{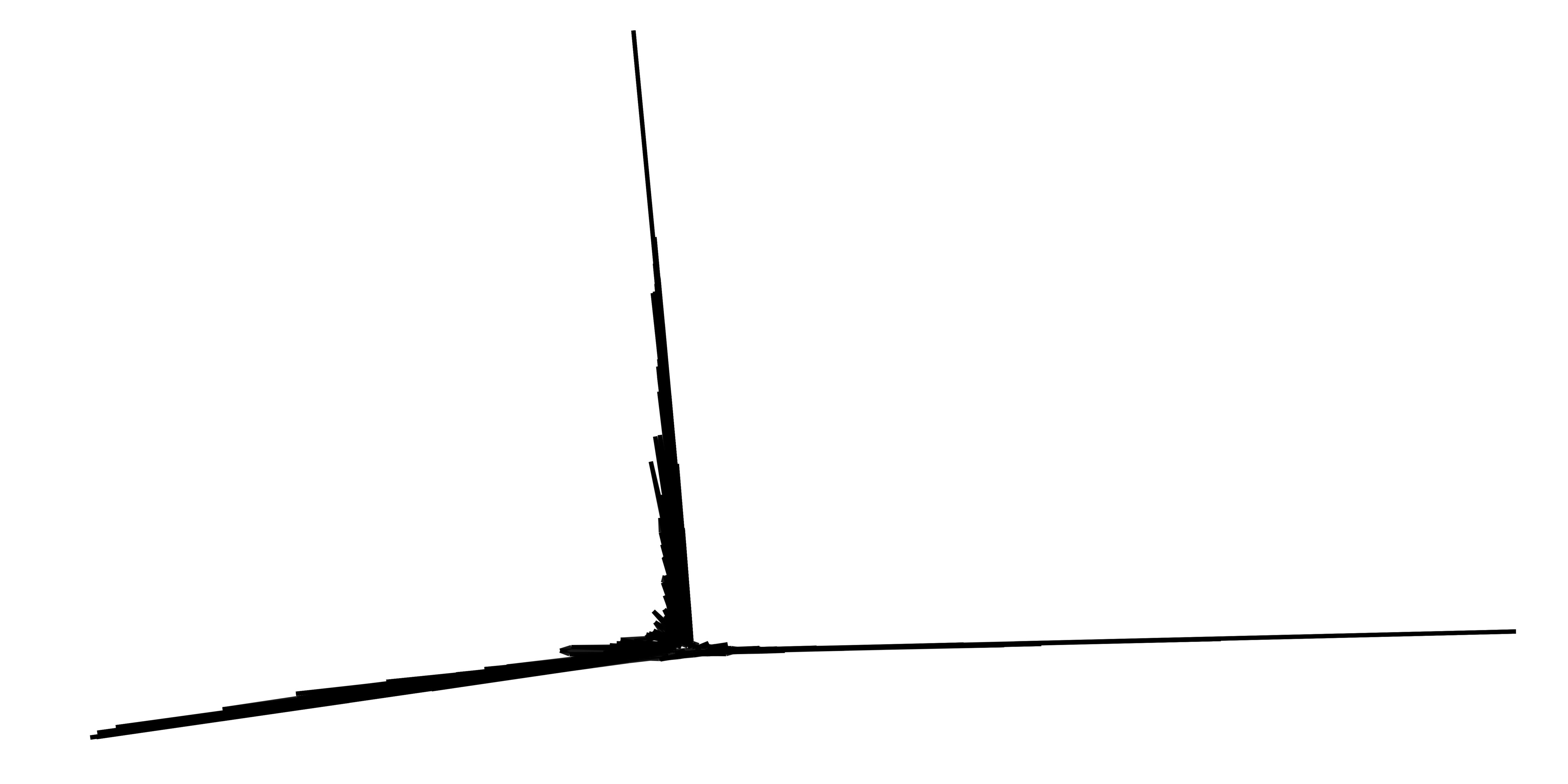}
    \caption{PageRank}
    \end{subfigure}%
    \begin{subfigure}{0.33\linewidth}
    \includegraphics[width=0.5\linewidth]{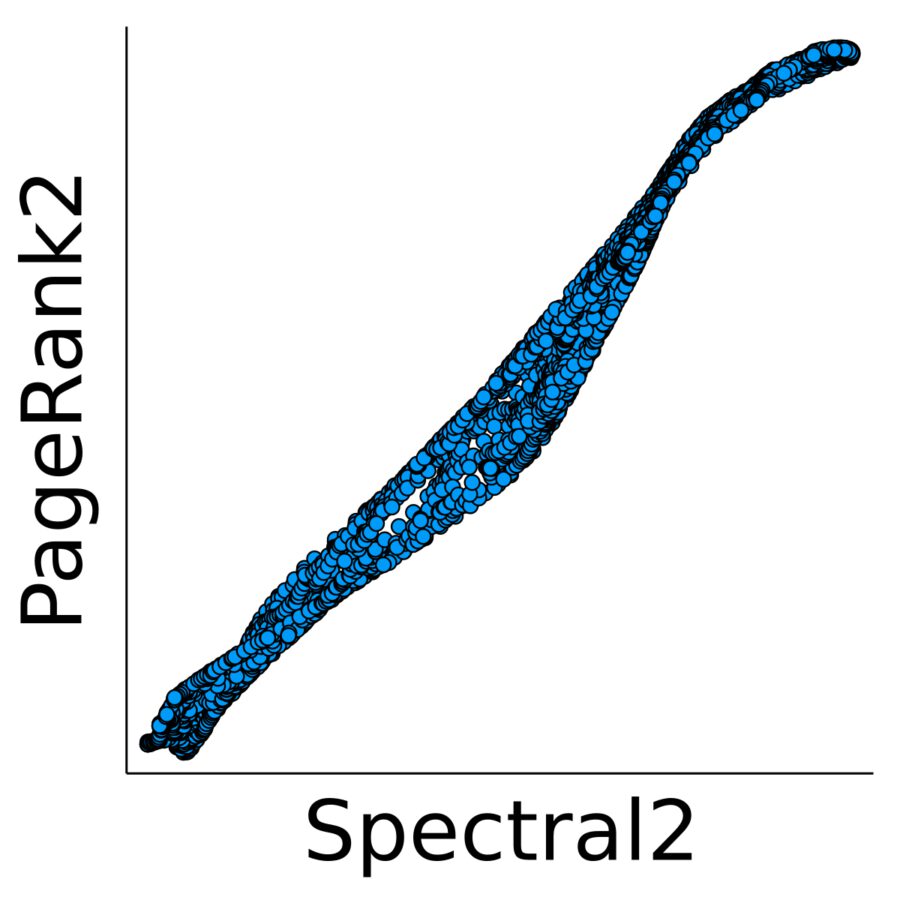}%
    \includegraphics[width=0.5\linewidth]{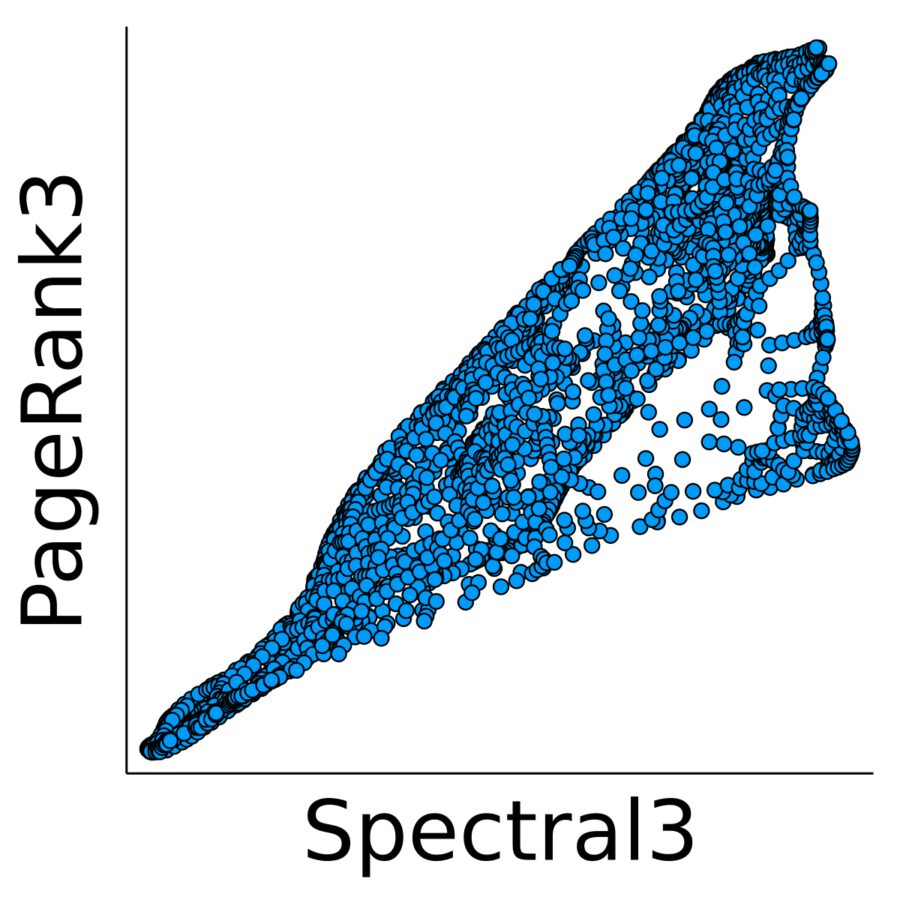}
    \caption{$\!\vu_2\!$ vs $\!\vz_2\!$ / $\!\vu_3\!$ vs $\!\vz_3\!$}
    \end{subfigure}
\centering

\bigskip 
    \smash{\raisebox{-24pt}{\llap{\rotatebox{90}{\footnotesize$\alpha=0.9999$}}}}%
    \begin{subfigure}{0.33\linewidth}
    \includegraphics[width=\linewidth]{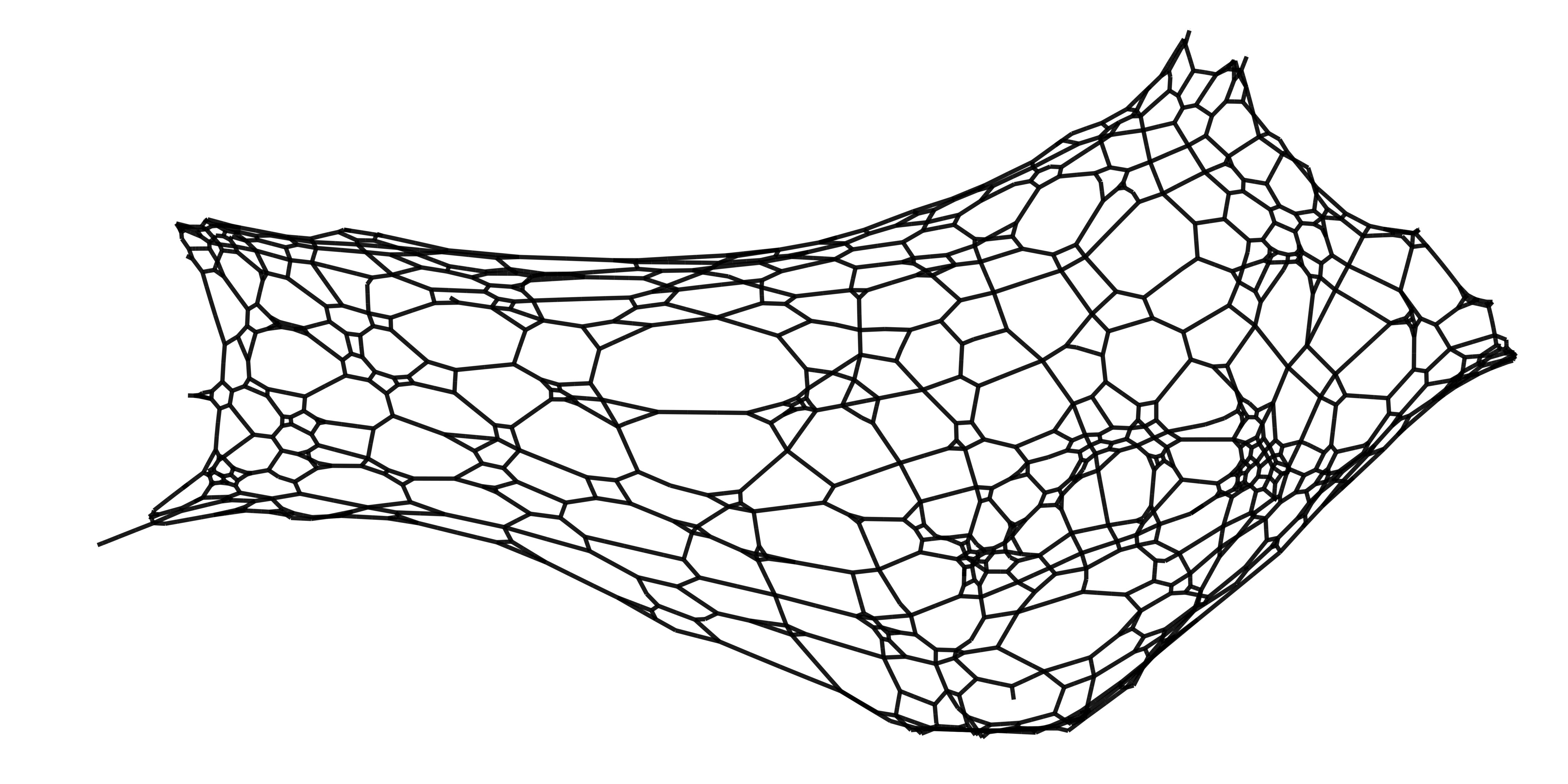}
    \caption{Log-PageRank}
    \end{subfigure}%
    \begin{subfigure}{0.33\linewidth}
    \includegraphics[width=\linewidth]{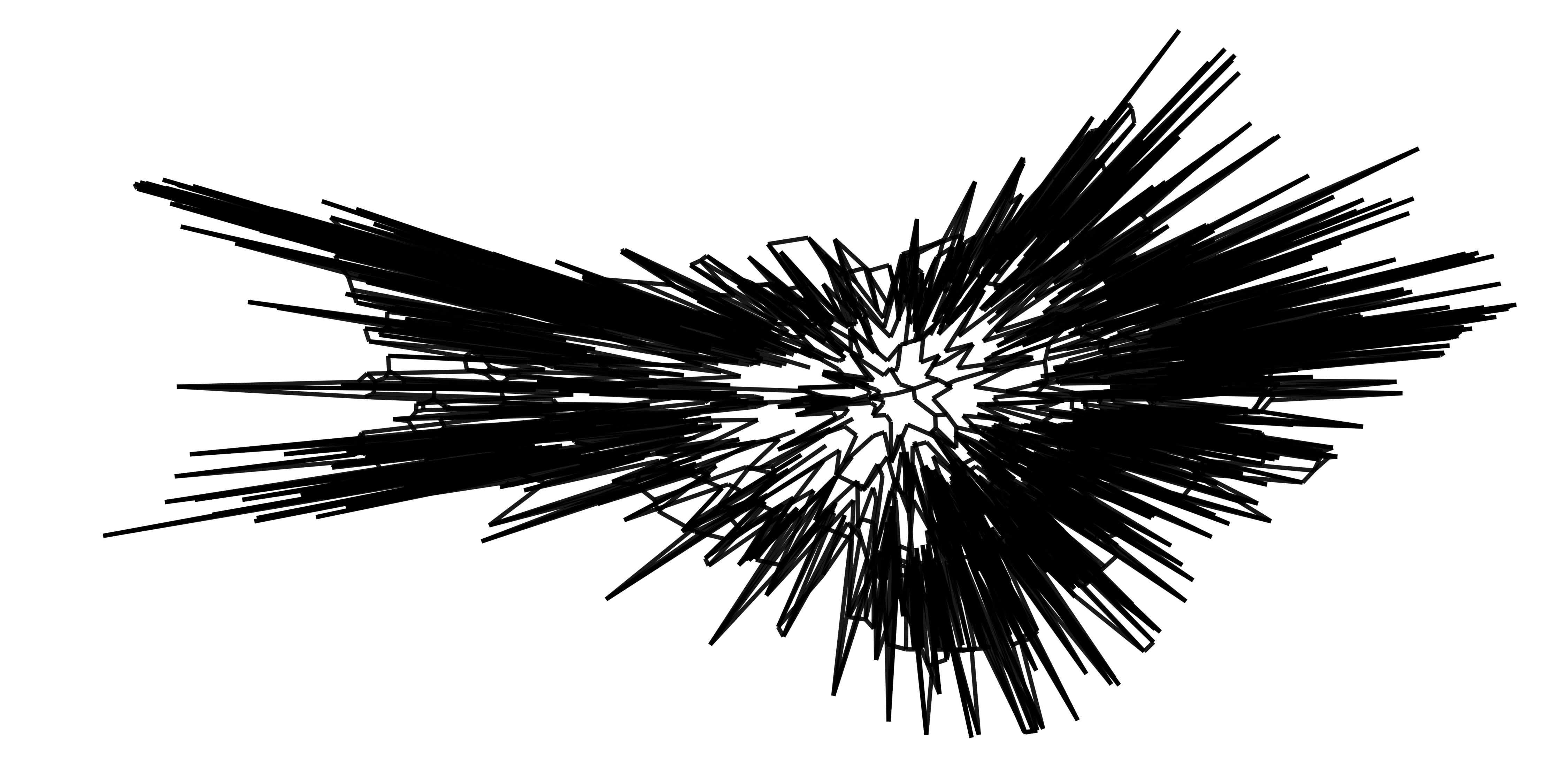}
    \caption{PageRank}
    \end{subfigure}%
    \begin{subfigure}{0.33\linewidth}
    \includegraphics[width=0.5\linewidth]{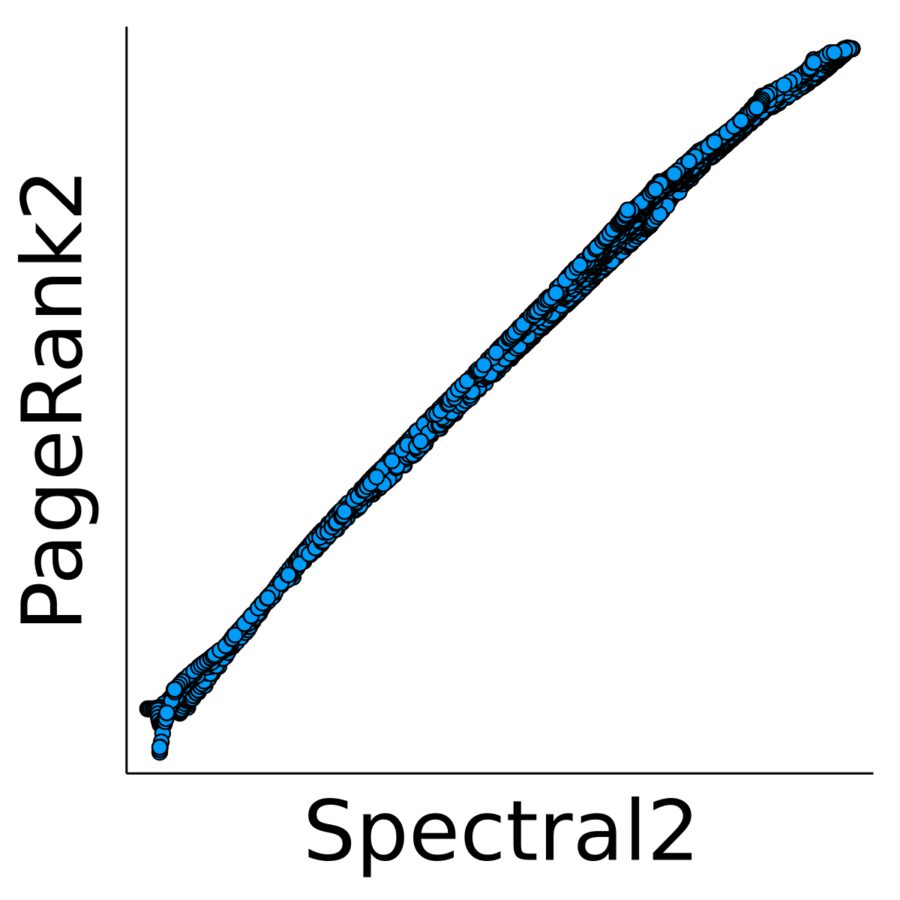}%
    \includegraphics[width=0.5\linewidth]{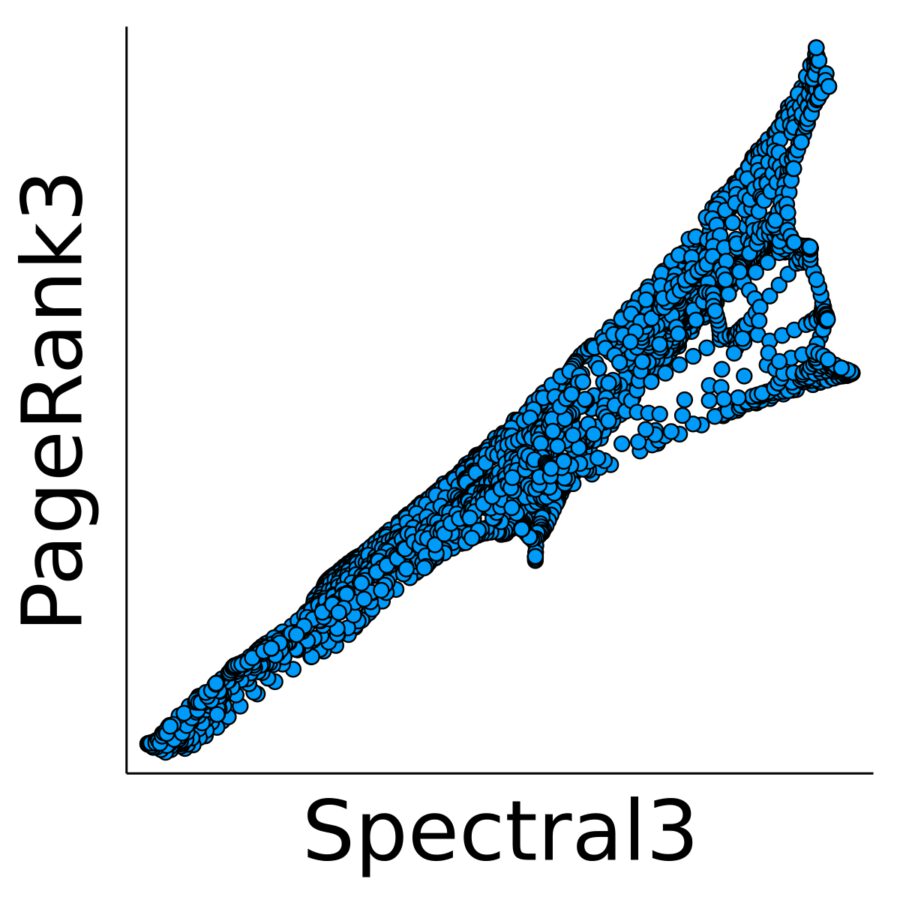}
    \caption{$\!\vu_2\!$ vs $\!\vz_2\!$ / $\!\vu_3\!$ vs $\!\vz_3\!$}
    \end{subfigure}
}

\smallskip
\hrule 
\smallskip
 
\centering
    \parbox{0.3\linewidth}{%
    \parbox{0.9\linewidth}{\caption{Comparison of embeddings for the Tapir graph.}
    \label{fig:tapir}}
    \begin{subfigure}{\linewidth}
    \includegraphics[width=\linewidth]{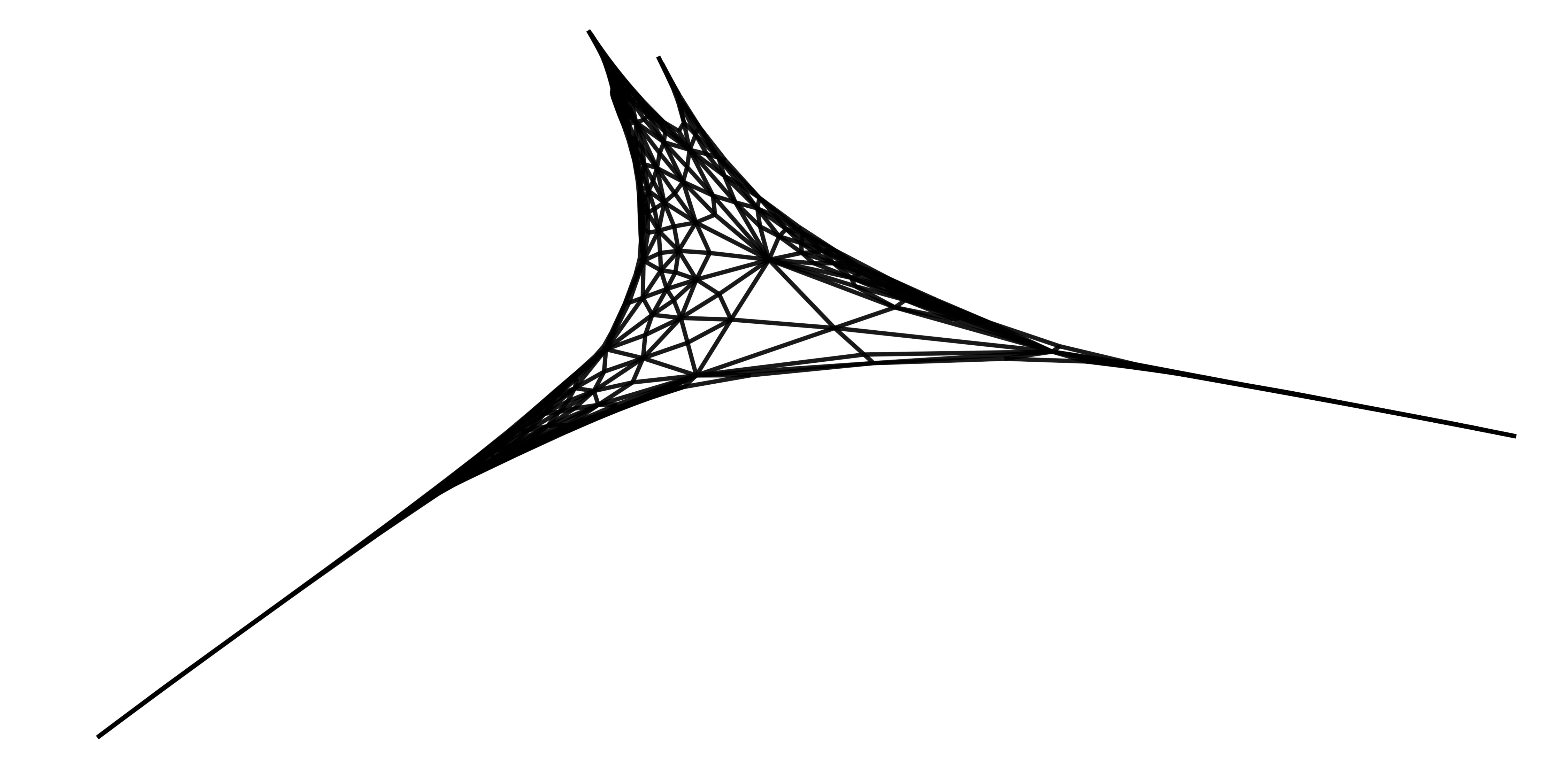}
    \caption{Spectral Embedding }
    \end{subfigure}%
    }%
    \parbox{0.7\linewidth}{%
    \smash{\raisebox{-18pt}{\llap{\rotatebox{90}{\footnotesize$\alpha=0.99$}}}}%
    \begin{subfigure}{0.33\linewidth}
    \includegraphics[width=\linewidth]{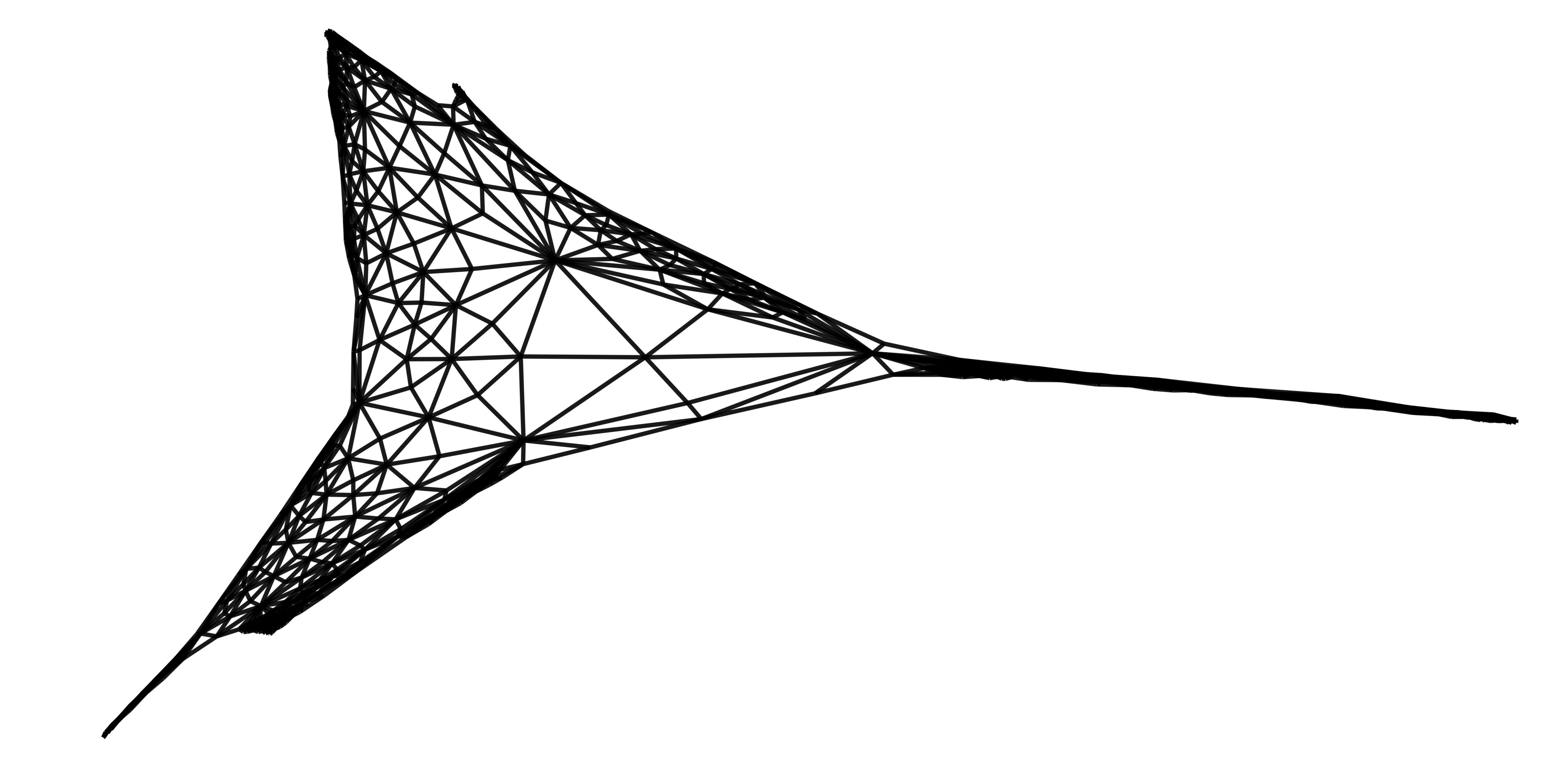}
    \caption{Log-PageRank}
    \end{subfigure}%
    \begin{subfigure}{0.33\linewidth}
    \includegraphics[width=\linewidth]{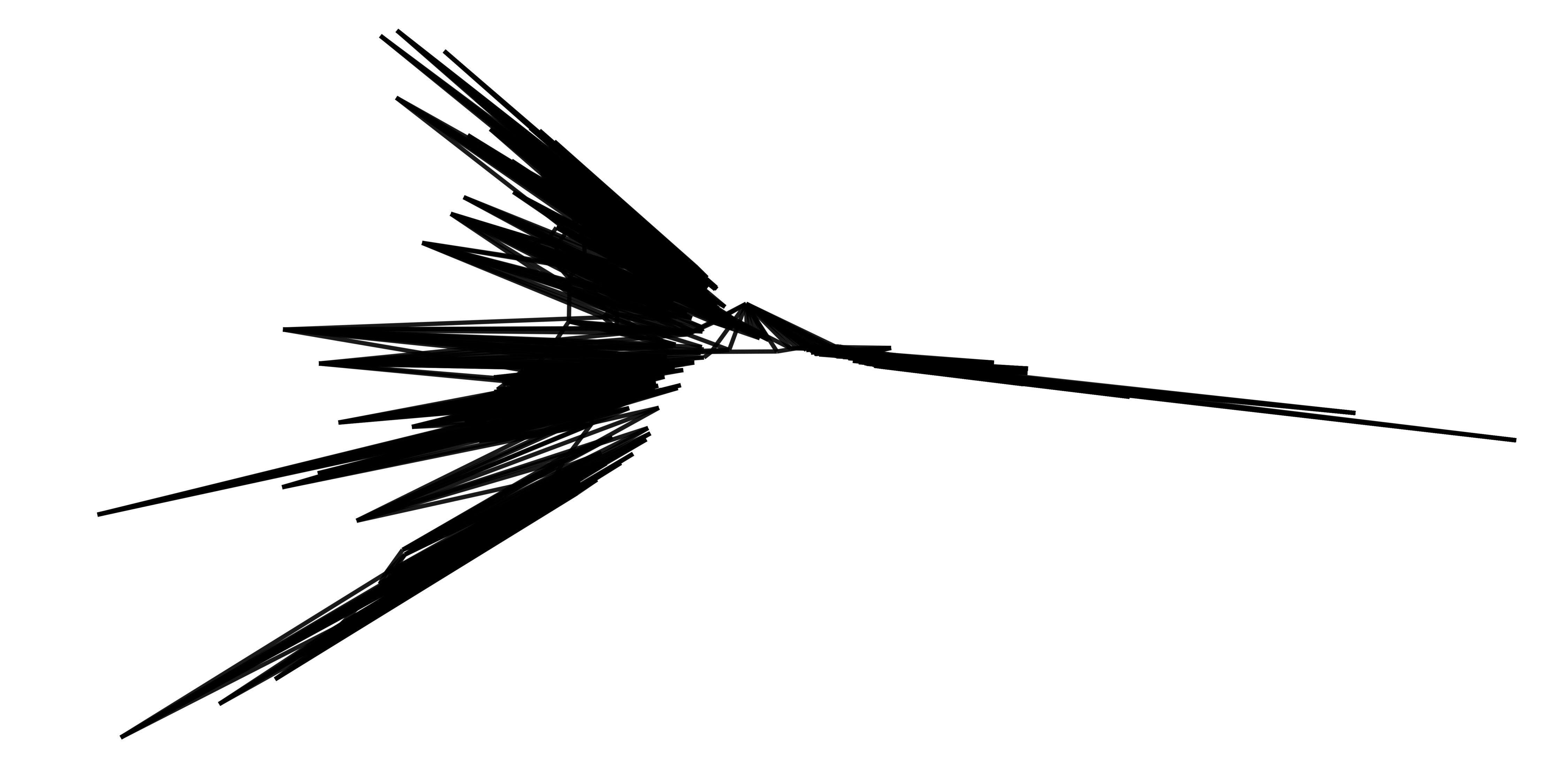}
    \caption{PageRank}
    \end{subfigure}%
    \begin{subfigure}{0.33\linewidth}
    \includegraphics[width=0.5\linewidth]{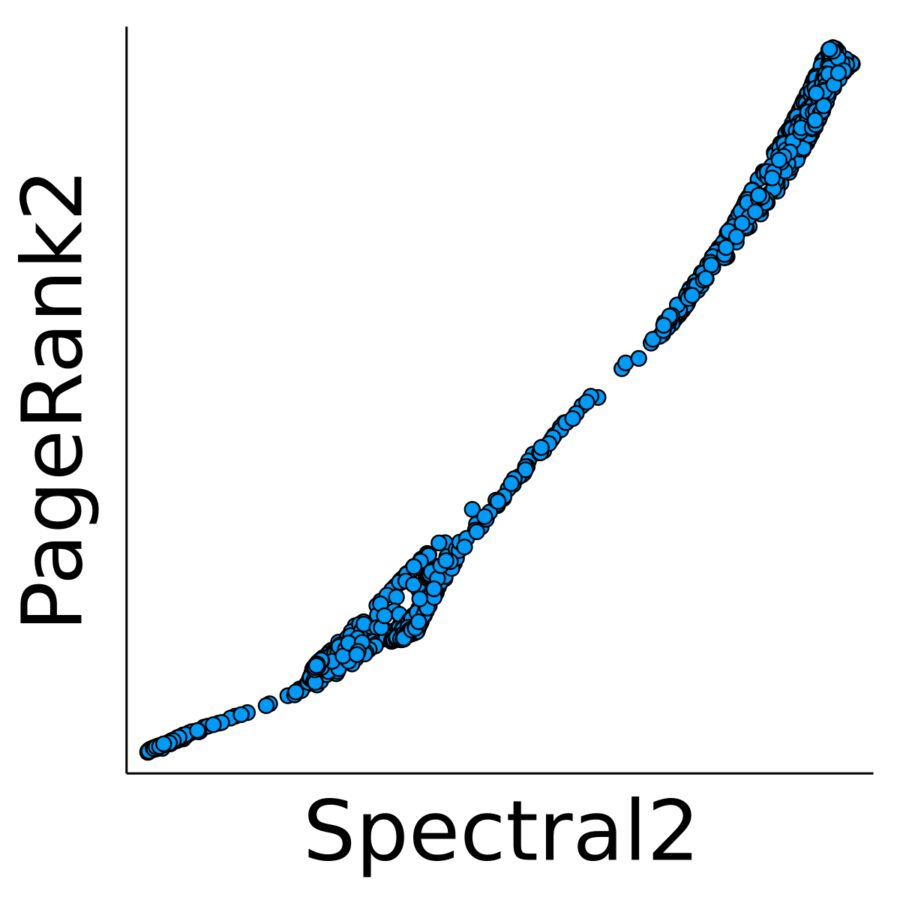}%
    \includegraphics[width=0.5\linewidth]{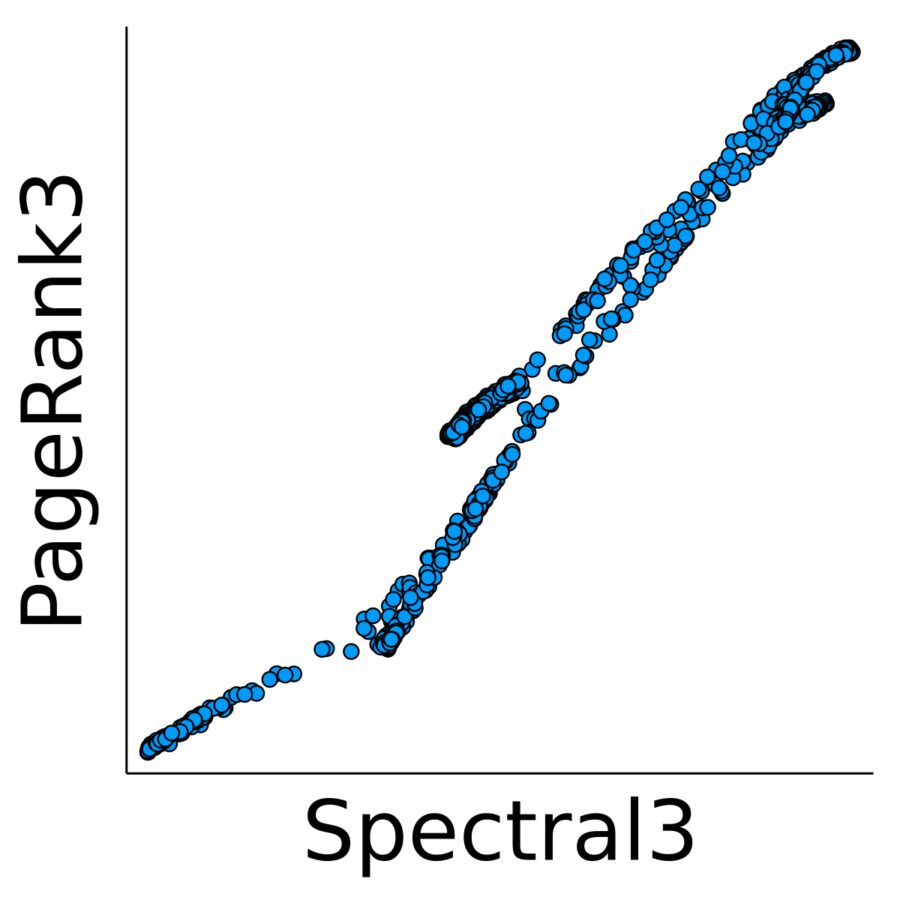}
    \caption{$\!\vu_2\!$ vs $\!\vz_2\!$ / $\!\vu_3\!$ vs $\!\vz_3\!$}
    \end{subfigure}
\centering

\bigskip 
    \smash{\raisebox{-24pt}{\llap{\rotatebox{90}{\footnotesize$\alpha=0.9999$}}}}%
    \begin{subfigure}{0.33\linewidth}
    \includegraphics[width=\linewidth]{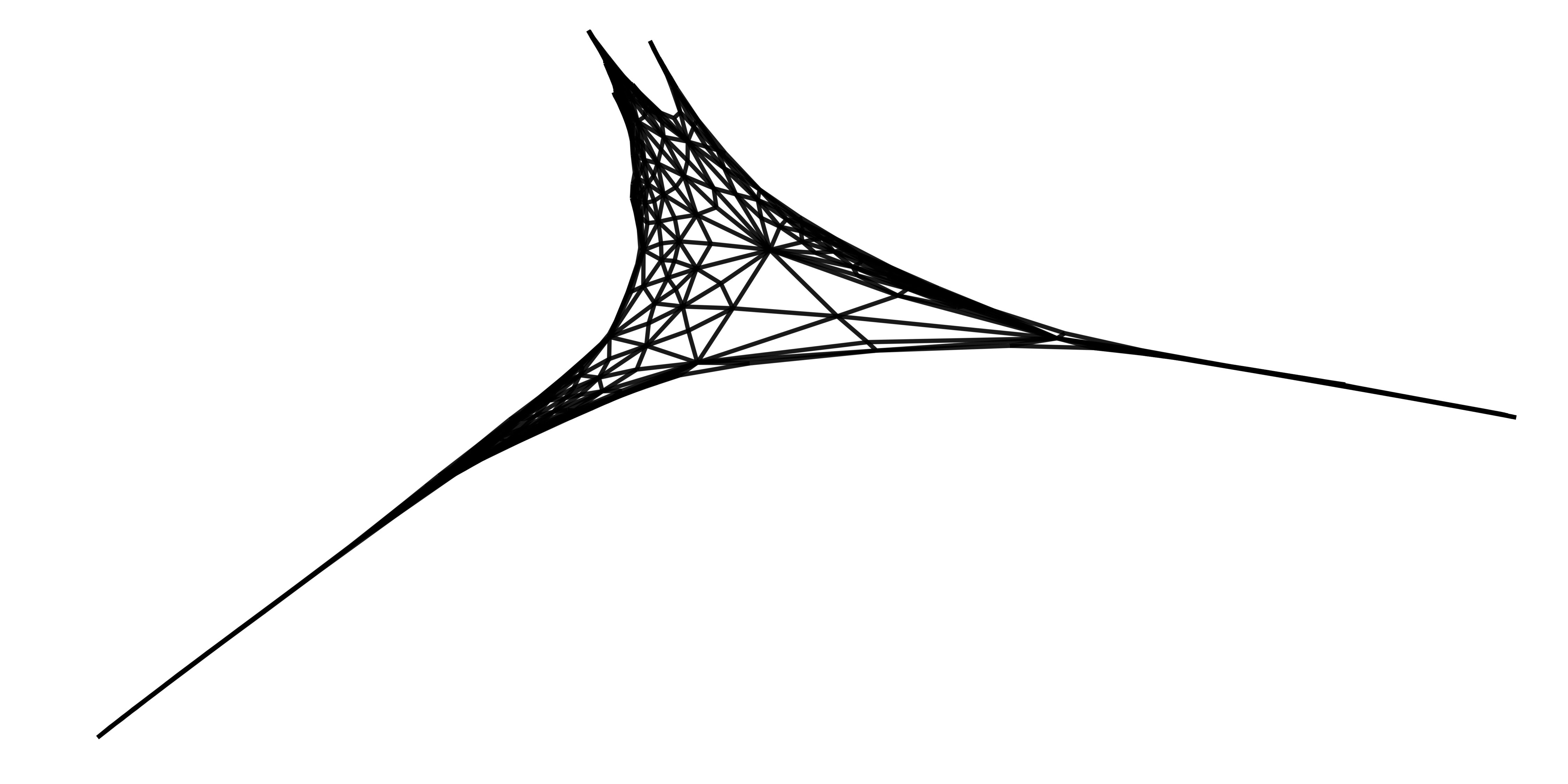}
    \caption{Log-PageRank}
    \end{subfigure}%
    \begin{subfigure}{0.33\linewidth}
    \includegraphics[width=\linewidth]{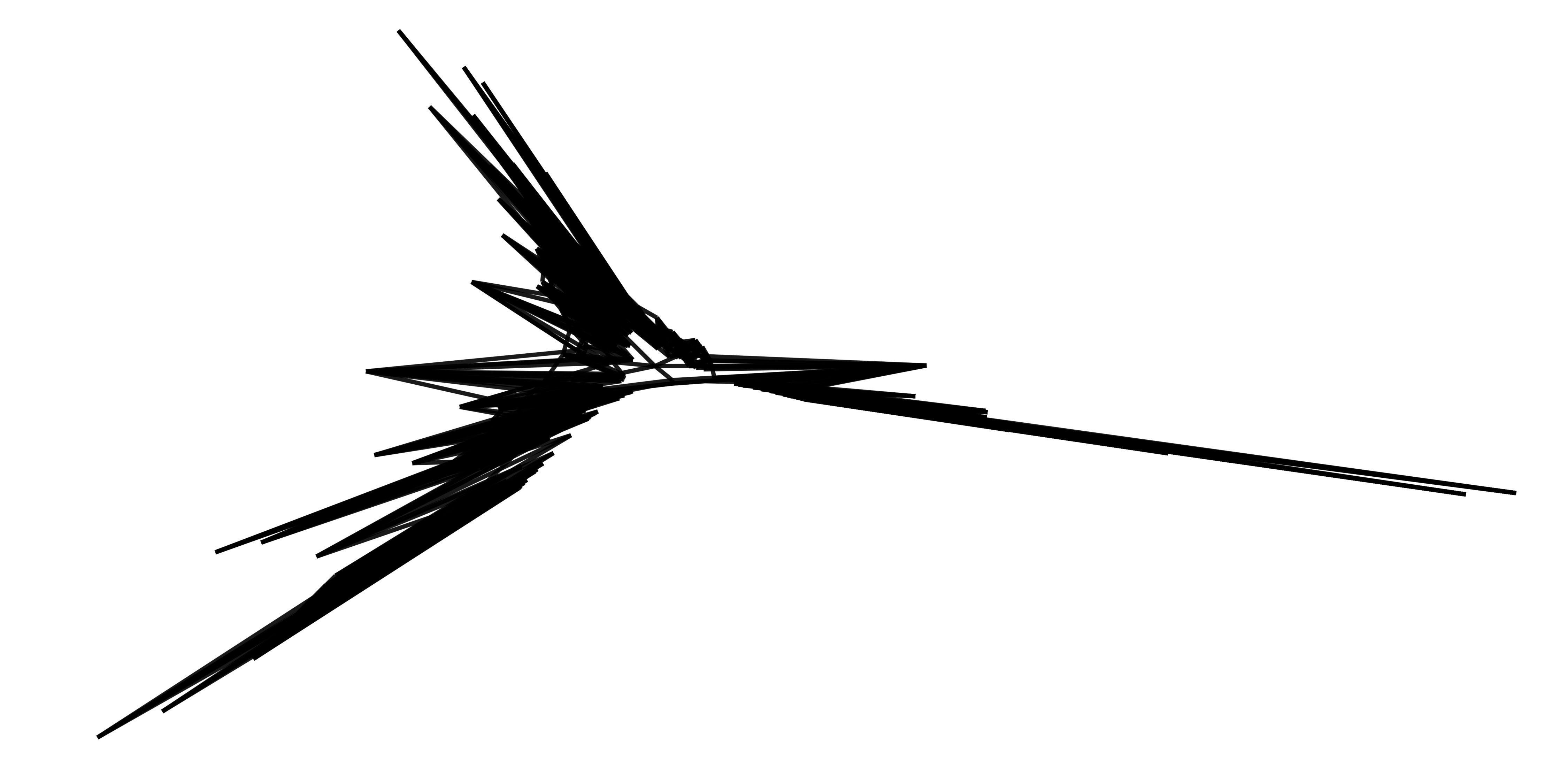}
    \caption{PageRank}
    \end{subfigure}%
    \begin{subfigure}{0.33\linewidth}
    \includegraphics[width=0.5\linewidth]{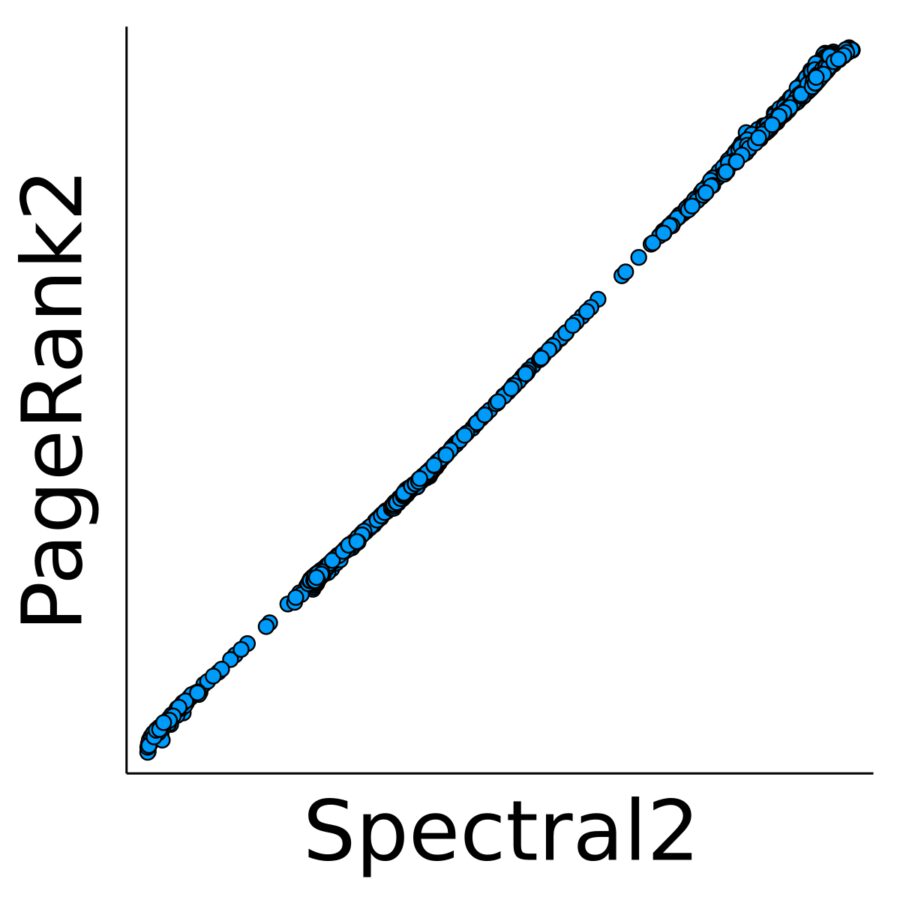}%
    \includegraphics[width=0.5\linewidth]{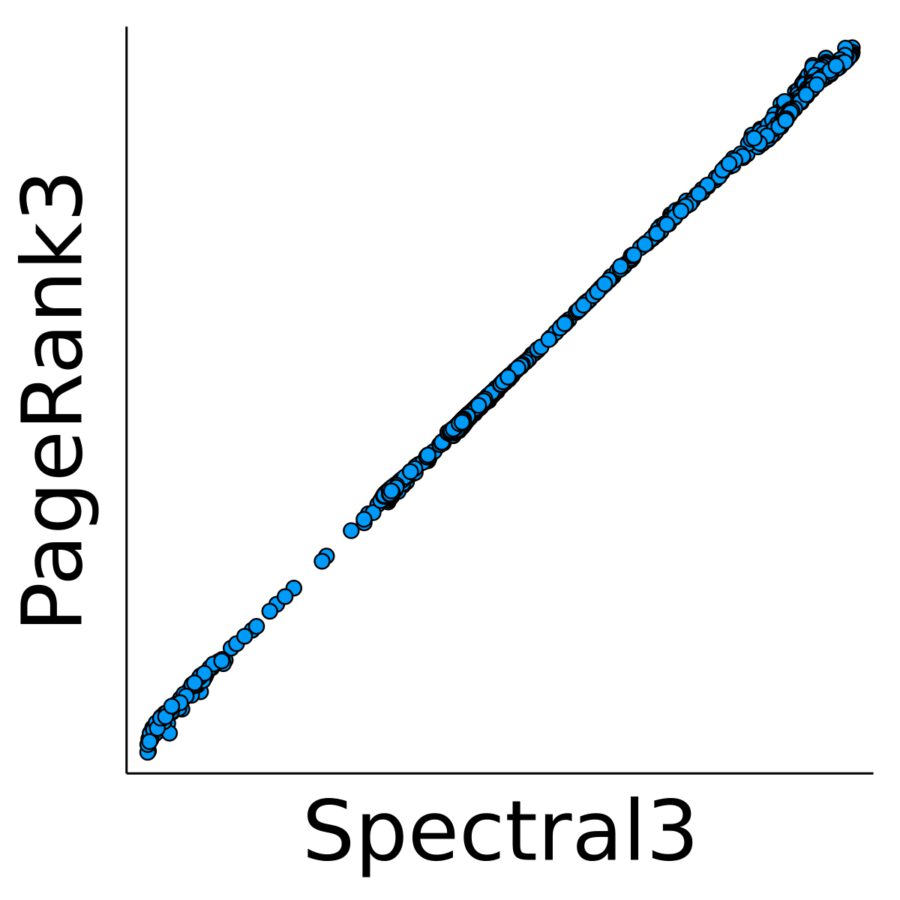}
    \caption{$\!\vu_2\!$ vs $\!\vz_2\!$ / $\!\vu_3\!$ vs $\!\vz_3\!$}
    \end{subfigure}
}

\bigskip 
\hrule 
\smallskip
\centering
    \parbox{0.3\linewidth}{%
    \parbox{0.95\linewidth}{\caption{Embedding for 10000 node graph with 6 nearest neighbours. Note that log-PageRank and spectral embeddings are similar after a rotation. This is expected because the eigenvalues are close in magnitude.}
    \label{fig:10knn}}
    \begin{subfigure}{\linewidth}
    \includegraphics[width=\linewidth]{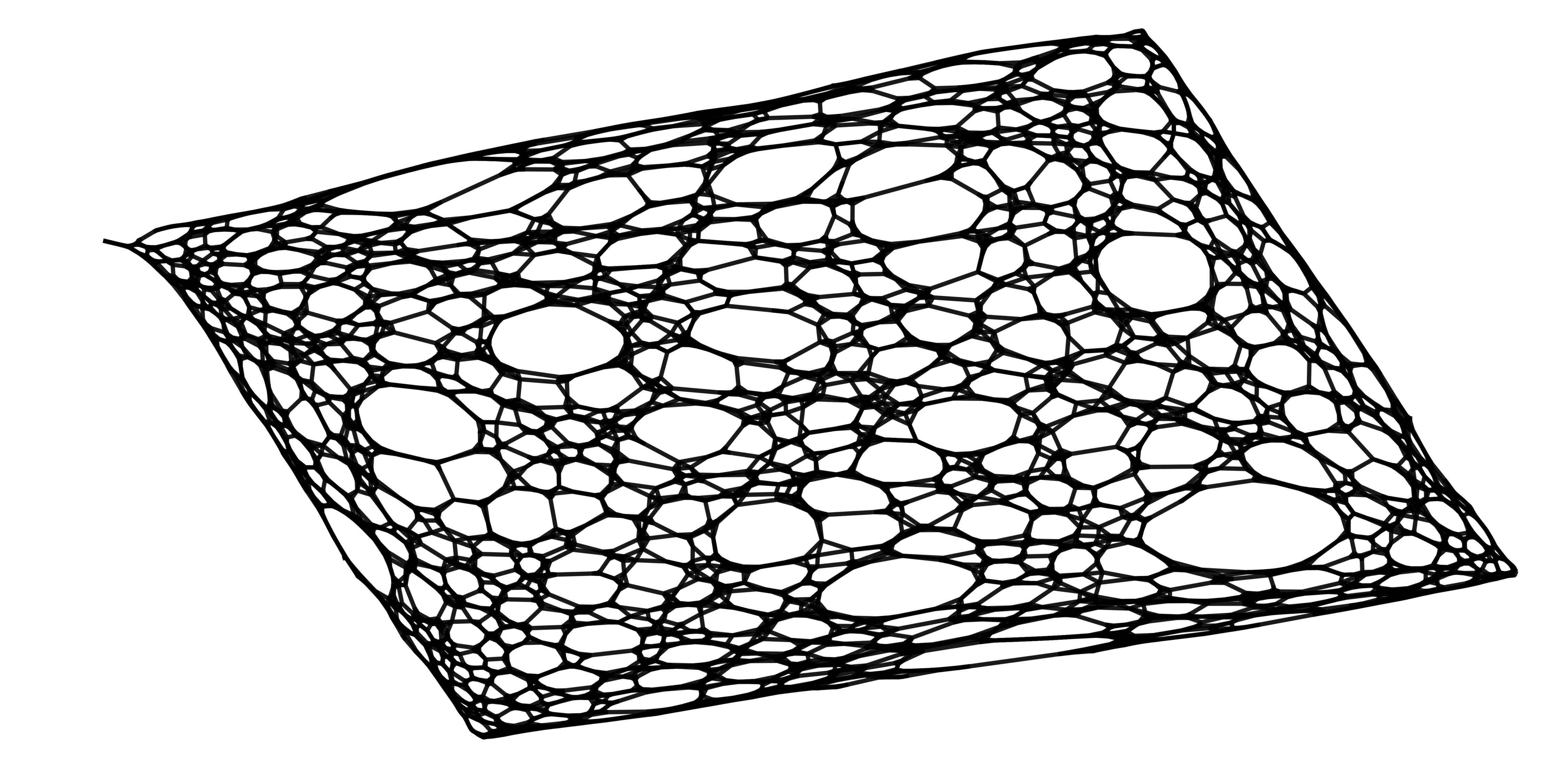}
    \caption{Spectral Embedding }
    \end{subfigure}%
    }%
    \parbox{0.7\linewidth}{%
    \smash{\raisebox{-18pt}{\llap{\rotatebox{90}{\footnotesize$\alpha=0.99$}}}}%
    \begin{subfigure}{0.33\linewidth}
    \includegraphics[width=\linewidth]{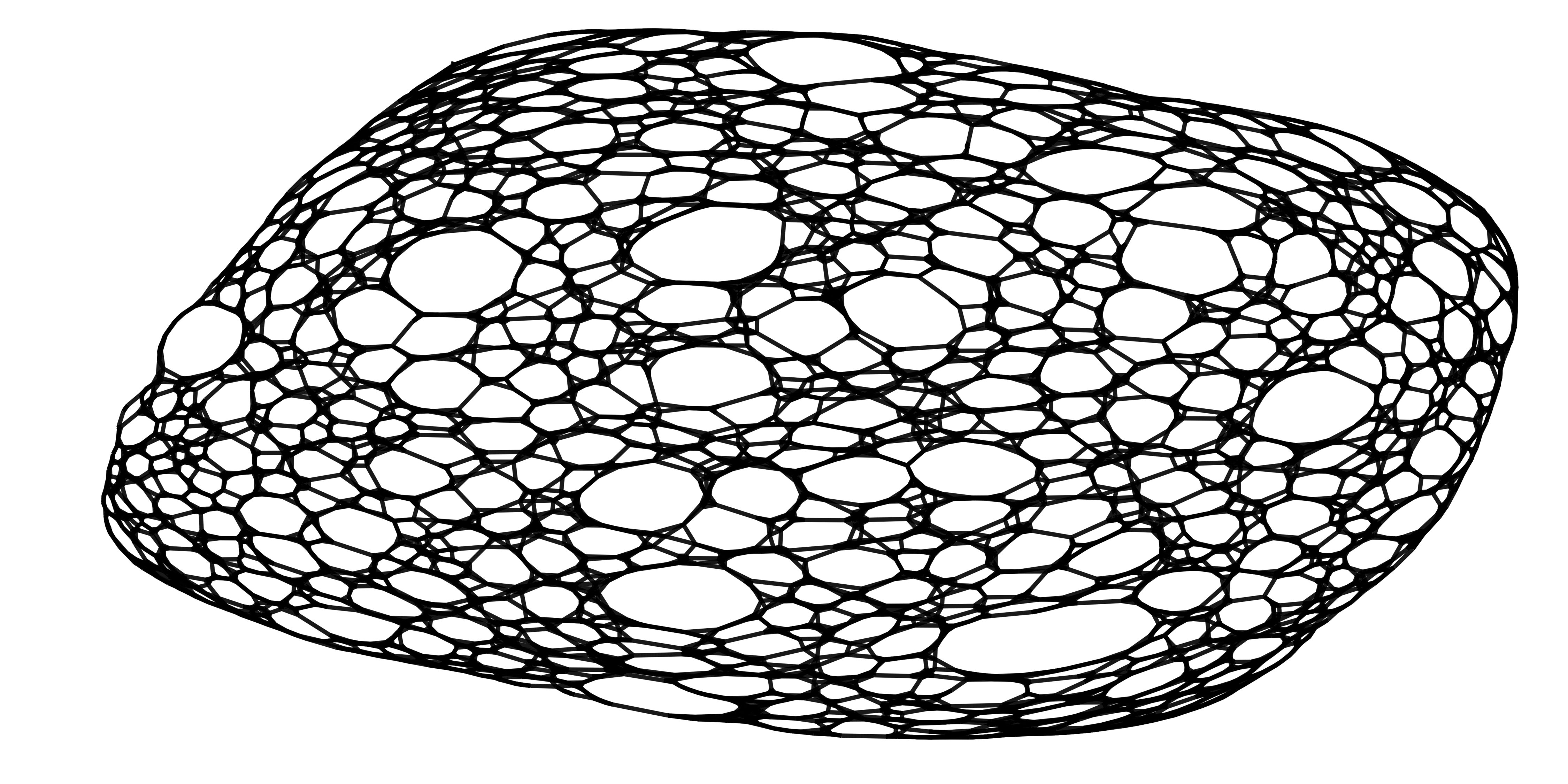}
    \caption{Log-PageRank}
    \end{subfigure}%
    \begin{subfigure}{0.33\linewidth}
    \includegraphics[width=\linewidth]{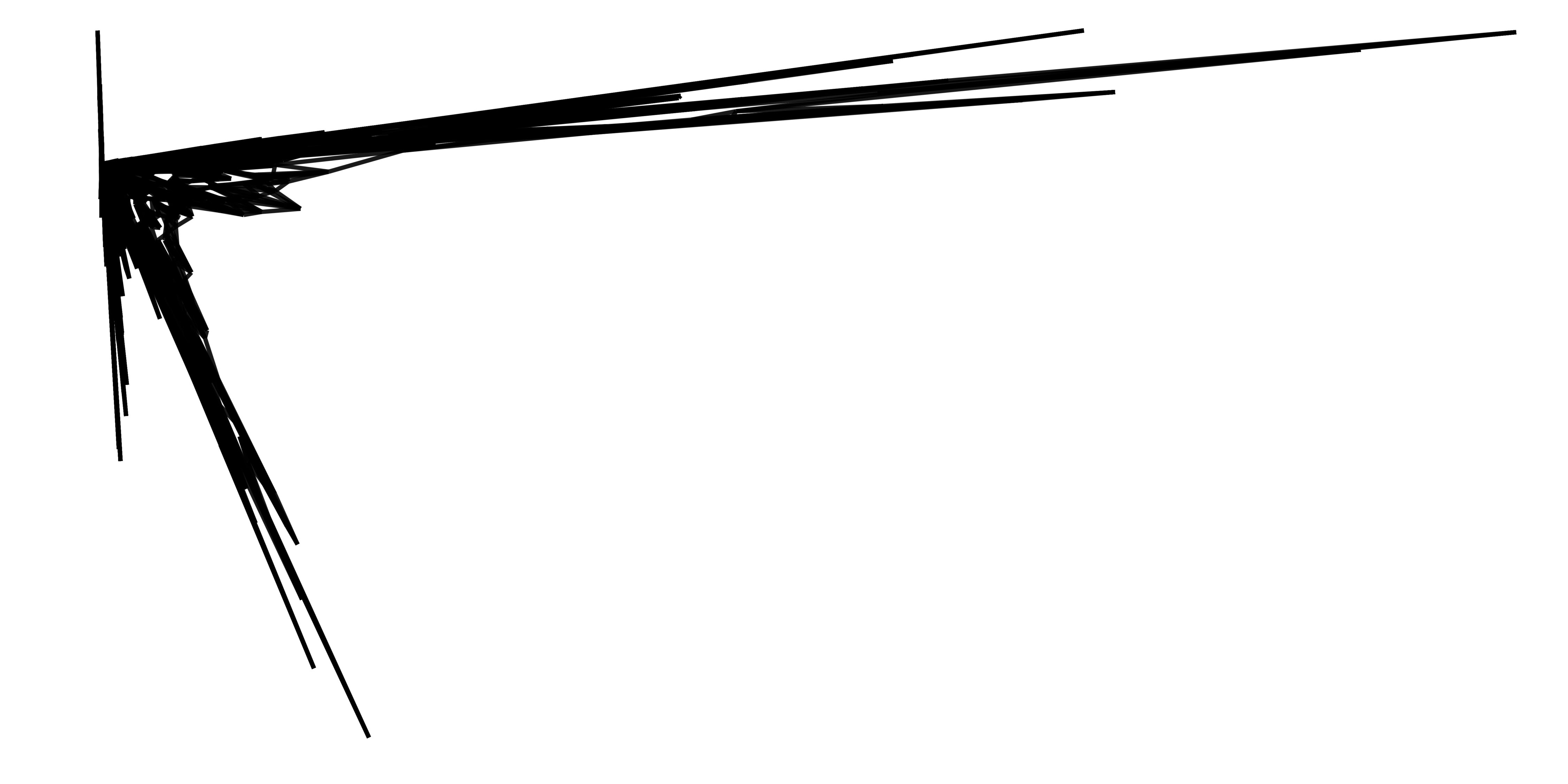}
    \caption{PageRank}
    \end{subfigure}%
    \begin{subfigure}{0.33\linewidth}
    \includegraphics[width=0.5\linewidth]{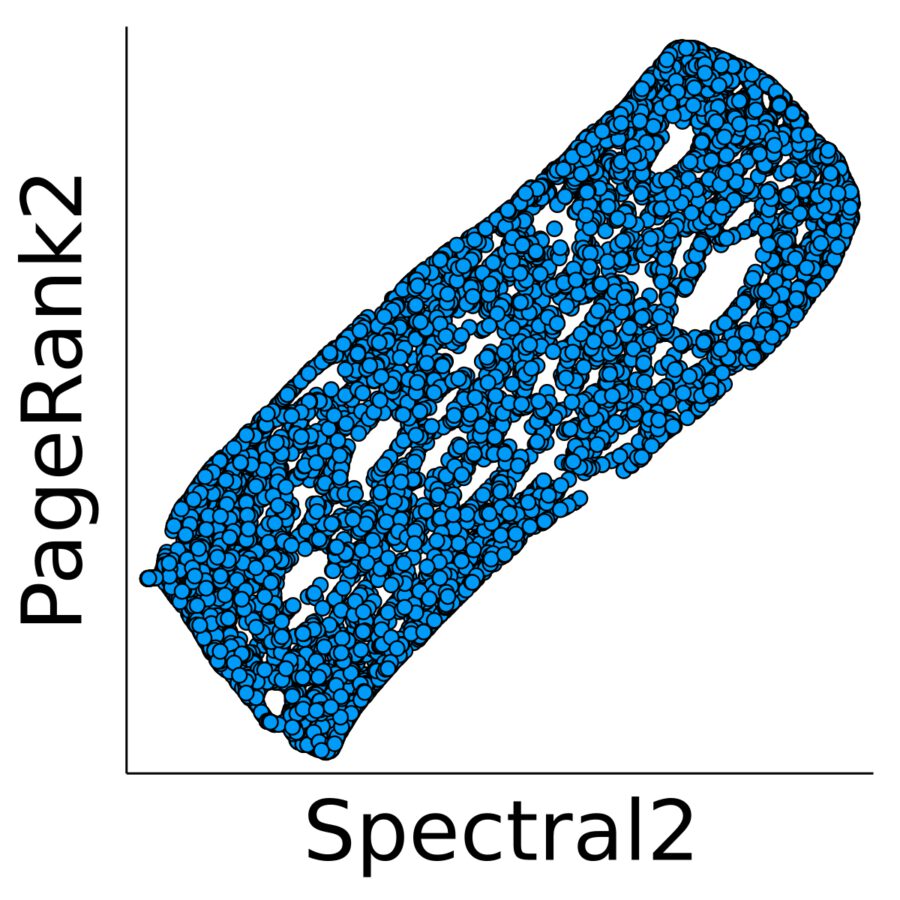}%
    \includegraphics[width=0.5\linewidth]{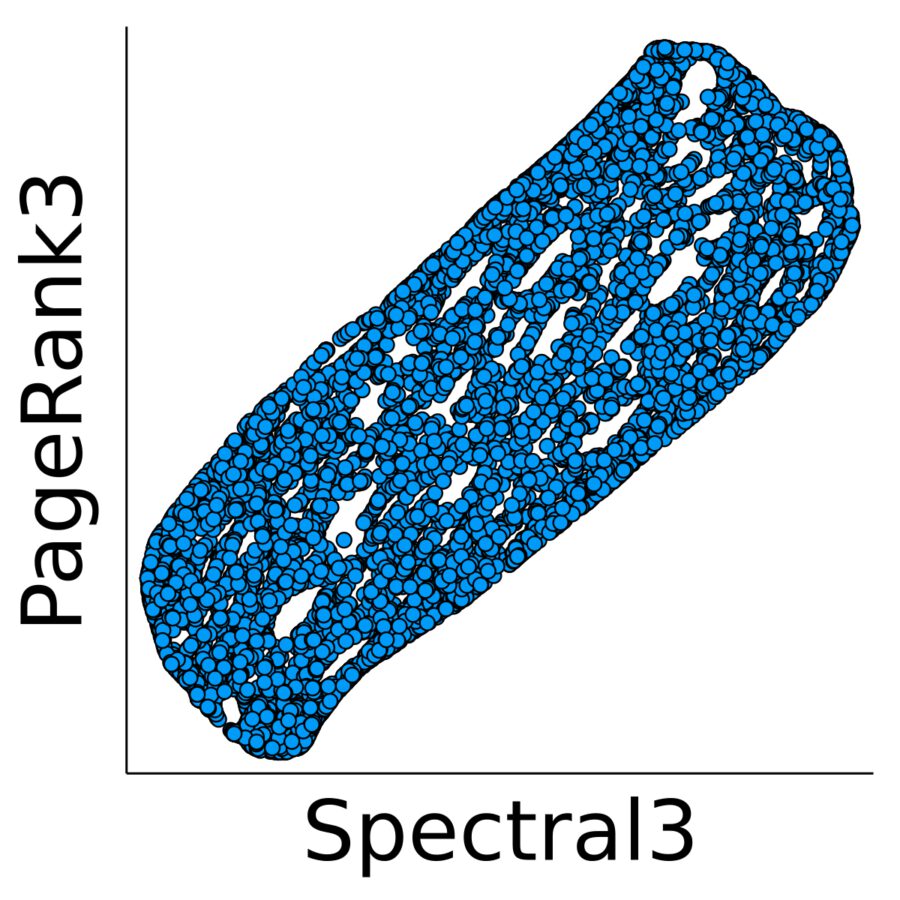}
    \caption{$\!\vu_2\!$ vs $\!\vz_2\!$ / $\!\vu_3\!$ vs $\!\vz_3\!$}
    \end{subfigure}
\centering

\bigskip 
    \smash{\raisebox{-24pt}{\llap{\rotatebox{90}{\footnotesize$\alpha=0.9999$}}}}%
    \begin{subfigure}{0.33\linewidth}
    \includegraphics[width=\linewidth]{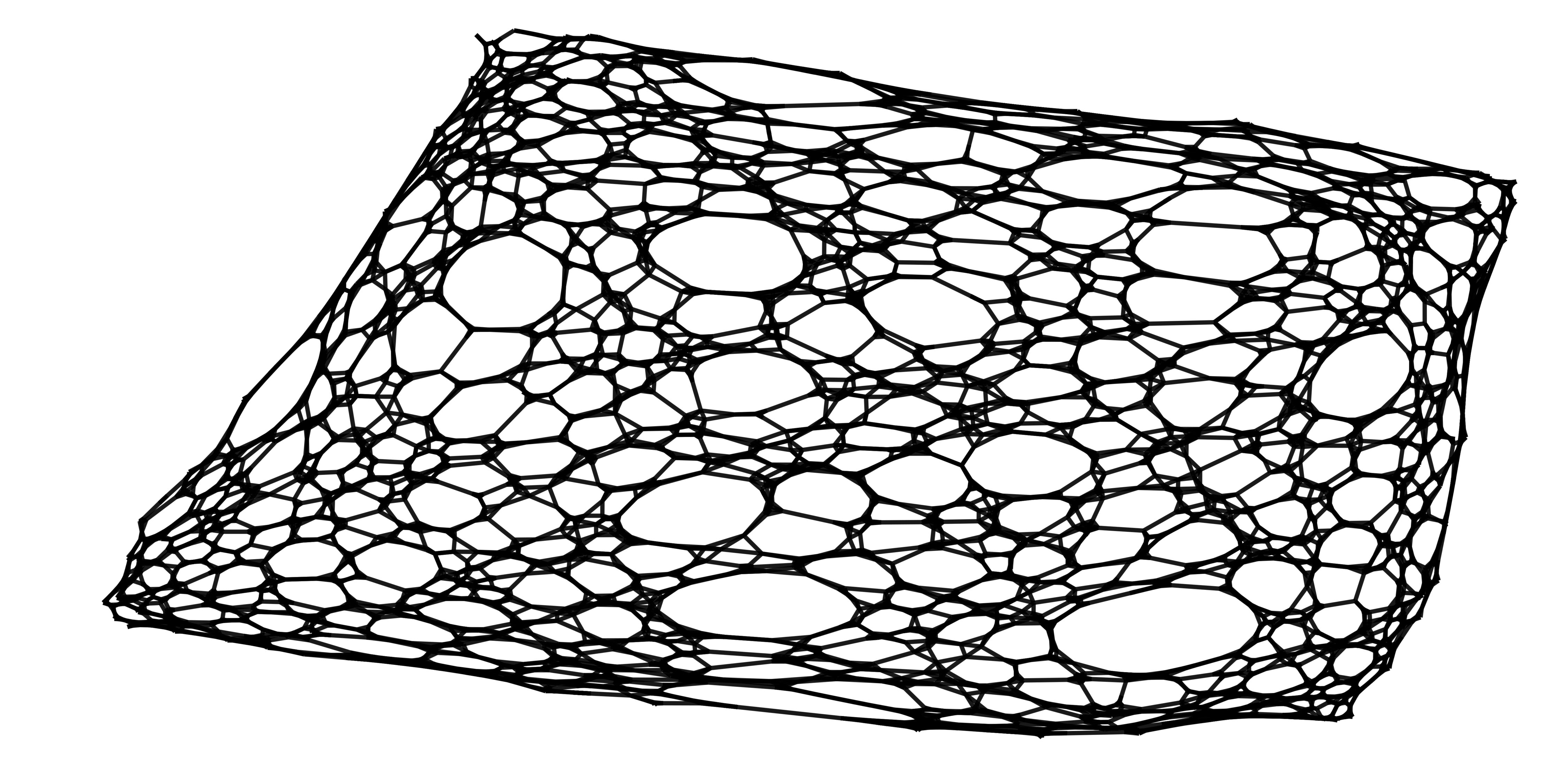}
    \caption{Log-PageRank}
    \end{subfigure}%
    \begin{subfigure}{0.33\linewidth}
    \includegraphics[width=\linewidth]{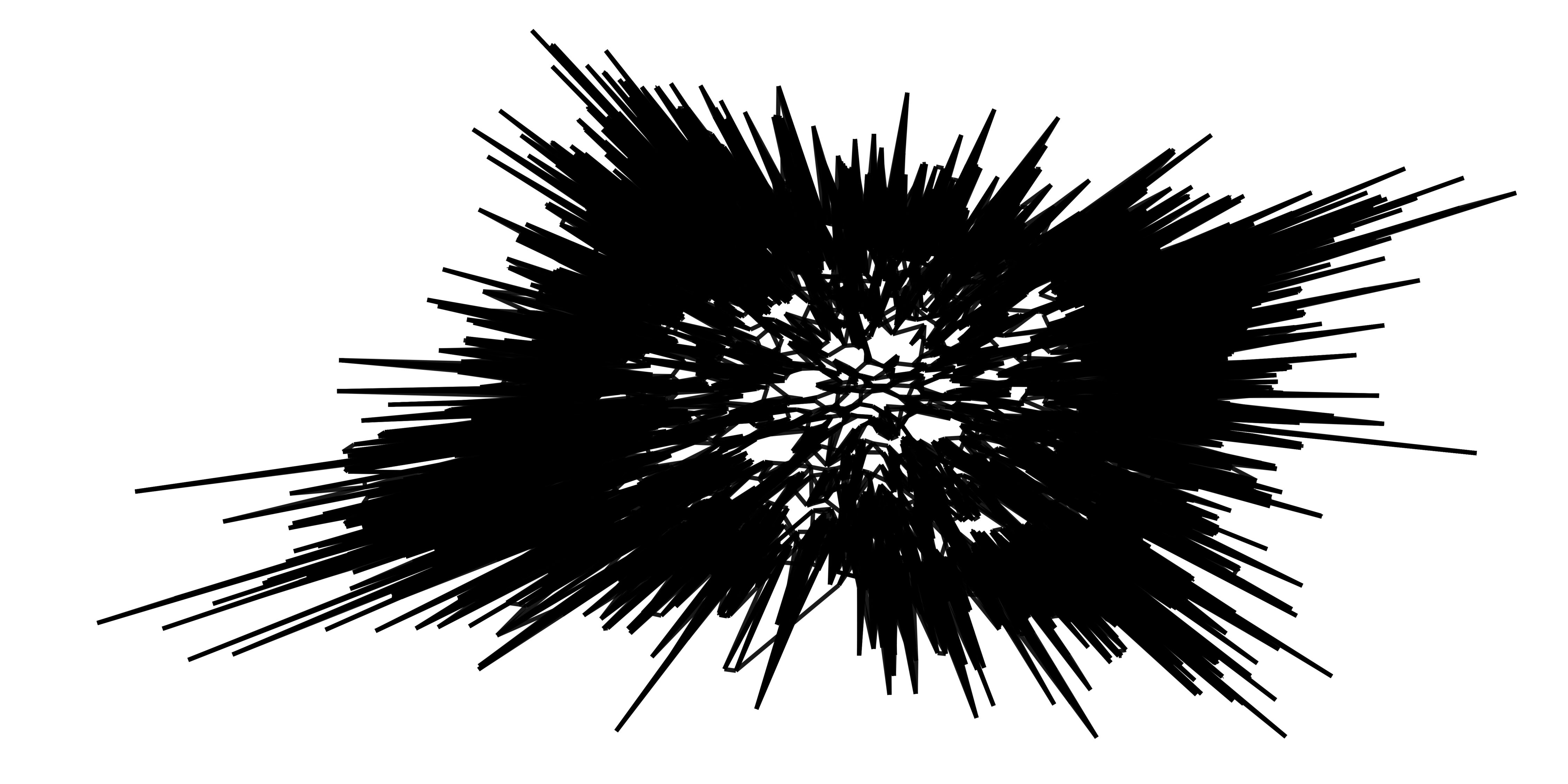}
    \caption{PageRank}
    \end{subfigure}%
    \begin{subfigure}{0.33\linewidth}
    \includegraphics[width=0.5\linewidth]{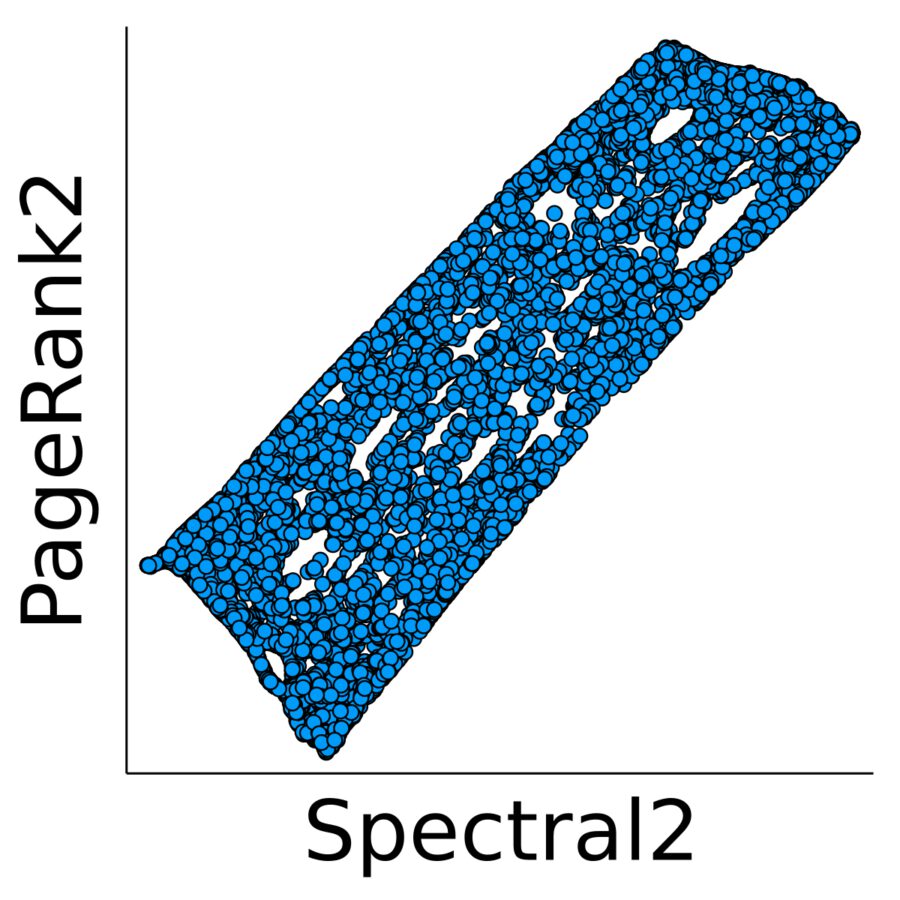}%
    \includegraphics[width=0.5\linewidth]{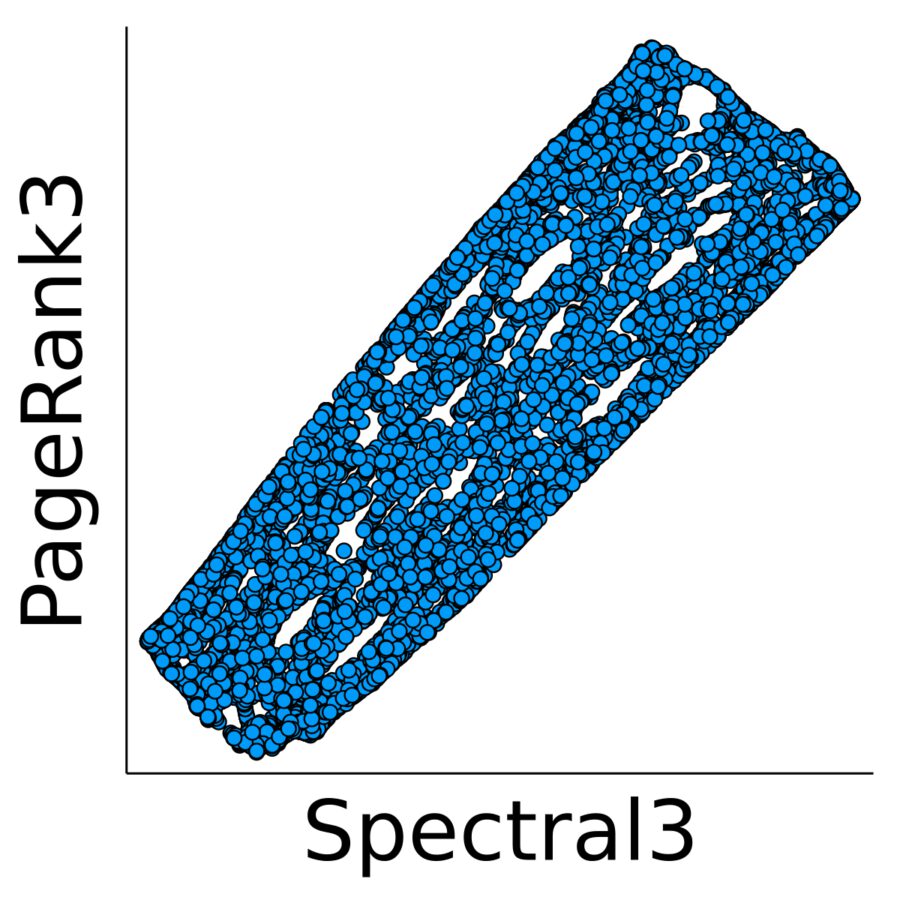}
    \caption{$\!\vu_2\!$ vs $\!\vz_2\!$ / $\!\vu_3\!$ vs $\!\vz_3\!$}
    \end{subfigure}
}

\smallskip
\hrule 
\smallskip
\centering
    \parbox{0.3\linewidth}{%
    \parbox{0.95\linewidth}{\caption{Comparison of the embedding techniques on the planted partition model with good conductance cuts. 
    This is one case where the technique does not seem to work. 
    }
    \label{fig:highercon}}
    \begin{subfigure}{\linewidth}
    \includegraphics[width=\linewidth]{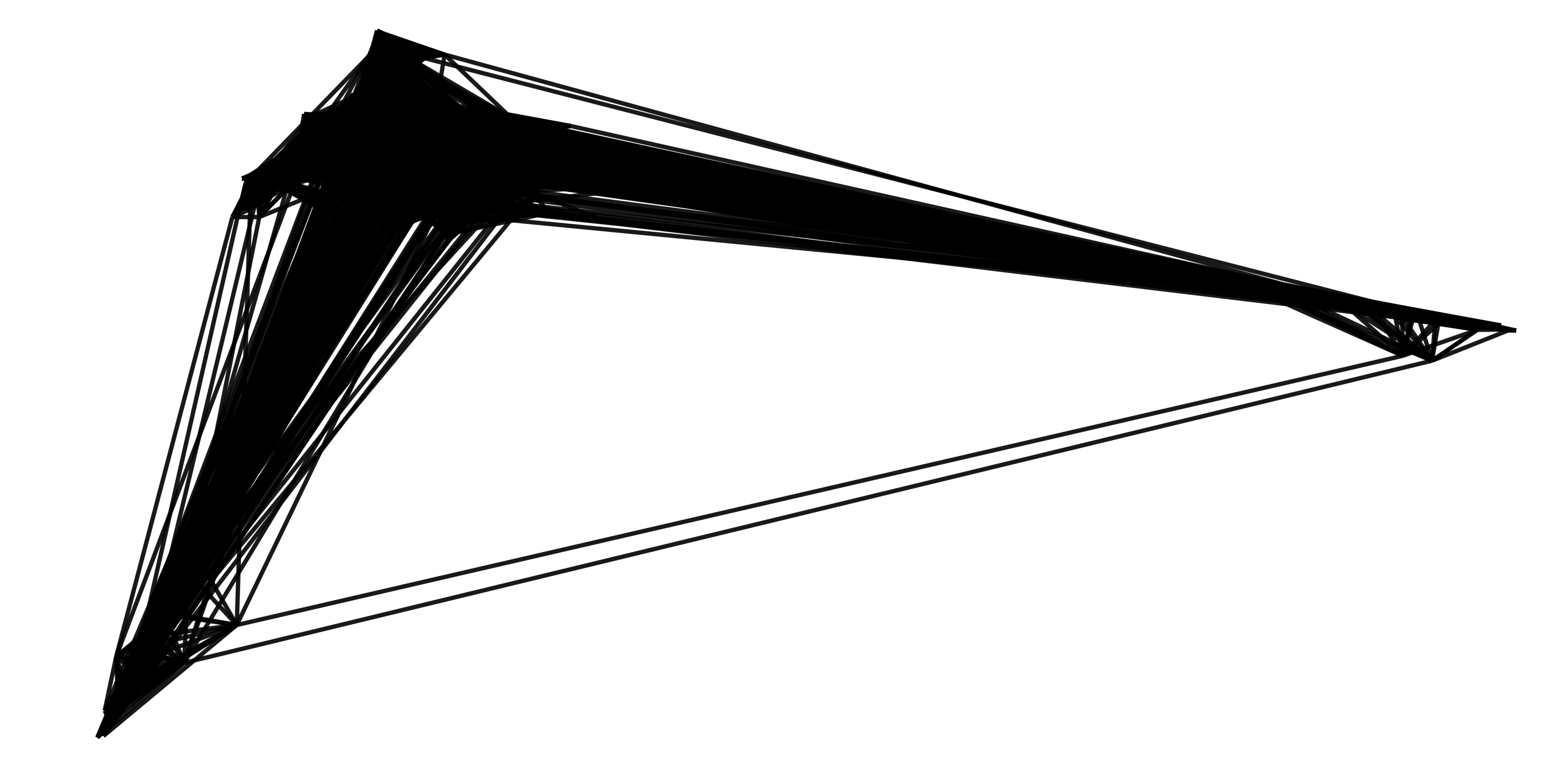}
    \caption{Spectral Embedding }
    \end{subfigure}%
    }%
    \parbox{0.7\linewidth}{%
    \smash{\raisebox{-18pt}{\llap{\rotatebox{90}{\footnotesize$\alpha=0.99$}}}}%
    \begin{subfigure}{0.33\linewidth}
    \includegraphics[width=\linewidth]{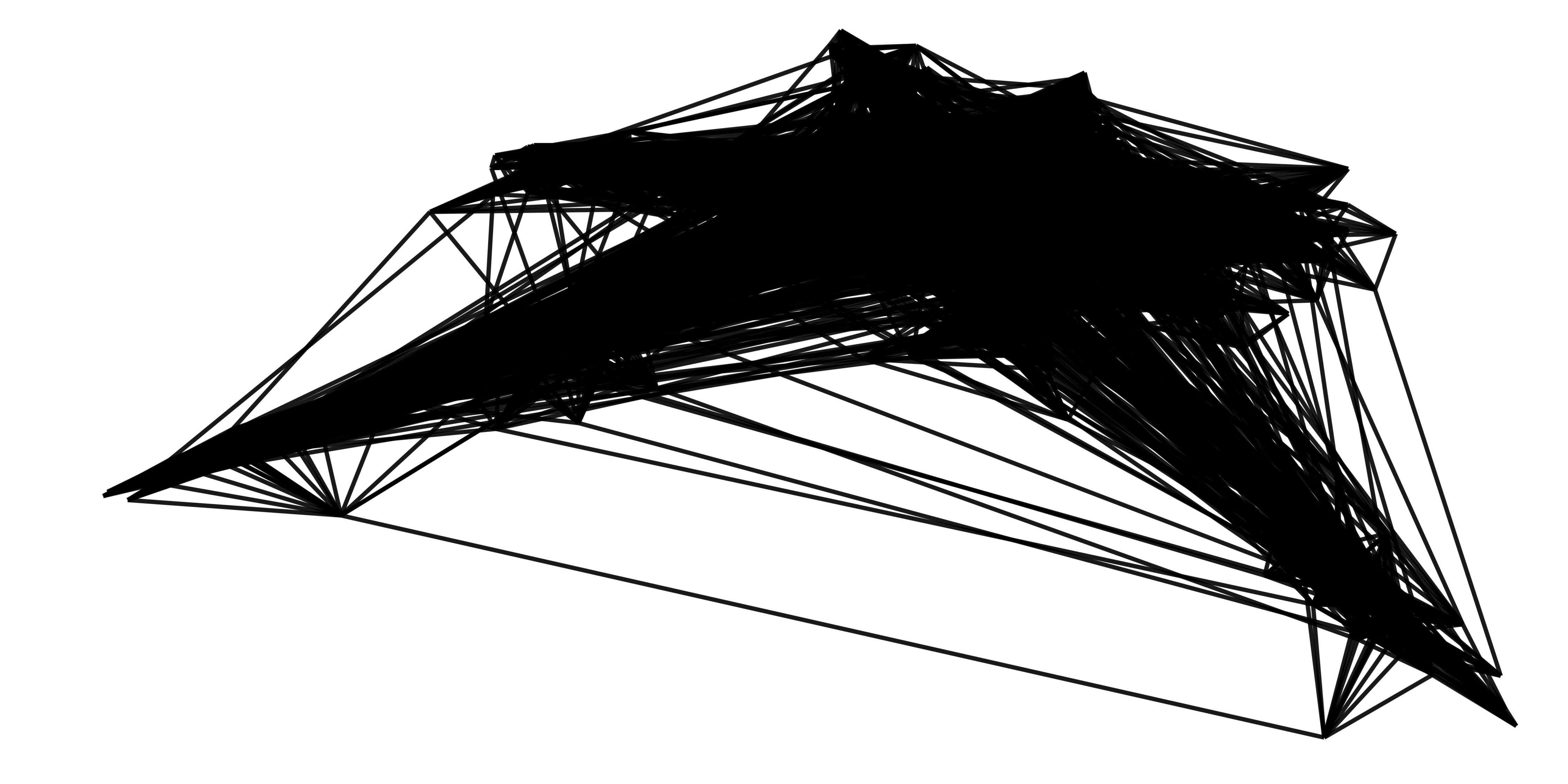}
    \caption{Log-PageRank}
    \end{subfigure}%
    \begin{subfigure}{0.33\linewidth}
    \includegraphics[width=\linewidth]{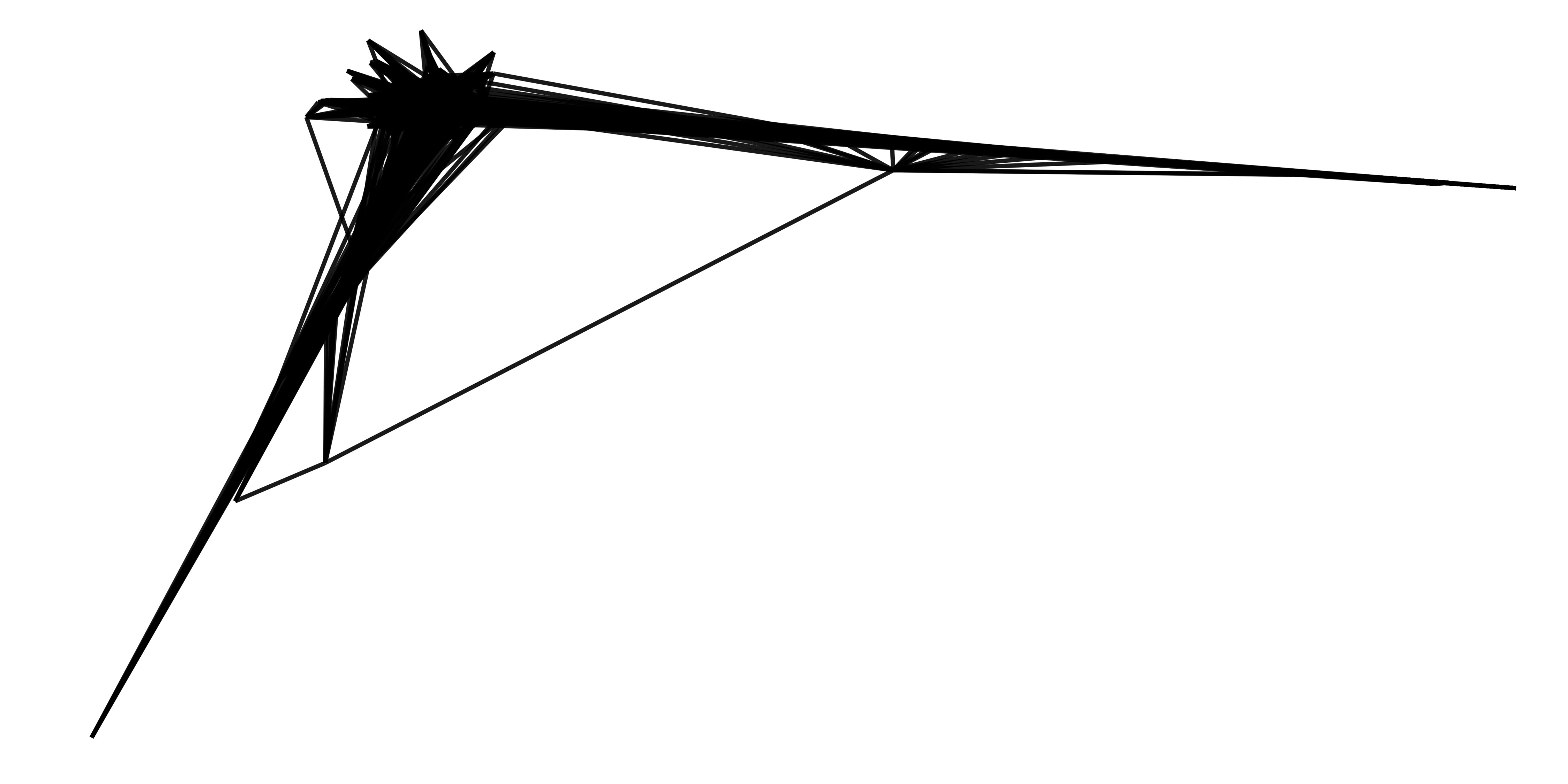}
    \caption{PageRank}
    \end{subfigure}%
    \begin{subfigure}{0.33\linewidth}
    \includegraphics[width=0.5\linewidth]{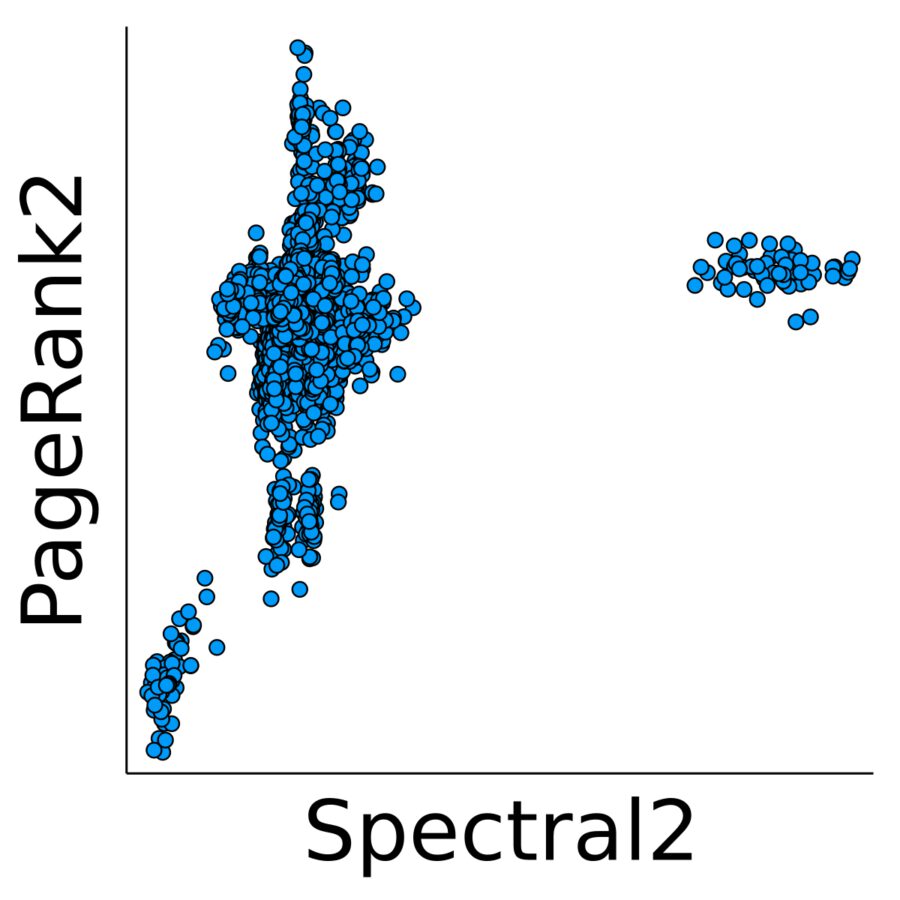}%
    \includegraphics[width=0.5\linewidth]{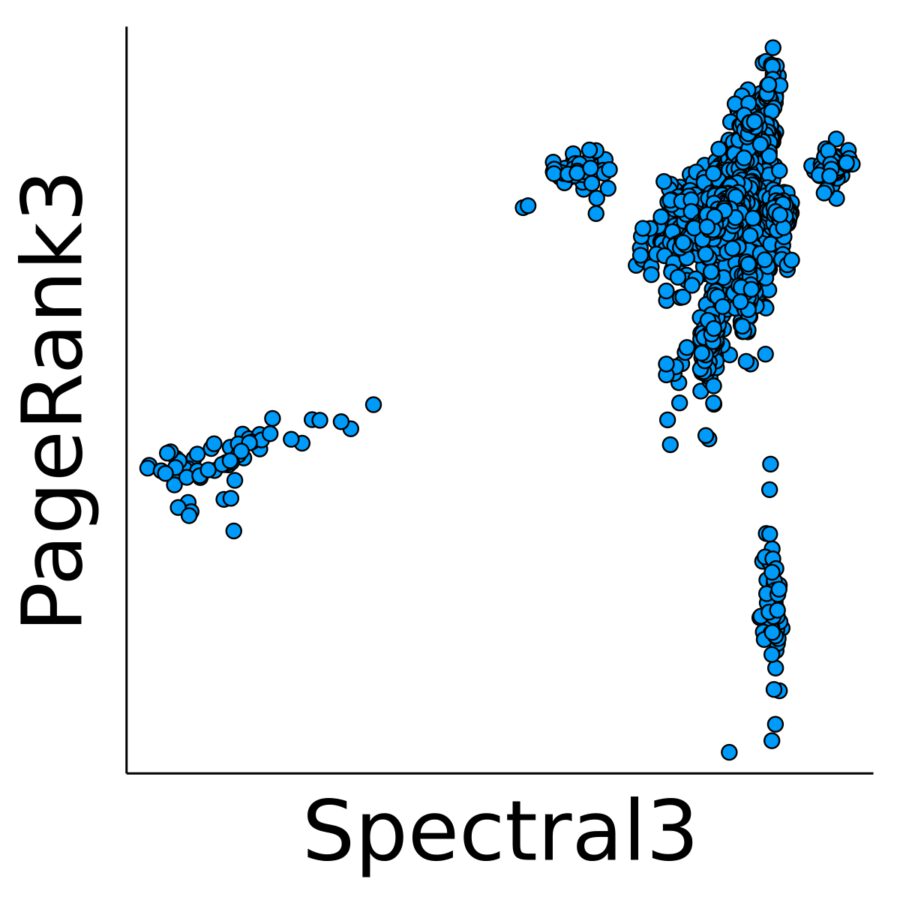}
    \caption{$\!\vu_2\!$ vs $\!\vz_2\!$ / $\!\vu_3\!$ vs $\!\vz_3\!$}
    \end{subfigure}
\centering

\bigskip 
    \smash{\raisebox{-24pt}{\llap{\rotatebox{90}{\footnotesize$\alpha=0.9999$}}}}%
    \begin{subfigure}{0.33\linewidth}
    \includegraphics[width=\linewidth]{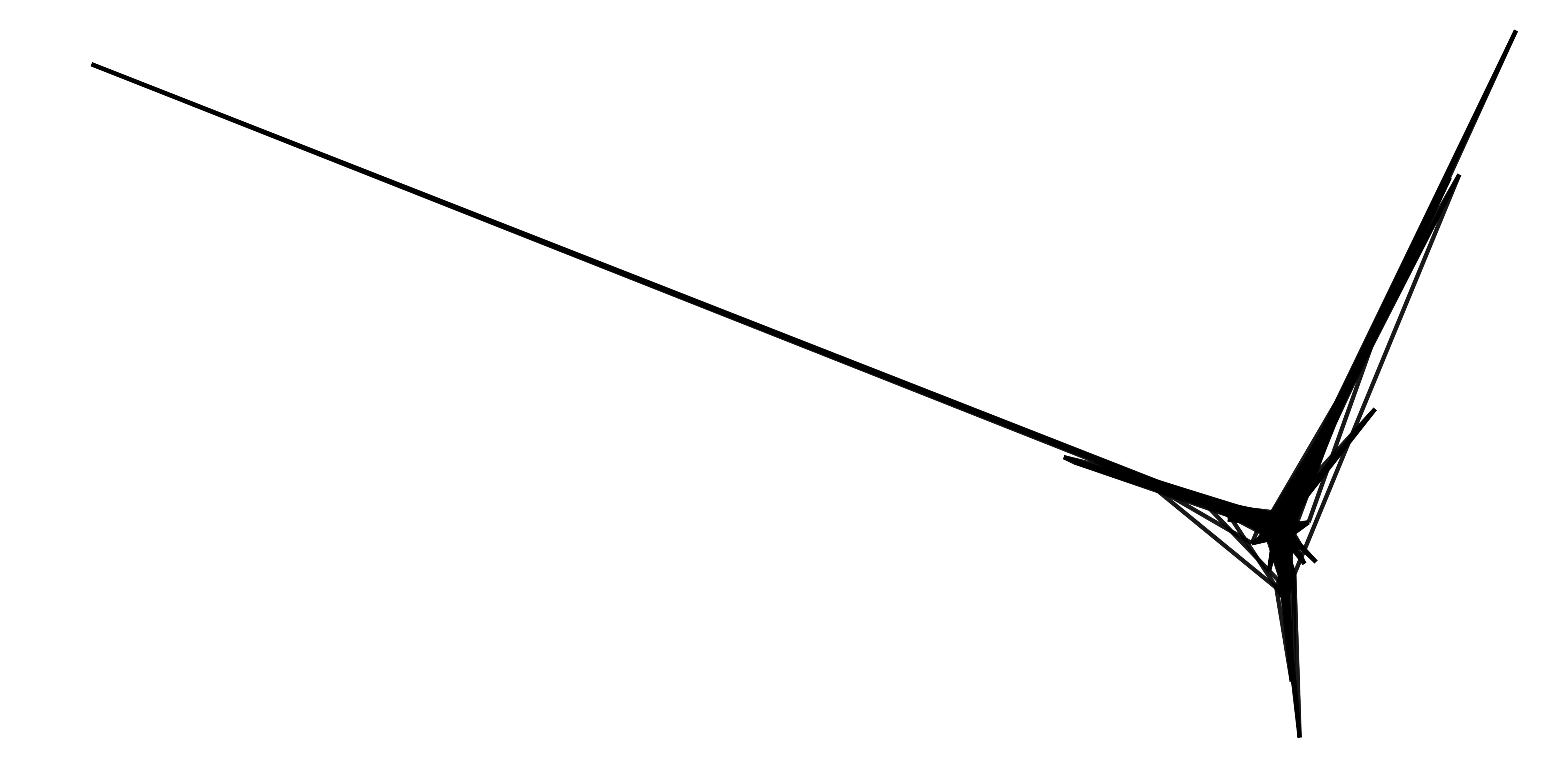}
    \caption{Log-PageRank}
    \end{subfigure}%
    \begin{subfigure}{0.33\linewidth}
    \includegraphics[width=\linewidth]{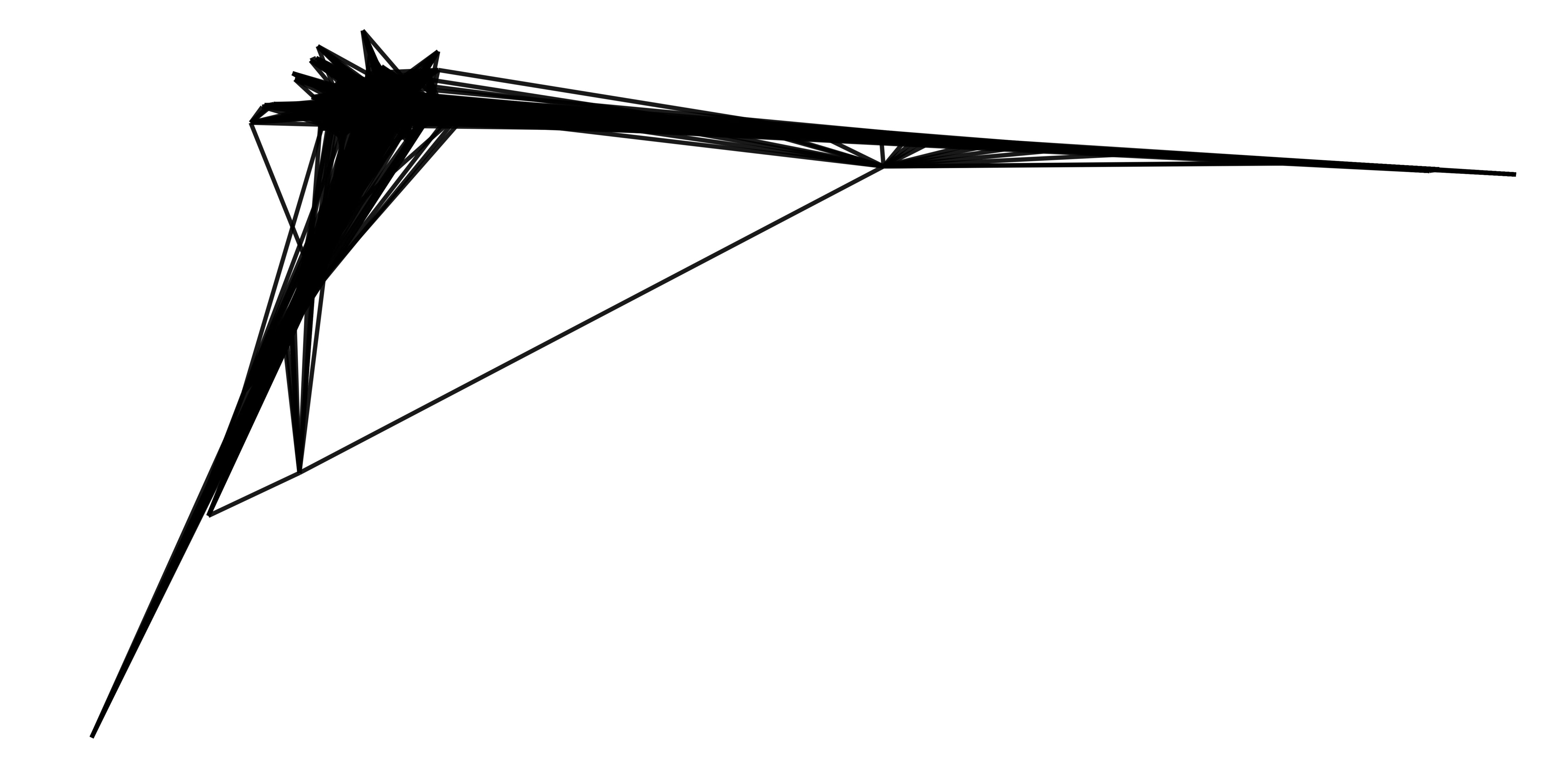}
    \caption{PageRank}
    \end{subfigure}%
    \begin{subfigure}{0.33\linewidth}
    \includegraphics[width=0.5\linewidth]{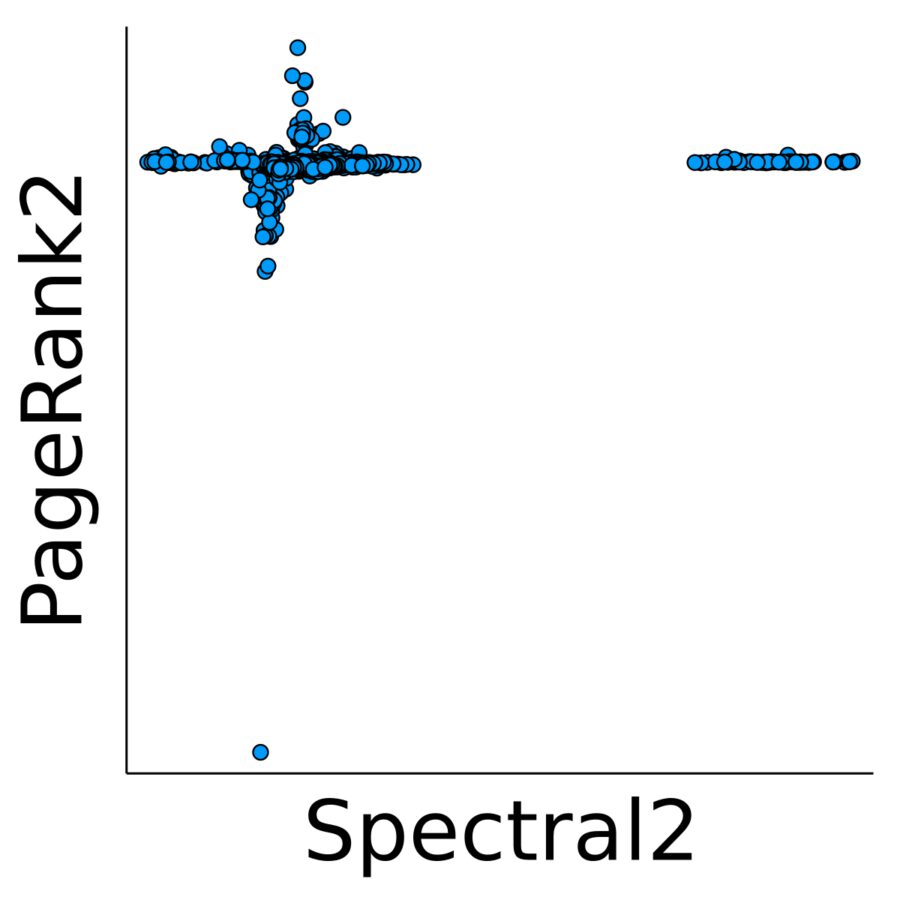}%
    \includegraphics[width=0.5\linewidth]{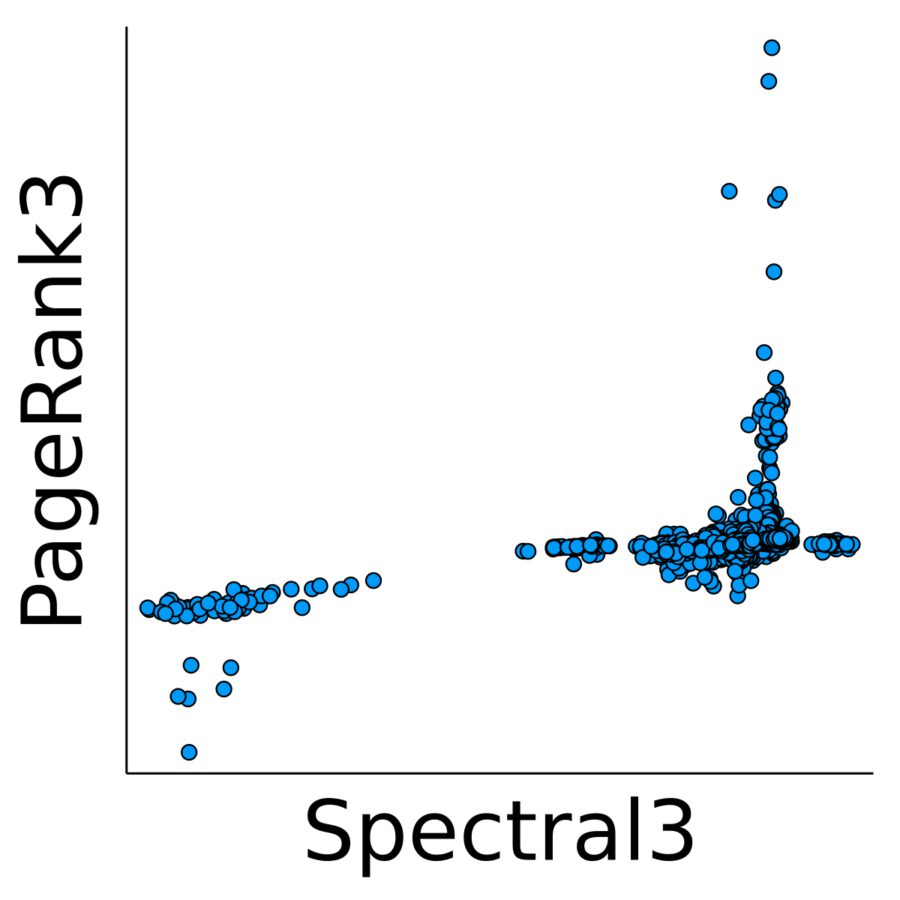}
    \caption{$\!\vu_2\!$ vs $\!\vz_2\!$ / $\!\vu_3\!$ vs $\!\vz_3\!$}
    \end{subfigure}
}

\end{fullwidth}
\end{figure}

\subsection{Embedding error variance}

We studied the dependence of this error on the number of randomly sampled nodes and which sampled nodes. We record this for log-PageRank at $\alpha = 0.99$ in Figure~\ref{tab:err_var_log}. 
This shows the distribution of errors as a density estimate, along with the max/min values (small) and the median value (big). 

As expected, there are largely minimal effects. This occurs because the sensitivity of PageRank to the seed vector, $\mathbf{v}$, is a function of $\alpha$ \cite{pprbook}. 
\begin{align*}
   \frac{d \vx }{d \vv} = (1 - \alpha) (\mI - \alpha \mP)^{-1}  
\end{align*}
which satisfies $\|\frac{d \vx}{ d \vv}\|_1 = 1$.
Further, for $\alpha \rightarrow 1$, dependence of the PageRank values on the $\mathbf{v}$ reduces. Our experiments confirm the same as the minimum, maximum and variance of error over 50 trials show negligible change. 

\section{Hypergraph Embeddings}\label{sec:hypergraph}

One driving reason for our study of the log-PageRank embedding is to support similar embedding strategies for different types of data, such as those studied in~\cite{beginng}.
In this section, we use the log-PageRank embedding technique on five hypergraphs: Yelp (\url{https://www.yelp.com/dataset}), Walmart Trips \cite{Amburg-2020-categorical} , a contact tracing network \cite{contact1, contact2}, posts on Math Overflow \cite{mathoverflow}, and a Drug Abuse network (DAWN) \cite{Amburg-2020-categorical}. 
The only modification to our algorithm is that
we replace seeded PageRank in Algorithm~\ref{alg:logPR} with the Local Quadratic PageRank (LQPR), a method 
proposed in \cite{meng}. Specifically we use LQHD with a 2-norm penalty with $\rho=0.5$ for all experiments
For the Yelp and Walmart trips network, we set $\kappa = 0.000025$ and $\gamma = 1.0$ while for the Math Overflow network, with the same sparsity factor $\kappa = 0.000025$ and $\gamma = 0.001$. For Contact Primary School and DAWN, we set $\kappa = 0.0025$ and $\gamma = 0.001$. These choices were made arbitrarily, there are small differences that result when changing them. 

\enlargethispage{\baselineskip}
Figure~\ref{fig:contact-primary-school} shows our 2d embedding on Contact Primary School dataset where
each node represents a student or a teacher, each hyperedge represents a group of people who are spatially
close at a given time. Each node of the graph is colored as a teacher or as classroom for the student.
Note that each classroom forms a cohesive group in the plot. 
The nodes for the teachers are not a good 
cluster because different teachers go to different classrooms. 
Moreover, we observe that the students from the same grade, e.g. students 
colored red (1B) and dark green (1A), share some spatial proximity, which is due to the fact
that their classrooms are close.

\begin{fullwidthfigure}[t]
    \centering
    \begin{subfigure}[t]{0.49\linewidth}
        \centering 
        \includegraphics[width=\linewidth]{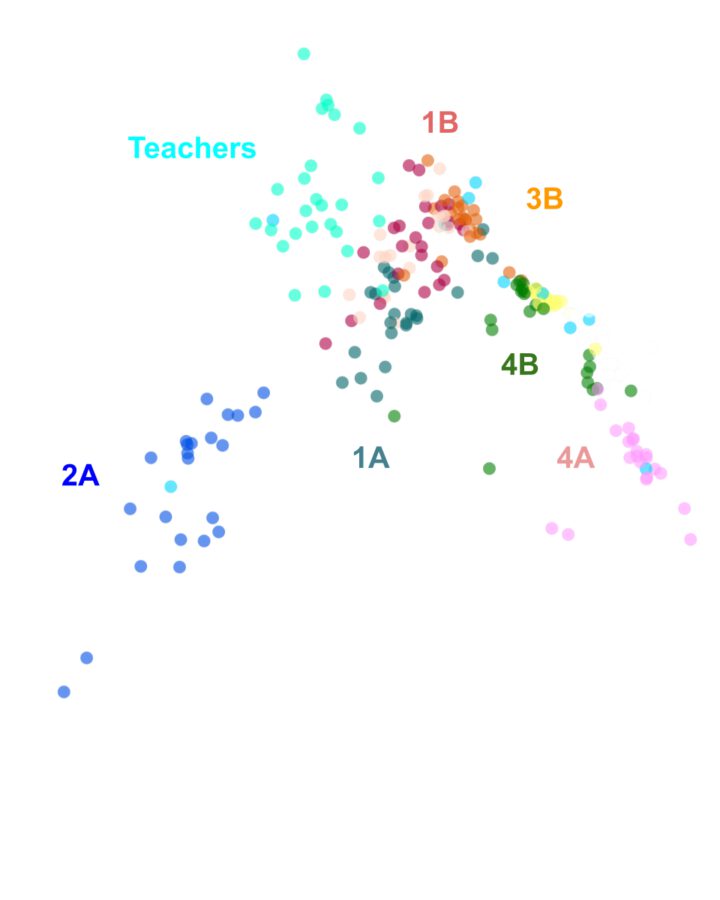}
        \caption{Contact Primary School}
        \sublabel{fig:contact-primary-school}
    \end{subfigure}
    \hfill
    \begin{subfigure}[t]{0.49\linewidth}
        \centering
        \includegraphics[width=\linewidth]{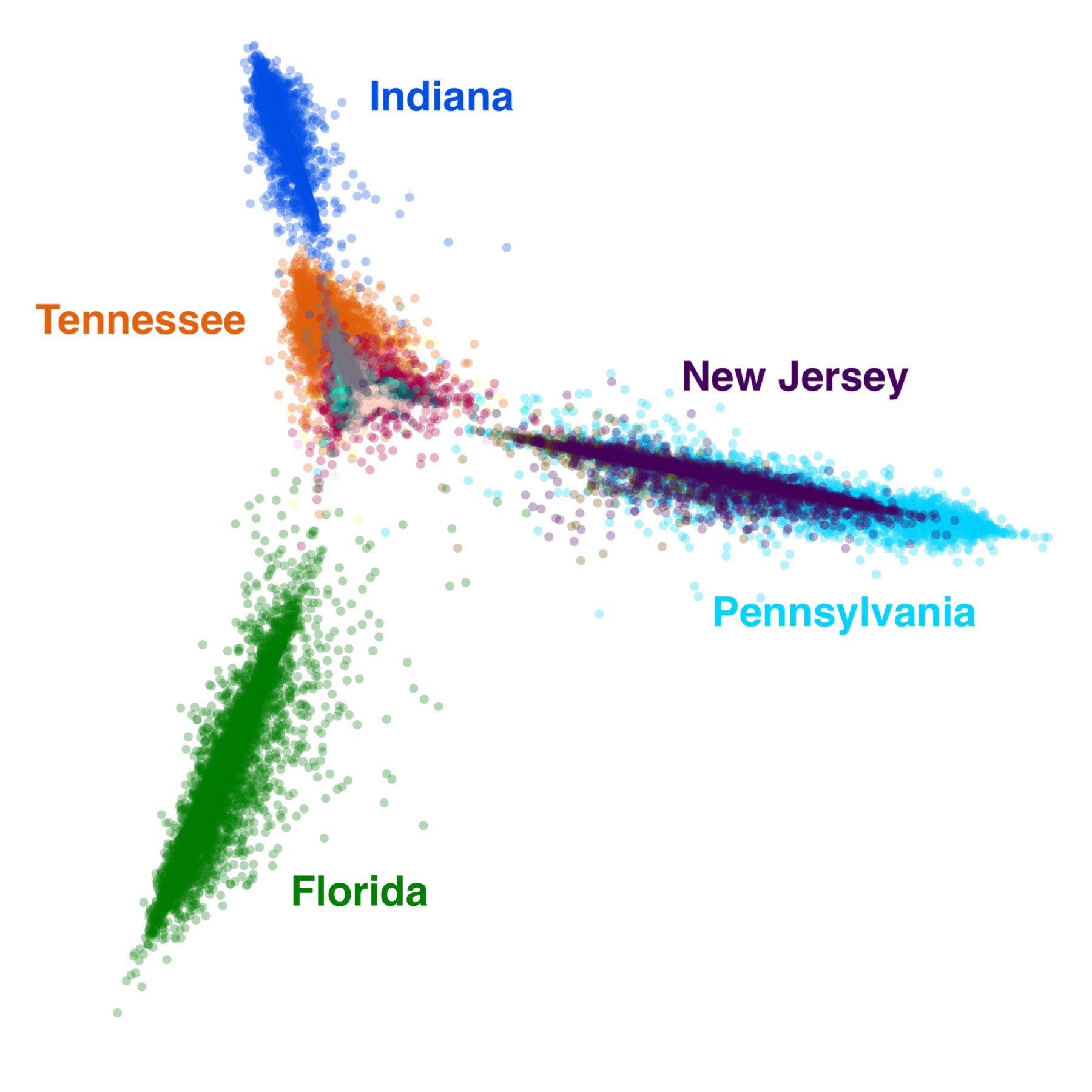}
        \caption{Yelp Restaurants\sublabel{fig:yelp}}
    \end{subfigure}
    \caption{Log-PageRank embedding of hypergraphs. The Contact Primary School dataset  has 242 nodes and 12704 hyperdges. Nodes are colored by classroom and teachers, which form cohesive groups due to the contact structure. The Yelp Restaurant dataset
    has 52260 nodes and 597261 hyperedges. Nodes are colored by one of 14 states used for analysis, which show clear geographic relationships.}
    \label{fig:hypergraph-geo}
\end{fullwidthfigure}

Figure~\ref{fig:yelp} shows our embedding on Yelp Review data. Following \citet{veldt20}, we build
one hypergraph with each restaurant being a node and each user being a hyperedge. We show the state associated with each location as the color. We can clearly see that our embedding captures the geographic information of the underlying hypergraph. For example, the nodes
labeled dark blue are those restaurants from state Indiana, which are close to the orange nodes
from state Tennessee. Also the green nodes from state Florida are quite well-separated from nodes
with other colors, which is due to the fact that none of other 13 states we plot is close 
to Florida.   

In addition, we show log-PageRank embeddings of three other hypergraphs in Figure~\ref{fig:dawn}, \ref{fig:walmart} and \ref{fig:mathoverflow}. We were unable to identify obvious relationships between these embeddings and the existing groups, which means the embeddings likely show a different type of structure.  
The promising results on all the datasets above show that our simple algorithm is capable of generating good embeddings even on higher order graphs.

\begin{fullwidthfigure}[t]
\centering
    \begin{subfigure}[t]{.3\linewidth}
        \centering
        \includegraphics[width=\linewidth]{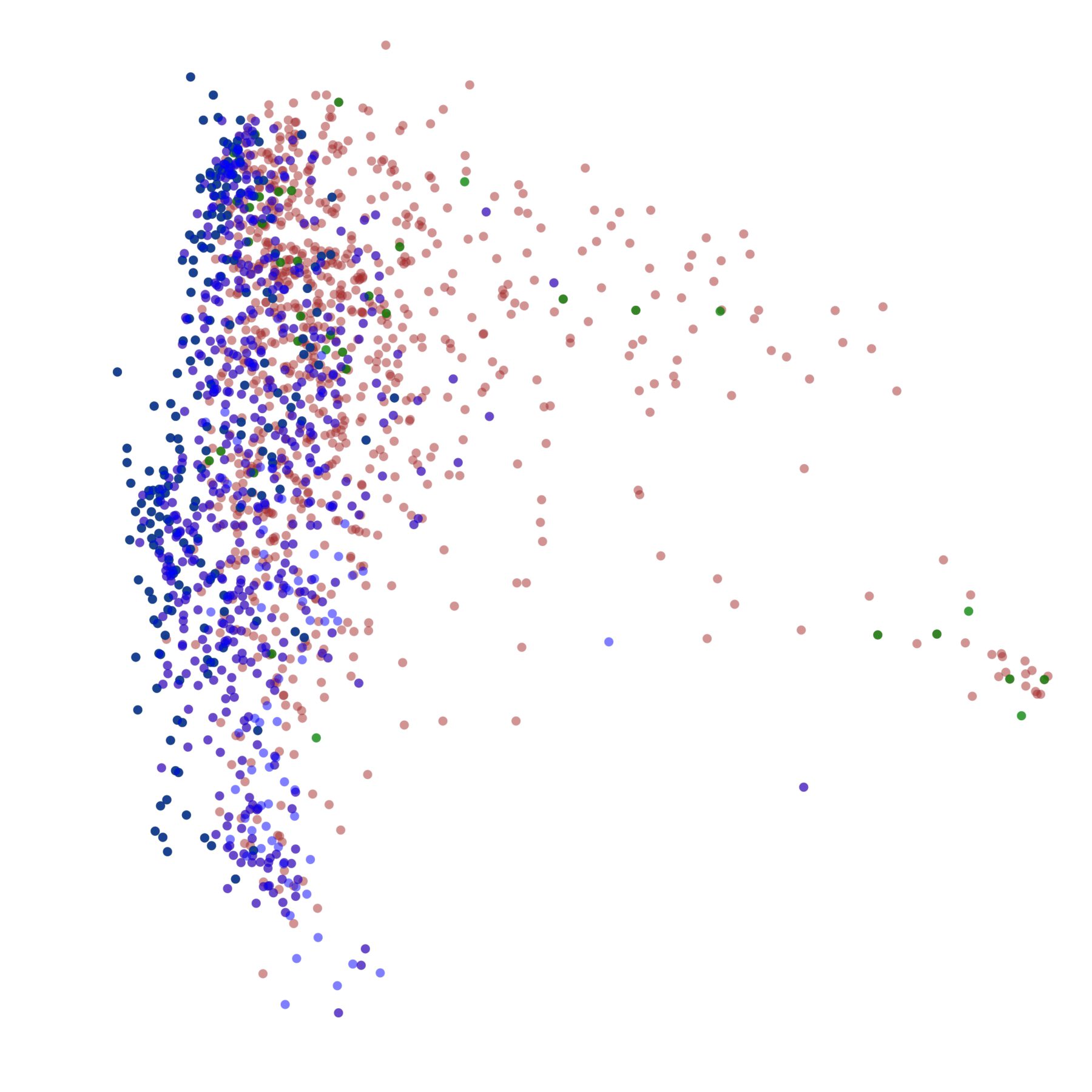}
        \caption{DAWN}
        \sublabel{fig:dawn}
    \end{subfigure}
    \hfill
    \begin{subfigure}[t]{.3\linewidth}
        \centering
        \includegraphics[width=\linewidth]{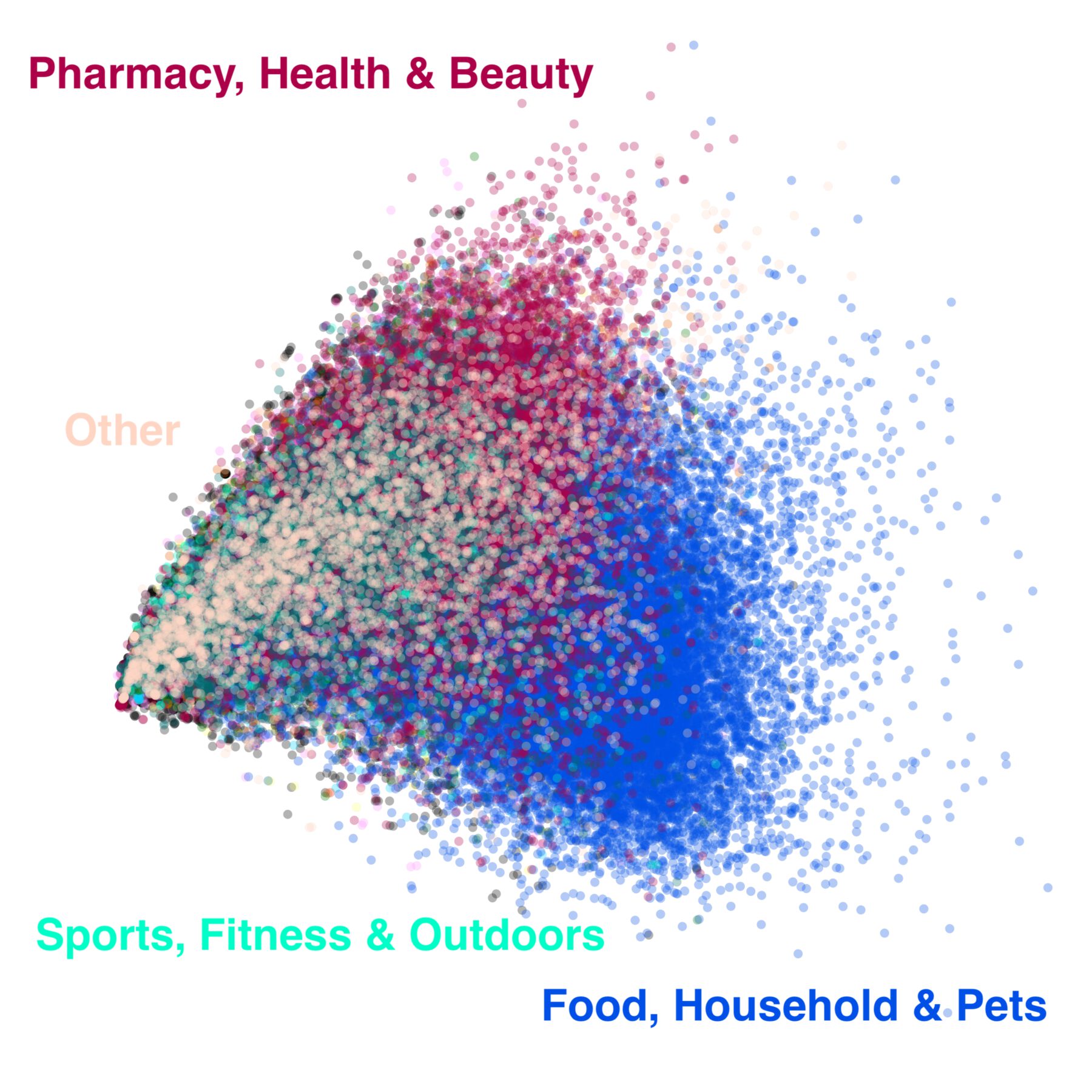}
        \caption{Walmart Trips}
        \sublabel{fig:walmart}
    \end{subfigure}
        \hfill
    \begin{subfigure}[t]{.3\linewidth}
        \centering
        \includegraphics[width=\linewidth]{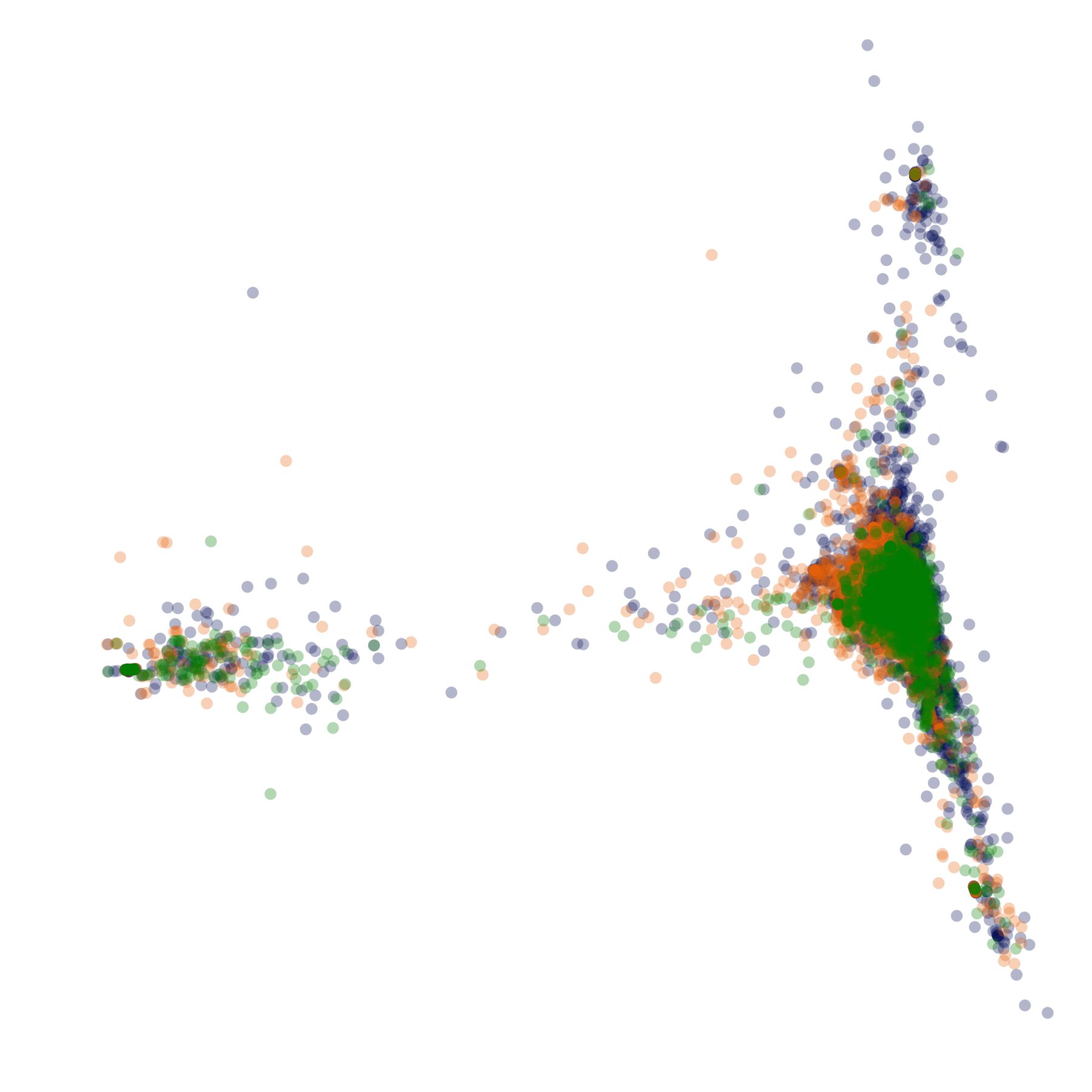}
        \caption{Math Overflow}
        \sublabel{fig:mathoverflow}
    \end{subfigure}
    \caption{Log-PageRank embeddings of the (a) DAWN dataset \cite{Amburg-2020-categorical} with 2109 nodes (drugs) and 
    87104 hyperedges where each hyperedge is an individual and the consisting nodes are 
    the drugs consumed by them; the plot shows 3 out 10 labels (b) the Walmart dataset \cite{Amburg-2020-categorical} with 88860 nodes and 69906 hyperedges where each node is a product and each
    hyperedge consists of products purchased in one trip to Walmart and (c) Math Overflow \cite{mathoverflow} dataset with 73851 nodes, 5446 hyperedges, where each hyperedge multiple labels associated with a question;
    Although there is structure evident in the plots, it does not strongly correlate with the known labels on the data (some of the plotting makes the structure look more present than it is).}
    \label{fig:hypergraph-bad}
\end{fullwidthfigure}

\section{Related work, Conclusion, and Future directions}

The key finding of this paper is the elementwise log of a matrix of seeded PageRank vector approximates the spectral embedding of the Laplacian in some scenarios.  This analysis can be transparently mapped to new scenarios such as hypergraphs given a PageRank-like primitive. This greatly simplifies the scenario compared with nonlinear spectral methods on hypergraphs~\cite{FKA21, FAK21, FH21, QFA17}. 

The idea of using the log of a PageRank vector originated in Google's initial use of these for their PageRank scores. That said, the elementwise log emerged in other scenarios as well. For example, \citet{mf1} detail a similar analysis between SkipGram \cite{sg}, a popular representation learning framework, and the SVD of the element wise log of a probability transition matrix developed from the data. Following that, multiple papers \cite{inf_walk, netmf} showed relationships between embedding techniques \cite{node2vec, pte, line, deepwalk} and asymptotic matrix expressions. 

We believe our framework offers a successful technique for structural embedding and opens up some nontrivial research problems.
    Our code to compute these embeddings for these examples is available: \url{https://github.com/dishashur/log-pagerank}.
For example, although we provide an approximate analysis as to why the log function and PageRank-like matrices converge to the spectral embedding and hence generated meaningful representation, there does not yet exist a quantifiable expression between the strength of this relation and the elements of the structure, such as its conductance, sparsity, degree distribution. 
We believe this work lays the foundation for a reliable structural representation and the generalizability of this technique offers ample ground for new results.

\newpage
\begin{fullwidth}

\begin{thebibliography}{53}
\providecommand{\natexlab}[1]{#1}
\providecommand{\bibextraformatting}{\relax}
\bibextraformatting

\bibitem[\protect\citeauthoryear{Amburg et~al.}{2020}]{Amburg-2020-categorical}
I.~\textsc{Amburg}, N.~\textsc{Veldt}, and A.~R. \textsc{Benson}.
\newblock \emph{Clustering in graphs and hypergraphs with categorical edge
  labels}.
\newblock In \emph{Proceedings of the Web Conference}. 2020.

\bibitem[\protect\citeauthoryear{Andersen et~al.}{2006}]{andersen2006}
R.~\textsc{Andersen}, F.~\textsc{Chung}, and K.~\textsc{Lang}.
\newblock \emph{Local graph partitioning using pagerank vectors}.
\newblock In \emph{2006 47th Annual IEEE Symposium on Foundations of Computer
  Science (FOCS'06)}, pp. 475--486. 2006.

\bibitem[\protect\citeauthoryear{Bar-Yossef and
  Mashiach}{2008}]{Bar-Yossef2008-reverse-PageRank}
Z.~\textsc{Bar-Yossef} and L.-T. \textsc{Mashiach}.
\newblock \href{http://dx.doi.org/10.1145/1458082.1458122}{\emph{Local
  approximation of {PageRank} and reverse {PageRank}}}.
\newblock In \emph{CIKM '08: Proceeding of the 17th ACM conference on
  Information and knowledge management}, pp. 279--288. 2008.
\newblock \href {http://dx.doi.org/10.1145/1458082.1458122}
  {\normalcolor\path{doi:10.1145/1458082.1458122}}.

\bibitem[\protect\citeauthoryear{Becchetti et~al.}{2008}]{becchetti2008-spam}
L.~\textsc{Becchetti}, C.~\textsc{Castillo}, D.~\textsc{Donato},
  R.~\textsc{Baeza-Yates}, and S.~\textsc{Leonardi}.
\newblock \href{http://dx.doi.org/10.1145/1326561.1326563}{\emph{Link analysis
  for web spam detection}}.
\newblock ACM Trans. Web, 2~(1), pp. 1--42, 2008.
\newblock \href {http://dx.doi.org/10.1145/1326561.1326563}
  {\normalcolor\path{doi:10.1145/1326561.1326563}}.

\bibitem[\protect\citeauthoryear{Benson et~al.}{2018}]{contact1}
A.~R. \textsc{Benson}, R.~\textsc{Abebe}, M.~T. \textsc{Schaub},
  A.~\textsc{Jadbabaie}, and J.~\textsc{Kleinberg}.
\newblock \href{http://dx.doi.org/10.1073/pnas.1800683115}{\emph{Simplicial
  closure and higher-order link prediction}}.
\newblock Proceedings of the National Academy of Sciences, 2018.
\newblock \href {http://dx.doi.org/10.1073/pnas.1800683115}
  {\normalcolor\path{doi:10.1073/pnas.1800683115}}.

\bibitem[\protect\citeauthoryear{Bern et~al.}{1994}]{tapir}
M.~\textsc{Bern}, S.~\textsc{Mitchell}, and J.~\textsc{Ruppert}.
\newblock \href{http://dx.doi.org/10.1145/177424.177974}{\emph{Linear-size
  nonobtuse triangulation of polygons}}.
\newblock In \emph{Proceedings of the Tenth Annual Symposium on Computational
  Geometry}, p. 221–230. 1994.
\newblock \href {http://dx.doi.org/10.1145/177424.177974}
  {\normalcolor\path{doi:10.1145/177424.177974}}.

\bibitem[\protect\citeauthoryear{Brin and Page}{1998}]{origpr}
S.~\textsc{Brin} and L.~\textsc{Page}.
\newblock
  \href{http://dx.doi.org/https://doi.org/10.1016/S0169-7552(98)00110-X}{\emph{The
  anatomy of a large-scale hypertextual web search engine}}.
\newblock Computer Networks and ISDN Systems, 30~(1), pp. 107--117, 1998.
\newblock Proceedings of the Seventh International World Wide Web Conference.
\newblock \href
  {http://dx.doi.org/https://doi.org/10.1016/S0169-7552(98)00110-X}
  {\normalcolor\path{doi:https://doi.org/10.1016/S0169-7552(98)00110-X}}.

\bibitem[\protect\citeauthoryear{Carletti et~al.}{2020}]{rw1}
T.~\textsc{Carletti}, F.~\textsc{Battiston}, G.~\textsc{Cencetti}, and
  D.~\textsc{Fanelli}.
\newblock \href{http://dx.doi.org/10.1103/PhysRevE.101.022308}{\emph{Random
  walks on hypergraphs}}.
\newblock Phys. Rev. E, 101, p. 022308, 2020.
\newblock \href {http://dx.doi.org/10.1103/PhysRevE.101.022308}
  {\normalcolor\path{doi:10.1103/PhysRevE.101.022308}}.

\bibitem[\protect\citeauthoryear{Chanpuriya and Musco}{2020}]{inf_walk}
S.~\textsc{Chanpuriya} and C.~\textsc{Musco}.
\newblock \href{https://doi.org/10.1145/3394486.3403185}{\emph{InfiniteWalk:
  Deep Network Embeddings as Laplacian Embeddings with a Nonlinearity}}, p.
  1325–1333.
\newblock Association for Computing Machinery, New York, NY, USA, 2020.

\bibitem[\protect\citeauthoryear{Chung}{2007}]{hkrnl}
F.~\textsc{Chung}.
\newblock \href{http://dx.doi.org/10.1073/pnas.0708838104}{\emph{The heat
  kernel as the pagerank of a graph}}.
\newblock Proceedings of the National Academy of Sciences, 104~(50), pp.
  19735--19740, 2007.
\newblock \href
  {http://arxiv.org/abs/https://www.pnas.org/content/104/50/19735.full.pdf}
  {\normalcolor\path{arXiv:https://www.pnas.org/content/104/50/19735.full.pdf}},
  \href {http://dx.doi.org/10.1073/pnas.0708838104}
  {\normalcolor\path{doi:10.1073/pnas.0708838104}}.

\bibitem[\protect\citeauthoryear{Chung et~al.}{2011}]{dev}
F.~\textsc{Chung}, A.~\textsc{Tsiatas}, and W.~\textsc{Xu}.
\newblock \emph{Dirichlet pagerank and trust-based ranking algorithms}.
\newblock In \emph{Algorithms and Models for the Web Graph}, pp. 103--114.
  2011.

\bibitem[\protect\citeauthoryear{Chung}{1992}]{Chung-1992-book}
F.~R.~L. \textsc{Chung}.
\newblock \emph{Spectral Graph Theory}, American Mathematical Society, 1992.

\bibitem[\protect\citeauthoryear{Donnat et~al.}{2018}]{gwav}
C.~\textsc{Donnat}, M.~\textsc{Zitnik}, D.~\textsc{Hallac}, and
  J.~\textsc{Leskovec}.
\newblock \href{http://dx.doi.org/10.1145/3219819.3220025}{\emph{Learning
  structural node embeddings via diffusion wavelets}}.
\newblock In \emph{Proceedings of the 24th ACM SIGKDD International Conference
  on Knowledge Discovery and Data Mining}, p. 1320–1329. 2018.
\newblock \href {http://dx.doi.org/10.1145/3219819.3220025}
  {\normalcolor\path{doi:10.1145/3219819.3220025}}.

\bibitem[\protect\citeauthoryear{Fountoulakis et~al.}{2020}]{beginng}
K.~\textsc{Fountoulakis}, M.~\textsc{Liu}, D.~F. \textsc{Gleich}, and M.~W.
  \textsc{Mahoney}.
\newblock \emph{Flow-based algorithms for improving clusters: A unifying
  framework, software, and performance}.
\newblock 2020.
\newblock \href {http://arxiv.org/abs/2004.09608}
  {\normalcolor\path{arXiv:2004.09608}}.

\bibitem[\protect\citeauthoryear{Frobenius}{1912}]{fro12}
G.~\textsc{Frobenius}.
\newblock \emph{{\"U}ber matrizen aus nicht negativen elementen}.
\newblock K{\"o}nigliche Akademie der Wissenschaften Sitzungsber, K{\"o}n, pp.
  456--477, 1912.

\bibitem[\protect\citeauthoryear{Gleich}{2009}]{gleich-thesis}
D.~F. \textsc{Gleich}.
\newblock
  \href{http://www.stanford.edu/group/SOL/dissertations/pagerank-sensitivity-thesis-online.pdf}{\emph{Models
  and Algorithms for {PageRank} Sensitivity}}.
\newblock Ph.D. thesis, Stanford University, 2009.

\bibitem[\protect\citeauthoryear{Gleich}{2015}]{ppr_basics}
---{}---{}---.
\newblock \href{http://dx.doi.org/10.1137/140976649}{\emph{Pagerank beyond the
  web}}.
\newblock SIAM Rev., 57~(3), p. 321–363, 2015.
\newblock \href {http://dx.doi.org/10.1137/140976649}
  {\normalcolor\path{doi:10.1137/140976649}}.

\bibitem[\protect\citeauthoryear{Grover and Leskovec}{2016}]{node2vec}
A.~\textsc{Grover} and J.~\textsc{Leskovec}.
\newblock \href{http://dx.doi.org/10.1145/2939672.2939754}{\emph{Node2vec:
  Scalable feature learning for networks}}.
\newblock In \emph{Proceedings of the 22nd ACM SIGKDD International Conference
  on Knowledge Discovery and Data Mining}, p. 855–864. 2016.
\newblock \href {http://dx.doi.org/10.1145/2939672.2939754}
  {\normalcolor\path{doi:10.1145/2939672.2939754}}.

\bibitem[\protect\citeauthoryear{Hall}{1970}]{hall}
K.~M. \textsc{Hall}.
\newblock \emph{An r-dimensional quadratic placement algorithm}.
\newblock Management science, 17~(3), pp. 219--229, 1970.

\bibitem[\protect\citeauthoryear{Hammond et~al.}{2011}]{hammond}
D.~K. \textsc{Hammond}, P.~\textsc{Vandergheynst}, and R.~\textsc{Gribonval}.
\newblock
  \href{http://dx.doi.org/https://doi.org/10.1016/j.acha.2010.04.005}{\emph{Wavelets
  on graphs via spectral graph theory}}.
\newblock Applied and Computational Harmonic Analysis, 30~(2), pp. 129--150,
  2011.
\newblock \href {http://dx.doi.org/https://doi.org/10.1016/j.acha.2010.04.005}
  {\normalcolor\path{doi:https://doi.org/10.1016/j.acha.2010.04.005}}.

\bibitem[\protect\citeauthoryear{Huang et~al.}{2021}]{combining}
Q.~\textsc{Huang}, H.~\textsc{He}, A.~\textsc{Singh}, S.-N. \textsc{Lim}, and
  A.~\textsc{Benson}.
\newblock \href{https://openreview.net/forum?id=8E1-f3VhX1o}{\emph{Combining
  label propagation and simple models out-performs graph neural networks}}.
\newblock In \emph{International Conference on Learning Representations}. 2021.

\bibitem[\protect\citeauthoryear{Klicpera et~al.}{2019}]{gcnpr}
J.~\textsc{Klicpera}, A.~\textsc{Bojchevski}, and S.~\textsc{Günnemann}.
\newblock \href{https://openreview.net/forum?id=H1gL-2A9Ym}{\emph{Combining
  neural networks with personalized pagerank for classification on graphs}}.
\newblock In \emph{International Conference on Learning Representations}. 2019.

\bibitem[\protect\citeauthoryear{Kloster and Gleich}{2014}]{Gleich2014}
K.~\textsc{Kloster} and D.~F. \textsc{Gleich}.
\newblock \emph{Heat kernel based community detection}.
\newblock In \emph{Proceedings of the 20th ACM SIGKDD international conference
  on Knowledge discovery and data mining}, pp. 1386--1395. 2014.

\bibitem[\protect\citeauthoryear{Koren}{2003}]{kohen}
Y.~\textsc{Koren}.
\newblock \emph{On spectral graph drawing}.
\newblock In \emph{Computing and Combinatorics}, pp. 496--508. 2003.

\bibitem[\protect\citeauthoryear{Lang}{2005}]{Lang-2005-weaknesses}
K.~\textsc{Lang}.
\newblock \emph{Fixing two weaknesses of the spectral method}.
\newblock Advances in Neural Information Processing Systems, 18, 2005.

\bibitem[\protect\citeauthoryear{Langville and Meyer}{2006}]{pprbook}
A.~N. \textsc{Langville} and C.~D. \textsc{Meyer}.
\newblock \emph{Google's PageRank and Beyond: The Science of Search Engine
  Rankings}, Princeton University Press, USA, 2006.

\bibitem[\protect\citeauthoryear{Levy and Goldberg}{2014}]{mf1}
O.~\textsc{Levy} and Y.~\textsc{Goldberg}.
\newblock \emph{Neural word embedding as implicit matrix factorization}.
\newblock Advances in neural information processing systems, 27, 2014.

\bibitem[\protect\citeauthoryear{Liu et~al.}{2021}]{meng}
M.~\textsc{Liu}, N.~\textsc{Veldt}, H.~\textsc{Song}, P.~\textsc{Li}, and D.~F.
  \textsc{Gleich}.
\newblock \href{http://dx.doi.org/10.1145/3442381.3449887}{\emph{Strongly local
  hypergraph diffusions for clustering and semi-supervised learning}}.
\newblock In \emph{{WWW} '21: The Web Conference 2021, Virtual Event /
  Ljubljana, Slovenia, April 19-23, 2021}, pp. 2092--2103. 2021.
\newblock \href {http://dx.doi.org/10.1145/3442381.3449887}
  {\normalcolor\path{doi:10.1145/3442381.3449887}}.

\bibitem[\protect\citeauthoryear{Mahoney et~al.}{2012}]{mahoney2012}
M.~W. \textsc{Mahoney}, L.~\textsc{Orecchia}, and N.~K. \textsc{Vishnoi}.
\newblock \emph{A local spectral method for graphs: With applications to
  improving graph partitions and exploring data graphs locally}.
\newblock The Journal of Machine Learning Research, 13~(1), pp. 2339--2365,
  2012.

\bibitem[\protect\citeauthoryear{Mikolov et~al.}{2013}]{sg}
T.~\textsc{Mikolov}, K.~\textsc{Chen}, G.~\textsc{Corrado}, and
  J.~\textsc{Dean}.
\newblock \href{http://arxiv.org/abs/1301.3781}{\emph{Efficient estimation of
  word representations in vector space}}.
\newblock In \emph{1st International Conference on Learning Representations,
  {ICLR} 2013, Scottsdale, Arizona, USA, May 2-4, 2013, Workshop Track
  Proceedings}. 2013.

\bibitem[\protect\citeauthoryear{Nguyen et~al.}{2017}]{QFA17}
Q.~\textsc{Nguyen}, F.~\textsc{Tudisco}, A.~\textsc{Gautier}, and
  M.~\textsc{Hein}.
\newblock \emph{An efficient multilinear optimization framework for hypergraph
  matching}.
\newblock IEEE transactions on pattern analysis and machine intelligence,
  39~(6), pp. 1054--1075, 2017.

\bibitem[\protect\citeauthoryear{Perozzi et~al.}{2014}]{deepwalk}
B.~\textsc{Perozzi}, R.~\textsc{Al-Rfou}, and S.~\textsc{Skiena}.
\newblock \href{http://dx.doi.org/10.1145/2623330.2623732}{\emph{Deepwalk:
  Online learning of social representations}}.
\newblock In \emph{Proceedings of the 20th ACM SIGKDD International Conference
  on Knowledge Discovery and Data Mining}, p. 701–710. 2014.
\newblock \href {http://dx.doi.org/10.1145/2623330.2623732}
  {\normalcolor\path{doi:10.1145/2623330.2623732}}.

\bibitem[\protect\citeauthoryear{Perron}{1907}]{per07}
O.~\textsc{Perron}.
\newblock \emph{Zur theorie der matrices}.
\newblock Mathematische Annalen, 64~(2), pp. 248--263, 1907.

\bibitem[\protect\citeauthoryear{Pothen
  et~al.}{1990}]{Pothen-1990-partitioning}
A.~\textsc{Pothen}, H.~D. \textsc{Simon}, and K.-P. \textsc{Liou}.
\newblock \href{http://dx.doi.org/10.1137/0611030}{\emph{Partitioning sparse
  matrices with eigenvectors of graphs}}.
\newblock SIAM J. Matrix Anal. Appl., 11, pp. 430--452, 1990.
\newblock \href {http://dx.doi.org/10.1137/0611030}
  {\normalcolor\path{doi:10.1137/0611030}}.

\bibitem[\protect\citeauthoryear{Qiu et~al.}{2018}]{netmf}
J.~\textsc{Qiu}, Y.~\textsc{Dong}, H.~\textsc{Ma}, J.~\textsc{Li},
  K.~\textsc{Wang}, and J.~\textsc{Tang}.
\newblock \href{http://dx.doi.org/10.1145/3159652.3159706}{\emph{Network
  embedding as matrix factorization: Unifying deepwalk, line, pte, and
  node2vec}}.
\newblock In \emph{Proceedings of the Eleventh ACM International Conference on
  Web Search and Data Mining}, p. 459–467. 2018.
\newblock \href {http://dx.doi.org/10.1145/3159652.3159706}
  {\normalcolor\path{doi:10.1145/3159652.3159706}}.

\bibitem[\protect\citeauthoryear{Sahai et~al.}{2011}]{hearing2011}
T.~\textsc{Sahai}, A.~\textsc{Speranzon}, and A.~\textsc{Banaszuk}.
\newblock \emph{Hearing the clusters in a graph: A distributed algorithm}.
\newblock 2011.
\newblock \href {http://arxiv.org/abs/0911.4729}
  {\normalcolor\path{arXiv:0911.4729}}.

\bibitem[\protect\citeauthoryear{Serra-Capizzano}{2005}]{Serra-Capizzano-JCF}
S.~\textsc{Serra-Capizzano}.
\newblock \href{http://dx.doi.org/10.1137/S0895479804441407}{\emph{Jordan
  canonical form of the google matrix: A potential contribution to the pagerank
  computation}}.
\newblock SIAM Journal on Matrix Analysis and Applications, 27~(2), pp.
  305--312, 2005.
\newblock \href
  {http://arxiv.org/abs/https://doi.org/10.1137/S0895479804441407}
  {\normalcolor\path{arXiv:https://doi.org/10.1137/S0895479804441407}}, \href
  {http://dx.doi.org/10.1137/S0895479804441407}
  {\normalcolor\path{doi:10.1137/S0895479804441407}}.

\bibitem[\protect\citeauthoryear{Shi and Malik}{2000}]{MalikShi}
J.~\textsc{Shi} and J.~\textsc{Malik}.
\newblock \href{http://dx.doi.org/10.1109/34.868688}{\emph{Normalized cuts and
  image segmentation}}.
\newblock IEEE Transactions on Pattern Analysis and Machine Intelligence,
  22~(8), pp. 888--905, 2000.
\newblock \href {http://dx.doi.org/10.1109/34.868688}
  {\normalcolor\path{doi:10.1109/34.868688}}.

\bibitem[\protect\citeauthoryear{Stehl{\'{e}} et~al.}{2011}]{contact2}
J.~\textsc{Stehl{\'{e}}}, N.~\textsc{Voirin}, A.~\textsc{Barrat},
  C.~\textsc{Cattuto}, L.~\textsc{Isella}, J.-F. \textsc{Pinton},
  M.~\textsc{Quaggiotto}, W.~V. \textsc{den Broeck}, C.~\textsc{R{\'{e}}gis},
  B.~\textsc{Lina}, and P.~\textsc{Vanhems}.
\newblock
  \href{http://dx.doi.org/10.1371/journal.pone.0023176}{\emph{High-resolution
  measurements of face-to-face contact patterns in a primary school}}.
\newblock {PLoS} {ONE}, 6~(8), p. e23176, 2011.
\newblock \href {http://dx.doi.org/10.1371/journal.pone.0023176}
  {\normalcolor\path{doi:10.1371/journal.pone.0023176}}.

\bibitem[\protect\citeauthoryear{Takai et~al.}{2020}]{hgpr}
Y.~\textsc{Takai}, A.~\textsc{Miyauchi}, M.~\textsc{Ikeda}, and
  Y.~\textsc{Yoshida}.
\newblock \href{https://doi.org/10.1145/3394486.3403248}{\emph{Hypergraph
  Clustering Based on PageRank}}, p. 1970–1978.
\newblock Association for Computing Machinery, New York, NY, USA, 2020.

\bibitem[\protect\citeauthoryear{Tang et~al.}{2015{\natexlab{a}}}]{pte}
J.~\textsc{Tang}, M.~\textsc{Qu}, and Q.~\textsc{Mei}.
\newblock \href{http://dx.doi.org/10.1145/2783258.2783307}{\emph{Pte:
  Predictive text embedding through large-scale heterogeneous text networks}}.
\newblock In \emph{Proceedings of the 21th ACM SIGKDD International Conference
  on Knowledge Discovery and Data Mining}, p. 1165–1174. 2015{\natexlab{a}}.
\newblock \href {http://dx.doi.org/10.1145/2783258.2783307}
  {\normalcolor\path{doi:10.1145/2783258.2783307}}.

\bibitem[\protect\citeauthoryear{Tang et~al.}{2015{\natexlab{b}}}]{line}
J.~\textsc{Tang}, M.~\textsc{Qu}, M.~\textsc{Wang}, M.~\textsc{Zhang},
  J.~\textsc{Yan}, and Q.~\textsc{Mei}.
\newblock \href{http://dx.doi.org/10.1145/2736277.2741093}{\emph{Line:
  Large-scale information network embedding}}.
\newblock In \emph{Proceedings of the 24th International Conference on World
  Wide Web}, p. 1067–1077. 2015{\natexlab{b}}.
\newblock \href {http://dx.doi.org/10.1145/2736277.2741093}
  {\normalcolor\path{doi:10.1145/2736277.2741093}}.

\bibitem[\protect\citeauthoryear{Tong et~al.}{2006}]{tong}
H.~\textsc{Tong}, C.~\textsc{Faloutsos}, and J.-Y. \textsc{Pan}.
\newblock \emph{Fast random walk with restart and its applications}.
\newblock In \emph{Sixth international conference on data mining (ICDM'06)},
  pp. 613--622. 2006.

\bibitem[\protect\citeauthoryear{Tsitsulin et~al.}{2018}]{hearing2018}
A.~\textsc{Tsitsulin}, D.~\textsc{Mottin}, P.~\textsc{Karras},
  A.~\textsc{Bronstein}, and E.~\textsc{Müller}.
\newblock \href{http://dx.doi.org/10.1145/3219819.3219991}{\emph{Netlsd:
  Hearing the shape of a graph}}.
\newblock Proceedings of the 24th ACM SIGKDD International Conference on
  Knowledge Discovery and Data Mining, 2018.
\newblock \href {http://dx.doi.org/10.1145/3219819.3219991}
  {\normalcolor\path{doi:10.1145/3219819.3219991}}.

\bibitem[\protect\citeauthoryear{Tudisco et~al.}{2021{\natexlab{a}}}]{FAK21}
F.~\textsc{Tudisco}, A.~R. \textsc{Benson}, and K.~\textsc{Prokopchik}.
\newblock \href{http://dx.doi.org/10.1145/3442381.3450035}{\emph{Nonlinear
  higher-order label spreading}}.
\newblock In \emph{Proceedings of the Web Conference 2021}, pp. 2402--2413.
  2021{\natexlab{a}}.
\newblock \href {http://dx.doi.org/10.1145/3442381.3450035}
  {\normalcolor\path{doi:10.1145/3442381.3450035}}.

\bibitem[\protect\citeauthoryear{Tudisco and Higham}{2021}]{FH21}
F.~\textsc{Tudisco} and D.~J. \textsc{Higham}.
\newblock \emph{Node and edge eigenvector centrality for hypergraphs}.
\newblock arXiv:2101.06215, 2021.

\bibitem[\protect\citeauthoryear{Tudisco
  et~al.}{2021{\natexlab{b}}}]{nonlin_dif_hg}
F.~\textsc{Tudisco}, K.~\textsc{Prokopchik}, and A.~R. \textsc{Benson}.
\newblock \emph{A nonlinear diffusion method for semi-supervised learning on
  hypergraphs}.
\newblock 2021{\natexlab{b}}.
\newblock \href {http://arxiv.org/abs/2103.14867}
  {\normalcolor\path{arXiv:2103.14867}}.

\bibitem[\protect\citeauthoryear{Tudisco et~al.}{2021{\natexlab{c}}}]{FKA21}
---{}---{}---.
\newblock \emph{A nonlinear diffusion method for semi-supervised learning on
  hypergraphs}.
\newblock arXiv:2103.14867, 2021{\natexlab{c}}.

\bibitem[\protect\citeauthoryear{Veldt
  et~al.}{2020{\natexlab{a}}}]{mathoverflow}
N.~\textsc{Veldt}, A.~R. \textsc{Benson}, and J.~\textsc{Kleinberg}.
\newblock \emph{Minimizing localized ratio cut objectives in hypergraphs}.
\newblock In \emph{Proceedings of the 26th {ACM} {SIGKDD} International
  Conference on Knowledge Discovery and Data Mining}. 2020{\natexlab{a}}.

\bibitem[\protect\citeauthoryear{Veldt et~al.}{2020{\natexlab{b}}}]{veldt20}
---{}---{}---.
\newblock \emph{Minimizing localized ratio cut objectives in hypergraphs}.
\newblock In \emph{Proceedings of the 26th ACM SIGKDD International Conference
  on Knowledge Discovery \& Data Mining}, pp. 1708--1718. 2020{\natexlab{b}}.

\bibitem[\protect\citeauthoryear{Wang and Leskovec}{2020}]{unifying}
H.~\textsc{Wang} and J.~\textsc{Leskovec}.
\newblock \href{https://openreview.net/forum?id=rkgdYhVtvH}{\emph{Unifying
  graph convolutional neural networks and label propagation}}.
\newblock 2020.

\bibitem[\protect\citeauthoryear{Yadati et~al.}{2019}]{hypergcn}
N.~\textsc{Yadati}, M.~\textsc{Nimishakavi}, P.~\textsc{Yadav},
  V.~\textsc{Nitin}, A.~\textsc{Louis}, and P.~\textsc{Talukdar}.
\newblock \emph{Hypergcn: A new method of training graph convolutional networks
  on hypergraphs}.
\newblock 2019.
\newblock \href {http://arxiv.org/abs/1809.02589}
  {\normalcolor\path{arXiv:1809.02589}}.

\bibitem[\protect\citeauthoryear{Yang et~al.}{2020}]{hne}
R.~\textsc{Yang}, J.~\textsc{Shi}, X.~\textsc{Xiao}, Y.~\textsc{Yang}, and
  S.~S. \textsc{Bhowmick}.
\newblock \href{http://dx.doi.org/10.14778/3377369.3377376}{\emph{Homogeneous
  network embedding for massive graphs via reweighted personalized pagerank}}.
\newblock Proc. VLDB Endow., 13~(5), p. 670–683, 2020.
\newblock \href {http://dx.doi.org/10.14778/3377369.3377376}
  {\normalcolor\path{doi:10.14778/3377369.3377376}}.

\end{thebibliography}

\end{fullwidth}

\end{document}